\documentclass[pra,twocolumn,showpacs,amsmath,amssymb,superscriptaddress,floatfix,longbibliography,nofootinbib]{revtex4-2}
\usepackage{times}
\usepackage{helvet}
\usepackage{courier}
\usepackage{hyperref}
\hypersetup{
    colorlinks=true,   
    linkcolor=cyan,    
    citecolor=magenta, 
    filecolor=magenta, 
    urlcolor=cyan,     
    runcolor=cyan
}

\usepackage{graphicx}
\usepackage{natbib}

\usepackage{algorithm}
\usepackage{algorithmic}

\usepackage{newfloat}
\usepackage{listings}

\usepackage{amsmath}
\usepackage{amsfonts}
\usepackage{amssymb}
\usepackage{amsthm}
\usepackage{braket}
\usepackage[capitalize]{cleveref}
\usepackage{xcolor}
\usepackage{booktabs}
\usepackage{tikz}
\usepackage{mathtools}

\newtheorem{mytheorem}{Theorem}[section]
\newtheorem{mylemma}{Lemma}[section]
\newtheorem{mycorollary}{Corollary}[section]
\newtheorem{myproposition}{Proposition}[section]

\DeclareMathOperator{\Res}{Res}
\DeclareMathOperator{\Tr}{Tr}
\DeclareMathOperator{\diag}{diag}

\DeclareMathOperator{\rank}{rank}
\DeclareMathOperator{\real}{Re}
\DeclareMathOperator{\NN}{NN}
\DeclareMathOperator{\NNN}{NNN}
\newcommand{\floor}[1]{\left\lfloor #1 \right\rfloor}
\newcommand{\myvec}[1]{{\boldsymbol{#1}}}
\newcommand{\vlambda}{\myvec{\lambda}}
\newcommand{\vdlambda}{\myvec{\delta\lambda}}
\newcommand{\vecv}{\myvec{v}}
\newcommand{\veczero}{\myvec{0}}
\newcommand{\abs}[1]{\left|#1\right|}
\newcommand{\norm}[1]{\left\|#1\right\|}
\newcommand{\ketbra}[1]{|{#1}\>\!\<#1|}
\newcommand{\bk}[2]{\<{#1}|{#2}\>}
\def\>{\rangle}
\def\<{\langle}
\newcommand{\mcM}{\mathcal{M}}
\newcommand{\mcH}{\mathcal{H}}
\newcommand{\mcS}{\mathcal{S}}
\newcommand{\mcD}{\mathcal{D}}
\newcommand{\mcDtrain}{\mathcal{D}_{\textnormal{train}}}
\newcommand{\mcDtest}{\mathcal{D}_{\textnormal{test}}}
\DeclareMathOperator{\DKL}{D_{\textnormal{KL}}}

\newcommand{\bal}{\begin{aligned}}
\newcommand{\eal}{\end{aligned}}
\newcommand{\bes}{\begin{subequations}}
\newcommand{\ees}{\end{subequations}}
\newcommand{\bea}{\begin{align}}
\newcommand{\eea}{\end{align}}
\newcommand{\beq}{\begin{equation}}
\newcommand{\eeq}{\end{equation}}

\begin{document}

\title{Detecting Quantum and Classical Phase Transitions via Unsupervised Machine Learning of the Fisher Information Metric}
\author{Victor Kasatkin}
\affiliation{Center for Quantum Information Science \& Technology, University of Southern California}
\affiliation{Department of Electrical \& Computer Engineering, University of Southern California}
\author{Evgeny Mozgunov}
\affiliation{Center for Quantum Information Science \& Technology, University of Southern California}
\affiliation{Department of Electrical \& Computer Engineering, University of Southern California}
\author{Nicholas Ezzell}
\affiliation{Center for Quantum Information Science \& Technology, University of Southern California}
\affiliation{Information Sciences Institute, University of Southern California}
\affiliation{Department of Physics and Astronomy, University of Southern California}
\author{Daniel Lidar}
\affiliation{Center for Quantum Information Science \& Technology, University of Southern California}
\affiliation{Department of Electrical \& Computer Engineering, University of Southern California}
\affiliation{Department of Physics and Astronomy, University of Southern California}
\affiliation{Department of Chemistry, University of Southern California}

\begin{abstract}
The detection of quantum and classical phase transitions in the absence of an order parameter is possible using the Fisher information metric (FIM), also known as fidelity susceptibility. Here, we investigate an unsupervised machine learning (ML) task: estimating the FIM given limited samples from a multivariate probability distribution of measurements made throughout the phase diagram.
We utilize an unsupervised ML method called ClassiFIM (developed in a companion paper) to solve this task and demonstrate its empirical effectiveness in detecting both quantum and classical phase transitions using a variety of spin and fermionic models, for which we generate several publicly available datasets with accompanying ground-truth FIM.
We find that ClassiFIM reliably detects both topological (e.g., XXZ chain) and dynamical (e.g., metal-insulator transition in the Hubbard model) quantum phase transitions. We perform a detailed quantitative comparison with two prior unsupervised ML methods for detecting quantum phase transitions. We demonstrate that ClassiFIM is competitive with these prior methods in terms of appropriate accuracy metrics while requiring significantly less resource-intensive training data compared to the original formulation of the prior methods. In particular, ClassiFIM only requires classical (single-basis) measurements. As part of our methodology development, we prove several theorems connecting the classical and quantum fidelity susceptibilities through equalities or bounds. We also significantly expand the existence conditions of the fidelity susceptibility, e.g., by relaxing standard differentiability conditions. These results may be of independent interest to the mathematical physics community.
\end{abstract}

\maketitle

\section{Introduction}
\label{sec:intro}

The use of machine learning (ML) to study many-body quantum systems~\cite{Sachdev:book,Wen:2017aa,Weimer:2021aa} is an active and exciting area of research~\cite{carleo2017solving, beach2019qucumber, carleo2019netket, carrasquilla2020machine, tibaldi2023unsupervised}. An important application of ML in this setting is the study of phase transitions (PTs)~\cite{vanNieuwenburg2016LearningPT, carrasquilla2017machine, ch2017machine, Lidiak2020UnsupervisedML, uvarov2020machine, Kming2021UnsupervisedML, Guo2022LearningPT, Karsch2022AML, abram2022inferring, Huang:22, Caro:2022aa, Li2023MachineLP}. 

To informally describe a PT, suppose that a system is in state $\rho_A$ subject to external parameters $\vlambda = (\lambda_1,\lambda_2,\dots)$ (e.g., magnetic and electric fields). A PT occurs when by adjusting $\vlambda \rightarrow \vlambda'$, the new state $\rho_B$ is radically different. If the transition is temperature-driven (e.g., water to ice), the PT is usually classified as classical (CPT), whereas if it occurs at $T=0$ due to quantum fluctuations driven by another parameter (e.g., the superfluid to Mott insulator transition in ultracold atoms~\cite{Greiner:2002aa}), it is classified as quantum (QPT)~\cite{Sachdev:book}. For a brief but more detailed discussion of QPTs, see \cref{as:background}.

In most cases, different phases are distinguishable by some local measurement. Specifically, a nonzero value of an order parameter corresponds to one of the (approximately) degenerate ground states in the ordered phase. If no local order parameter distinguishes ordered ground states in this way, the order is considered topological and involves topological QPTs (TQPTs)~\cite{Wen:2017aa,hamma2008entanglement}.  
Measurements of the order parameter (nonlocal in the case of TQPTs) as a function of $\vlambda$ allow one to locate a phase boundary: the values of $\vlambda$ where the order parameter behaves in a non-analytic way in the thermodynamic limit. A phase diagram is, by definition, a plot of regions in $\vlambda$-space (phases) and the boundaries between them. The phase boundaries are smeared for finite-size systems.

Computing a quantum phase diagram or phase boundary is generally numerically intractable, except for simple models with a high degree of symmetry. For this reason, most uses of ML for QPTs are concerned with efficiently determining such boundaries for a fixed system size under different assumptions. Early approaches involved supervised learning on data labeled with its phase using a binary classifier~\cite{carrasquilla2020machine}. More recent approaches focus on the unsupervised setting, where no labels or order parameters are known. Early progress in this setting was made by showing that an unsupervised QPT task could be solved with a series of mislabeled binary classifiers~\cite{vanNieuwenburg2016LearningPT}, but this approach is susceptible to certain practical obstacles we discuss below.
Substantial progress was made in~\cite{Huang:22} using classical shadow tomography~\cite{Huang:2020wo}, the current state-of-the-art method.

However, since the models, input data structures, and the outputs of some of the previous work on ML for QPTs such as Refs.~\cite{vanNieuwenburg2016LearningPT, carrasquilla2017machine, ch2017machine, Huang:22} are different, they are not directly comparable. 
More generally, there is currently no convention regarding goals, input data assumptions, datasets, or quantitative metrics in the area of ML for QPTs. This lack of standardization makes it difficult to quantitatively compare methods even indirectly, an issue we address in this work. We note that this standardization issue was already addressed to some extent by Ref.~\cite{arnold2023machine} by adjusting other methods, including that of Ref.~\cite{vanNieuwenburg2016LearningPT}, to output a lower bound on the Fisher information metric (FIM).

Another difficulty is identifying order parameters in the first place: except in simple cases, it may be difficult to identify a physical observable whose expectation value plays the role of an order parameter. 
For this reason, much attention has focused on alternatives to order parameters for finding transitions detectable via the state $\rho$, such as the Fisher information metric (FIM) or the equivalent fidelity susceptibility, which are known to exhibit sharp peaks at locations of classical or quantum phase transitions~\cite{venuti2007quantum,zanardi2007bures,zanardi2007mixed,zanardi2007information,GU:2010vv} (with certain pathological exceptions~\cite{cincio2019universal}). As also noted by Ref.~\cite{arnold2023machine}, a plot of the FIM as a function of $\vlambda$ can thus be viewed as an ``order-parameter-free phase diagram", which we call an FIM phase diagram. Thus, an accurate estimate of the FIM can be viewed as the ground truth for the detection of phase transitions (either classical or quantum) without the knowledge of an order parameter. The fidelity susceptibility can be defined both for classical probability distributions and for quantum states $\rho$, in which case it is called the quantum fidelity susceptibility. 

The first main contribution we make here concerns the relationship between the FIM extracted from this probability distribution directly, and the quantum fidelity susceptibility of the underlying states. We work in the limit where the number of possible measurement outcomes is much smaller than what is needed for full state tomography. In fact, for most physical models of interest, it is enough to measure all the qubits in the computational basis (the simultaneous eigenstates of all the Pauli $\sigma^z$ operators), resulting in a probability distribution of bitstrings with one bit per qubit. This probability distribution contains less information than the full quantum state $\rho$, and yet the classical and quantum fidelities of the two turn out to be related under very general assumptions.

It may be surprising that measurements in a fixed basis can be used to detect quantum phase transitions. We provide a rigorous mathematical foundation for this result by identifying conditions under which the quantum and classical fidelity susceptibilities are either identical or agree in expectation. We also find sharp bounds relating the quantum and classical susceptibilities. Moreover, generalizing earlier work~\cite{venuti2007quantum,zanardi2007bures,zanardi2007mixed,zanardi2007information,GU:2010vv}, we relax the commonly assumed differentiability conditions under which the fidelity susceptibility is well-defined.

The second main contribution we make in this work relies on a method developed in a companion paper \cite{ClassiFIM-ML}, which introduces a method called `ClassiFIM' for solving the task of estimation of the FIM. Here, we discuss the motivation behind this task and its relationship with unsupervised estimation of phase transitions. We present several datasets containing limited samples of measurement outcomes of various systems, together with the ground truth FIM of the corresponding probability distribution. We demonstrate using various example systems we analyze, that phase transitions have clear FIM signatures that are also detectable using ClassiFIM. 
For example, we use ClassiFIM to produce an accurate FIM phase diagram for a $300$-qubit XXZ model that is known to exhibit a TQPT~\cite{Elben}. 
ClassiFIM relies on estimating the \emph{classical} fidelity susceptibility, i.e., as mentioned above, it uses measurements in a \emph{fixed} basis, e.g., the computational basis. One major benefit of only needing fixed basis measurements is that ClassiFIM is straightforwardly compatible with analog quantum devices such as quantum annealers, as well as with various Monte Carlo methods. We show that ClassiFIM applies straightforwardly to classical phase transitions as well, and performs comparably or better than previous ML methods such as those proposed in Refs.~\cite{vanNieuwenburg2016LearningPT,Huang:22}.
Specifically, in addition to introducing machine learning methods, Refs.~\cite{vanNieuwenburg2016LearningPT,Huang:22} make suggestions about the datasets to which to apply these methods: \cite{vanNieuwenburg2016LearningPT} suggests using entanglement spectra, and \cite{Huang:22} presents the machine learning method in the context of shadow tomography data. Such datasets are much more resource-intensive to obtain than those discussed in this work, which contain bitstrings measured in the computational basis. To ensure a fair comparison, here we provide the methods of Refs.~\cite{vanNieuwenburg2016LearningPT,Huang:22} with the same datasets as the ones we generated for ClassiFIM, and moreover, we impose the same bounds on computational resources in all cases.

The structure of this paper is the following. In \cref{sec:FIM-FS} we first briefly review the FIM and its connection to the fidelity susceptibility. We then present our main theoretical results: we prove the equivalence of the classical and quantum fidelity susceptibilities under mild conditions and present additional formal results that generalize the known set of conditions under which the quantum fidelity susceptibility (for both pure and mixed states) is well-defined. In \cref{sec:classifim}, we start by defining the ML task of FIM-estimation (\cref{sec:task}) and describe its applicability to QPTs (\cref{ss:motivation}). Then, in \cref{ss:classifim}, we outline the ClassiFIM method.
\Cref{sec:validation-metrics-and-datasets} 
presents datasets we have generated to test and validate algorithms that attempt to solve the task of estimating the FIM. 
\Cref{sec:numerical-experiments} presents our main applied results: we use ClassiFIM for the ML task of detecting both classical and quantum phase transitions. To do so we use the datasets defined in the previous section.
We also demonstrate that ClassiFIM is competitive with earlier approaches~\cite{vanNieuwenburg2016LearningPT,Huang:22} in terms of the accuracy it achieves.
We provide a summary and outlook in \cref{as:limits-and-future}. The Appendices provide many supporting details, including more background and a literature review, theorem proofs, numerical implementation details, and additional simulation data.

Finally, we note that this work and the companion paper \cite{ClassiFIM-ML} are complementary, written for different audiences, and can be read independently. This work is aimed primarily at physicists, while the companion paper is aimed primarily at computer scientists.

\section{Fisher Information Metric and Fidelity Susceptibility}
\label{sec:FIM-FS}

We start with a brief review of the FIM and then define the classical and quantum fidelity susceptibilities. The classical fidelity susceptibility relates to probability distributions obtained when access to the quantum state is restricted, or the state is simulated by quantum Monte Carlo techniques, while the quantum one can be obtained via complete access to the quantum state. We present our first result: conditions under which the classical and quantum fidelity susceptibilities are equivalent (\cref{th:comput.basis} and \cref{th:avggc}). We then present a number of results relaxing the differentiability conditions under which the fidelity susceptibility (classical or quantum) is well-defined.

\subsection{FIM}
\label{subsec:fim}

The FIM is a formal way to measure the rate of change of a probability distribution. We briefly review some concepts needed for the rest of our discussion; see \cref{as:background} for additional pertinent facts and notation, and \cref{as:fim} for a discussion of the conditions needed for the FIM to be well defined.

Consider a manifold $\mcM$ and a family of probability distributions $P_{\vlambda}(x)$ depending on a continuous parameter $\vlambda \in \mcM$. Throughout this work, for simplicity, we
refer to $(\mcM, P)$ as a \emph{statistical manifold},
although the formal definition of a statistical manifold requires additional
regularity conditions~\cite{Schervish:1995aa}.
In local coordinates, one can write $\vlambda = (\lambda_1,\cdots,\lambda_N)$.
The score function for the $\mu$-th parameter $\lambda_{\mu}$
is the derivative of the log-likelihood and is given by
\begin{equation}
\label{eq:sff}
s_{\mu}(x;\vlambda) = \partial_{\lambda_{\mu}} \log P_{\vlambda}(x) \ .
\end{equation}

The \emph{Fisher Information Metric} $g$ is the variance of the score function or, equivalently (since by normalization $\mathbb{E}[s_{\mu}(X;\vlambda)]=\int dx\, P_{\vlambda}(x)s_{\mu}(x;\vlambda)=0$), the expectation value of the product of pairs of score functions. I.e., the components of the FIM are:
\begin{equation}
\label{eq:fim}
g_{\mu\nu}(\vlambda) = \mathbb{E}[s_{\mu}(X;\vlambda)s_{\nu}(X;\vlambda)] .
\end{equation}

Chencov's theorem implies that $g(\vlambda)$ is (up to rescaling) the unique Riemannian metric on the space of probability distributions satisfying a monotonicity condition and invariant under taking sufficient statistics \cite{chencov1972statistical,amari2016information}.
That is, the FIM is the natural way to measure the speed with which the probability distribution changes with respect to the parameters $\vlambda$. In fact, the FIM (up to a scaling factor) can be written as the first non-trivial term in the Taylor series expansion in the parameter differences $\vdlambda$ of multiple divergence measures,
including the Kullback-Leibler divergence and the Jensen-Shannon divergence.

The concept of \emph{fidelity susceptibility} plays a special role in this work:
quantum fidelity susceptibility is known as a universal (i.e., order-parameter-independent)
method for detecting phase transitions
\cite{venuti2007quantum,zanardi2007mixed,zanardi2007information,zanardi2007bures,GU:2010vv}.
The classical fidelity susceptibility, denoted $\chi_{F_c}$, is related to the FIM by a factor of 4 \cite{GU:2010vv}:
\begin{equation}
    g = 4 \chi_{F_c}.
    \label{eq:g-chi}
\end{equation}
The same factor of $4$ relates the quantum FIM and the quantum fidelity susceptibility, $\chi_F$. We use the FIM and fidelity susceptibility interchangeably in the remainder of this work.

\subsection{Formal equivalence between the classical and quantum fidelity susceptibilities}
\label{s:fidelity.susceptibility}

The FIM-estimation task we describe in \cref{sec:task} below estimates the classical FIM. Yet, as we will show, it successfully captures properties of the quantum fidelity susceptibility. This seemingly counterintuitive result has a firm theoretical foundation captured by the two theorems we present in this section. Before presenting the theorems, we establish the necessary terminology.

The quantum systems we discuss in this work occupy a complex Hilbert space
$\mathcal{H}$ with a finite orthonormal basis $\{\ket{x}\}_{x\in\mathcal{S}}$
indexed by bitstrings $x \in \mathcal{S} = \{0,1\}^n$, i.e.,
$\mathcal{H} \simeq \mathbb{C}^d$, where $d=\abs{\mathcal{S}} = 2^n$. When measuring the state $\ket{\psi}\in\mathcal{H}$ we obtain a bitstring $x$ with probability
$p_x = \abs{\bk{x}{\psi}}^2$.
In this case, we say that the probability distribution $p$ corresponds to the state $\ket{\psi}$.

Consider a state $\ket{\psi(\vlambda)}$
depending on a parameter $\vlambda \in \mcM \subset \mathbb{R}^m$, where $\mcM$ is the parameter manifold. The fidelity between two pure states $\ket{\phi}$ and $\ket{\psi}$ is $F(\ket{\phi}, \!\ket{\psi}) = \abs{\bk{\phi}{\psi}}$, and we denote the quantum fidelity as $F(\vlambda, \vlambda') = F(\ket{\psi(\vlambda)}, \ket{\psi(\vlambda')})$.
Let $p(\vlambda)$ and $q(\vlambda)$ denote the probability distributions obtained from measuring $\ket{\psi(\vlambda)}$ and $\ket{\phi(\vlambda)}$ in the same orthonormal basis. The fidelity between two discrete probability distributions $p$ and $q$ is defined as
\begin{equation}
  \label{eq:Fc.def0}
  F_c(p, q) = \sum_x \sqrt{p_x q_x},
\end{equation}
and we likewise denote the classical fidelity as
$F_c(\vlambda, \vlambda') = F_c(p(\vlambda), p(\vlambda'))$.
The quantum fidelity susceptibility $\chi_{F}(\vlambda)$ is the
coefficient in the first nontrivial term in the Taylor expansion of
$F(\vlambda, \vlambda')$:
\begin{align}
\label{eq:def.chi.F}
  & F\left(\ket{\psi(\vlambda)}, \ket{\psi(\vlambda')}\right) =\\
  & \ \ 1 - \frac{1}{2} \sum_{\mu,\nu} (\vlambda' - \vlambda)_{\mu} (\vlambda' - \vlambda)_{\nu} \left(\chi_{F}(\vlambda)\right)_{\mu\nu} + o(\abs{\vlambda' - \vlambda}^2). \notag
\end{align}
The classical fidelity susceptibility $\chi_{F_c}(\vlambda)$
is defined identically by expanding $F_c(\vlambda, \vlambda')$.
By performing such an expansion of \cref{eq:Fc.def0}
(or, more generally, of the Bhattacharyya distance; see \cref{as:fim})
and comparing the result with \cref{eq:def.chi.F}, one can check
that \cref{eq:g-chi} holds under certain regularity conditions.
As we discuss in \cref{as:fim}, when these regularity conditions do not hold, we define (for the purposes of this paper) the FIM $g$ using \cref{eq:g-chi}
instead of relying on the score function definition above.
We define the quantum FIM as $4\chi_F$. Both $g$ and $\chi_F$ can be interpreted as quadratic forms
on the tangent space of $\mcM$ at $\vlambda$, i.e., as metric tensors:
\begin{equation}
\chi_{F}(\vlambda; \vecv) = \sum_{\mu,\nu} \left(\chi_{F}(\vlambda)\right)_{\mu\nu} v_{\mu} v_{\nu}.
\end{equation}
Again, an identical expression holds for $\chi_{F_c}(\vlambda)$.
The classical fidelity upper bounds its quantum counterpart, while the classical fidelity susceptibility lower bounds its quantum counterpart~\cite{gu2010fidelity}: 
\bes
\begin{align}
&0\leq F(\vlambda, \vlambda') \leq F_c(\vlambda, \vlambda') \leq 1\\
&0 \leq \chi_{F_c}(\vlambda) \leq \chi_{F}(\vlambda) .
\end{align}
\ees

Since a quantum phase transition is often characterized by a drastic change in a (ground) state $\ket{\psi(\vlambda)}$ as $\vlambda$ is varied, the locations of the peaks of
$\chi_{F}(\vlambda)$ are of interest as indicators of the locations of
phase transitions. However, if --- as we assume in the QPT datasets throughout this work --- one only has access to samples from
$p(\vlambda)$ and not to other characteristics of $\ket{\psi(\vlambda)}$,
then one could only hope to estimate $\chi_{F_c}(\vlambda)$: indeed, for any
such distribution one could construct
$\ket{\phi(\vlambda)} = \sum_{x\in\mathcal{S}} \sqrt{p_x(\vlambda)} \ket{x}$
for which $\chi_{F}(\vlambda) = \chi_{F_c}(\vlambda)$, thus
the information about whether
$\chi_{F}(\vlambda) > \chi_{F_c}(\vlambda)$
for the original $\ket{\psi(\vlambda)}$ is lost.

Let $H(\vlambda)$ be a family of Hamiltonians and let $\ket{\psi_0(\vlambda)}$ denote its ground states.
We now state two theorems ensuring that in many important
cases, $\chi_{F_c}$ exhibits peaks at or close to the peaks of $\chi_{F}$.

\begin{mytheorem}
  \label{th:comput.basis}
  Let $H(\vlambda)$ be a continuously differentiable Hamiltonian family
  with a non-degenerate ground state and real matrix elements in the
  $\{\ket{x}\}_{x\in\mathcal{S}}$ basis, where $x$ is a bitstring in $\mathcal{S} = \{0,1\}^n$.  Then
  $\chi_{F}(\vlambda) = \chi_{F_c}(\vlambda)$ for its ground state.
\end{mytheorem}

\begin{proof}[Proof sketch]
Locally (near each $\vlambda_0 \in \mcM$) we can pick
  $\ket{\psi_0(\vlambda)}$ to have real
amplitudes. By evaluating
the definitions of $\chi_F$ and $\chi_{F_c}$ we conclude that they are equal.
\end{proof}

\begin{mytheorem}
  \label{th:avggc}
  Suppose $\ket{\psi(\vlambda)}$ is continuously differentiable, $\chi_{F}$ is its
  fidelity susceptibility, $p'(\vlambda)$ is the associated probability distribution, where the measurement
  is performed with respect to a Haar-random orthonormal basis $\{\ket{x'}\}_{x'\in \mathcal{S}'}$,
and $\chi_{F_c}$ is the classical fidelity susceptibility of $p'$. Then
  \begin{equation}
    \label{eq:Emeas}
    \mathbb{E}_{\mathcal{S}'} \left(\chi_{F_c}(\vlambda)\right) = \chi_{F}(\vlambda)/2.
  \end{equation}
\end{mytheorem}
Here $\mathbb{E}_{\mathcal{S}'}$ denotes the expectation value with respect to $\mathcal{S}'$. \cref{th:avggc} is proven by deriving explicit formulas for the two sides of \cref{eq:Emeas} in terms of $\ket{\psi(\vlambda)}$ and its derivatives.

Complete theorem statements, together with their proofs, are given in \cref{ass:two.theorems}.

The significance of these theorems is that it is both typical (\cref{th:avggc}) and relevant (\cref{th:comput.basis}) to use the classical fidelity susceptibility of computational basis measurements as a proxy for the quantum fidelity susceptibility of the corresponding states.
In many commonly studied quantum models the Hamiltonians can be represented with real matrix elements, including a variety of simple spin Hamiltonians. An important exception arises when particles couple to a magnetic field, which introduces imaginary matrix elements into an otherwise real-valued Hamiltonian. 

We note for later reference that the conditions of \cref{th:comput.basis} are mostly satisfied in the generation of all the QPT datasets in this work (see \cref{subsec:data-sources}): the Hamiltonians we use have only real matrix elements and are differentiable with respect to parameters.
We used symmetries and penalty terms in the Hamiltonian to avoid degeneracy whenever feasible.
However, some degeneracy might remain: the Theorem's conditions are satisfied for the FIL24, Hubbard12, Kitaev20, and XXZ300 datasets (defined in \cref{subsec:data-sources}), except for the locations of level crossings and small regions where the ground state is effectively degenerate.
However, this has only a minimal impact on the performance of ClassiFIM, as we illustrate in \cref{ss:level.crossing} and the phase diagrams in \cref{sec:numerical-experiments}.

In the remainder of this section, we provide additional technical results on the quantum fidelity susceptibility.
The reader who is interested in the results for various phase diagrams may safely skip ahead to \cref{sec:classifim}.

\subsection{Generalizations of the quantum fidelity susceptibility}
\label{ass:explicit.chi}

It is clear from \cref{eq:def.chi.F} that the fidelity susceptibility
can be expressed as the second derivative of the fidelity between pairs of pure states.
We now show that the second derivative of $\ket{\psi(\vlambda)}$ does not need to exist
in order for the fidelity susceptibility to be well defined:

\begin{mytheorem}
  \label{lm:explicit.chi-1}
  Assume that $\ket{\psi(\vlambda')}$ is a state defined in a
  neighborhood of
  $\vlambda' = \vlambda \in \mcM$ and differentiable at
  $\vlambda' = \vlambda$.  Then the fidelity susceptibility
  is well-defined at $\vlambda$ and is given by
  \begin{align}
    \label{eq:explicit.chi.pure}
    (\chi_F(\vlambda))_{\mu\nu} &= \real\bigl(\braket{\partial_\mu \psi(\vlambda)
      | \partial_\nu \psi(\vlambda)}\bigr)\notag \\
    &\quad - \braket{\partial_\mu \psi(\vlambda) | \psi(\vlambda)}
      \braket{\psi(\vlambda) | \partial_\nu \psi(\vlambda)}.
  \end{align}
\end{mytheorem}

\cref{lm:explicit.chi-1} is essentially Ref.~\cite[Eq.~(3)]{zanardi2007information} along with a specification of the exact conditions that $\ket{\psi(\lambda)}$ must satisfy.
As mentioned in Ref.~\cite{zanardi2007information}, the proof is done by
Taylor expansion of the pure state fidelity $F(\ket{\phi}, \!\ket{\psi}) = \abs{\bk{\phi}{\psi}}$
and usage of state normalization.
Since we do not require the second derivative to exist, we have to perform the
expansion more carefully than Ref.~\cite{zanardi2007information}. Our proof is given in \cref{app:proof-lm:explicit.chi-1}.

The classical fidelity $F_c$ generalizes for pairs of mixed quantum states $\rho$ and $\sigma$ to the Uhlmann fidelity $F(\rho, \sigma) =\|\sqrt{\rho}\sqrt{\sigma}\|_1$ (see \cref{as:background}). 
The statement that the fidelity susceptibility
can be expressed as a second derivative also holds when considering two mixed states $\rho(\vlambda)$ and $\rho(\vlambda')$, for which the fidelity susceptibility is given by 
\begin{equation}
\label{eq:def.chi.dFmixed}
  \left(\chi_{F}(\vlambda)\right)_{\mu\nu} =
    -\frac{\partial^2 F\left(\rho(\vlambda), \rho(\vlambda')\right)}{%
    \partial \lambda'_{\mu} \partial \lambda'_{\nu}} \Bigg|_{\vlambda' = \vlambda}.
\end{equation}
For mixed states $\chi_{F}$ is also known as the Bures metric~\cite{bures-extension-1969}.

Despite the second derivative appearing in \cref{eq:def.chi.dFmixed}, just as in the pure state case, we can again show that it is often sufficient to assume that just the first derivative of the state or probability distribution is defined in order for the fidelity susceptibility to be well-defined. This further generalizes previous work, which made stricter assumptions
regarding rank and differentiability than the ones we make below. Namely, in \cref{lm:explicit.chi-2} below,
we do not make any assumptions concerning the rank of $\rho(\vlambda')$ and only assume the existence of the first derivative of its components
and the second derivative of a trace of a restriction of $\rho(\lambda')$ onto $\ker(\rho(\vlambda))$.

We first introduce the notation used in \cref{lm:explicit.chi-2}. If $\rho\colon \mcH \to \mcH$
is a positive-semidefinite operator on a finite-dimensional Hilbert space $\mcH$,
then $\mcH$ can be decomposed
into the direct sum of two orthogonal components
$\mcH = \mcH_{0} \oplus \mcH_{+}$ corresponding to the null and positive eigenspaces of $\rho$
respectively: $\mcH_{0} = \ker(\rho)$ and $\mcH_{+} = \rho \mcH$.
We then can define two surjective maps
$P_{0}\colon \mcH \twoheadrightarrow \mcH_0$
and $P_{+} \colon \mcH \twoheadrightarrow \mcH_{+}$
projecting a vector $\ket{\psi}\in\mcH$ onto the two components.
Their adjoints are then injections $P_0^\dagger: \mcH_0 \hookrightarrow \mcH$
and $\mcH_{+} \hookrightarrow \mcH$. Note that with this notation
$P_{+} \rho P_{+}^\dagger$ is an invertible Hermitian operator on $\mcH_{+}$
(i.e., from $\mcH_{+}$ to $\mcH_{+}$)
even if $\rho$ is not invertible [i.e., even if $\dim(\mcH_0) > 0$]. We also use the notation $\partial_\mu = \frac{\partial}{\partial \lambda'_\mu}$
and $\partial_\mu p_z(\vlambda) = \left.\frac{\partial p_z(\vlambda')}{\partial \lambda'_\mu}\right|_{\vlambda' = \vlambda}$.

\begin{mytheorem}
  \label{lm:explicit.chi-2}
      Let $\rho(\vlambda')$ be a density matrix defined in the
      neighborhood of $\vlambda' = \vlambda \in \mcM$.
      Let $P_0: \mcH \to \ker\rho(\vlambda)$,
      $P_{+}: \mcH \to \rho(\vlambda)(\mathcal{H})$ be the maps
      projecting $\mcH$
      to the null and positive eigenspaces of $\rho(\vlambda)$ respectively, and let $\rho_{+}(\vlambda') = P_{+}\rho(\vlambda') P_{+}^{\dagger}$.
      Assume $\rho(\vlambda')$ is differentiable at $\vlambda' = \vlambda$,
      and $\Tr(P_0 \rho(\vlambda') P_0^{\dagger})$
      is twice differentiable at $\vlambda' = \vlambda$.
       Then, in the basis where
      $\rho_{+}(\vlambda) = \diag(\xi_1,\dots,\xi_{n_{+}})$ for some $n_{+} \leq \dim\mcH$,
      we have
      \begin{align}
        \label{eq:explicit.chi.mixed}
        (\chi_{F}(\vlambda))_{\mu\nu}
        &= \sum_{j,k=1}^{n_{+}} \frac{
            \real\left(
              (\partial_\mu\rho_{+}(\vlambda))_{jk}
              (\partial_\nu\rho_{+}(\vlambda))_{kj}
            \right)
          }{2(\xi_j + \xi_k)}\notag \\
        &\quad + \frac12 \partial_\mu \partial_\nu \Tr(P_0 \rho(\vlambda) P_0^{\dagger}).
      \end{align}
 \end{mytheorem}

The proof is presented in \cref{app:proof-lm:explicit.chi-2}.

Before we proceed, we make a few remarks concerning \cref{lm:explicit.chi-2}.

\begin{enumerate}

\item \cref{lm:explicit.chi-2} includes the second term in \cref{eq:explicit.chi.mixed},
which generalizes the fidelity susceptibility by accounting for non-full rank mixed states (it vanishes for full-rank states such as Gibbs (mixed) states at finite temperature~\cite{gu2010fidelity}).

\item If one applies \cref{lm:explicit.chi-2} to a rank 1 state $\rho$, i.e., a pure state $\rho(\vlambda') = \ketbra{\psi(\vlambda')}$, then the second term in \cref{eq:explicit.chi.mixed} vanishes. In this case the r.h.s. of \cref{eq:explicit.chi.mixed} is equal to both \cref{eq:explicit.chi.pure} and to
\begin{equation}
  \label{eq:explicit.chi.mixed.v2}
  (\chi_{F}(\vlambda))_{\mu\nu}
  = \sum_{j=1}^{\dim{\mcH}} \sum_{k=1}^{k_{\textnormal{max}}(j)} \frac{
      \real\left(
        (\partial_\mu\rho(\vlambda))_{jk}
        (\partial_\nu\rho(\vlambda))_{kj}
      \right)
    }{2(\xi_j + \xi_k)},
\end{equation}
where
\begin{equation}
  \label{eq:k.max}
  k_{\textnormal{max}}(j) = \begin{cases}
    \dim{\mcH} & \textnormal{if } j \leq n_{+}, \\
    n_{+} & \textnormal{if } j > n_{+},
  \end{cases}
\end{equation}
and $\rho(\vlambda')$ is written in a basis where
$\rho(\vlambda) = \diag(\xi_1,\dots,\xi_{\dim{\mcH}})$ with $\xi_{j} = 0$
for $j > n_{+}$.

\item In the special case when $\rho(\vlambda')$ is twice differentiable
and full rank, \cref{lm:explicit.chi-2} can be found in Ref.~\cite[section 15.1]{bengtsson2006geometry}, with a derivation that uses logarithmic derivatives.
This approach does not allow one to relax the requirements on
$\rho(\vlambda')$ as we do here since some of the objects encountered in the proof with logarithmic derivatives are not well-defined.

\end{enumerate}

 \begin{mycorollary}
  \label{lm:explicit.chi-3}
 Let $p(\vlambda') = \{p_z(\vlambda')\}_{z \in \mcS}$ be a discrete
      probability distribution on a finite set $\mcS$ defined for $\vlambda'$ in
      the neighborhood of $\vlambda \in \mcM$.
      Let $\mcS_0 = \{z \in \mcS: p_z(\vlambda) = 0\}$ and
      $\mcS_+ = \mcS \setminus \mcS_0$ be the set of zeros and support of
      $p_\bullet(\vlambda)$ respectively.
      Assume that $p_z$ is differentiable at $\vlambda' = \vlambda$
      and $\sum_{z\in\mcS_0} p_z(\vlambda')$
      is twice differentiable at $\vlambda' = \vlambda$. Then
      \begin{align}
        \label{eq:explicit.chifc}
        (\chi_{F_c}(\vlambda))_{\mu\nu} &=
          \sum_{z\in\mcS_{+}}
            \frac{
              \left(\partial_{\mu} p_z(\vlambda)\right)
              \left(\partial_{\nu} p_z(\vlambda)\right)
            }{
              4 p_z(\vlambda)
            }\notag\\
          &\quad + \frac12 \left.\partial_{\mu} \partial_{\nu}
            \sum_{z\in\mcS_0} p_z(\vlambda')\right|_{\vlambda' = \vlambda}.
      \end{align}
\end{mycorollary}

\cref{lm:explicit.chi-3} trivially follows from \cref{lm:explicit.chi-2} when applied to diagonal $\rho$.

\subsection{Bound for the mixed state fidelity susceptibility}
\label{ass:chif.mixed.bound}

Next, we bound the second term in
\cref{eq:explicit.chi.mixed} and show that \cref{eq:explicit.chi.mixed} is equal
to \cref{eq:explicit.chi.mixed.v2} for most $\vlambda$. Note that $\rho_{+}(\vlambda')$ is invertible in a neighborhood
  of $\vlambda' = \vlambda$ because we defined it (in \cref{lm:explicit.chi-2}) to be an operator $\rho(\vlambda)\mcH \to \rho(\vlambda) \mcH$ (and not $\mcH \to \mcH$, which would not be invertible).

\begin{myproposition}
  \label{lm:chif.mixed.bound}
  Assume that the conditions in \cref{lm:explicit.chi-2} hold. Then:
  \begin{enumerate}
    \item The following inequality holds for the second term
    in \cref{eq:explicit.chi.mixed}:
      \begin{align}
        \label{eq:chi.mixed2.bound}
        &\frac12 \partial_\mu \partial_\nu \Tr(P_0 \rho(\vlambda) P_0^{\dagger})
        \geq  \\
        &\quad \real\Tr\left(
          P_0 \left(\partial_\mu\rho(\vlambda)\right)
          P_{+}^{\dagger} \left(\rho_{+}(\vlambda)\right)^{-1}
          P_{+} \left(\partial_\nu\rho(\vlambda)\right) P_0^{\dagger}
        \right).\notag
      \end{align}
    \item The strict inequality in \cref{eq:chi.mixed2.bound} along a vector
      $\vecv \in T_{\vlambda} \mcM$, i.e., the inequality
      \begin{multline}
        \label{eq:chi.mixed2.bound.v2}
        \frac12 v_{\mu} v_{\nu}\partial_\mu \partial_\nu \Tr(P_0 \rho(\vlambda) P_0^{\dagger})
        > \\
        v_{\mu} v_{\nu}\real\Tr\left(
          P_0 \left(\partial_\mu\rho(\vlambda)\right)
          P_{+}^{\dagger} \left(\rho_{+}(\vlambda)\right)^{-1}
          P_{+} \left(\partial_\nu\rho(\vlambda)\right) P_0^{\dagger}
        \right),
      \end{multline}
      is only possible when for every curve $\vlambda'\colon(-1, 1)\to \mcM$ with $\vlambda'(0) = \vlambda$ and $\left.\frac{\partial\vlambda'(t)}{\partial t}\right|_{t=0} = \vecv$ there is $\varepsilon > 0$ s.t. $\rank(\rho(\vlambda'(t))) > \rank(\rho(\vlambda))$ for $t \in (-\varepsilon, \varepsilon) \setminus \{0\}$. i.e., $\rank(\rho(\vlambda'))$ is discontinuous at $\vlambda' = \vlambda$ and greater than $\rank(\rho(\vlambda))$ in some punctured neighborhood of $\vlambda$ along any curve tangent to $\vecv$.
    \item When \cref{eq:chi.mixed2.bound} is an equality, \cref{eq:explicit.chi.mixed.v2} holds.
  \end{enumerate}
\end{myproposition}

The proof is given in \cref{app:proof-lm:explicit.chi-3}.

Part 2 of \cref{lm:chif.mixed.bound} essentially states
that the points where \cref{eq:chi.mixed2.bound} is not an equality are
``rare''. For example, if $\mcM$ is a twice differentiable manifold and the conditions
of \cref{lm:explicit.chi-2} are satisfied for all $\vlambda \in \mcM$,
then the set of points where \cref{eq:chi.mixed2.bound} is not an equality has
measure zero; everywhere else, \cref{eq:explicit.chi.mixed.v2} applies.
Given that part 3 of \cref{lm:chif.mixed.bound} is a special case of \cref{lm:explicit.chi-2},
a similar observation applies to \cref{eq:explicit.chifc}:
the second term in \cref{eq:explicit.chifc} is equal to zero unless $\abs{\mcS_{{+}}(\vlambda)}$
satisfies a condition similar to the one in part 2 of \cref{lm:chif.mixed.bound}.

\subsection{Example illustrating \texorpdfstring{\cref{lm:chif.mixed.bound}}{Prop. \ref*{lm:chif.mixed.bound}}}
\label{ass:example}

The following example illustrates the second term in \cref{eq:explicit.chifc} and
a situation where the bound \cref{eq:chi.mixed2.bound} is not tight.
Consider $\mcS = \{0, 1\}$, $\mcM = \mathbb{R}$; i.e.,
the mixed state in a two-dimensional Hilbert space is parameterized by
a single parameter $s = \lambda_0$. Let $p_0(s) = \sin^2(s)$, $p_1(s) = \cos^2(s)$, $\rho(s) = \diag(p_0, p_1)$.
Then
\begin{equation}
  \label{eq:explicit.chi:example}
  \chi_{F}(s) = \chi_{F_c}(s) = ds^2,\quad\text{i.e.,}\quad (\chi_F(s))_{00} = 1.
\end{equation}
The first term in \cref{eq:explicit.chi.mixed,eq:explicit.chifc}
is $1 - \delta_{\sin(2s), 0}$,
where $\delta$ is the Kronecker delta. The second term in both of these
expressions is $\delta_{\sin(2s), 0}$.
The bound (i.e. the r.h.s. of \cref{eq:chi.mixed2.bound}) is $0$ since $\rho(s)$ is diagonal.
If $\tilde{p}(s)$ is the outcome of measuring $\rho(s)$ in any orthonormal basis
of $\mcH$ rotated away from $\{\ket{0}, \ket{1}\}$ (i.e., a basis
$\{\ket{\tilde 0}, \ket{\tilde 1}\}$ such that $\abs{\braket{\tilde 0|0}} \not\in \{0, 1\}$),
and $\chi_{F_c}^{\tilde{p}}(s)$ is the corresponding classical fidelity susceptibility,
then $\chi_{F_c}^{\tilde{p}}(0) = 0 \neq 1 = \chi_{F_c}(0) = \chi_F(0)$. Indeed, at $s=0$ the
first derivative of $\rho(s)$ is zero and the second derivative only appears
in the second term of \cref{eq:explicit.chifc}, which is zero unless some
components of $\tilde{p}(0)$ are equal to $0$.

This example illustrates that when the bound
\eqref{eq:chi.mixed2.bound} is not tight, the classical fidelity susceptibility
is not continuous with respect to the choice of the measurement basis.
In more detail, this example shows that at $s=0$ the average of the classical fidelity susceptibility over all possible
measurement bases vanishes while the quantum fidelity susceptibility is $1$; this shows that \cref{th:avggc} does not hold for mixed states.

\subsection{Example illustrating inapplicability of
\texorpdfstring{\cref{th:comput.basis}}{Theorem \ref*{th:comput.basis}} to level crossings}
\label{ss:level.crossing}
To illustrate the breakdown of \cref{th:comput.basis} at level crossings, consider the simplest example
of a level crossing. I.e., consider Hilbert space $\mcH = \mathbb{C}^2$ and a family of Hamiltonians $H(s) = s X$
parameterized by $s \in (-1, 1)$, where $X$ is the Pauli $X$ matrix. This family is continuously
differentiable and has real matrix elements in the computational basis. Its ground state is non-degenerate
except for $s = 0$, where $H(s) = H(0) = 0$. Therefore, the conditions of \cref{th:comput.basis} are satisfied
everywhere except at $s = 0$. Indeed, $\chi_F(s) = \chi_{F_c}(s) = 0$ for $s \neq 0$. However, at $s = 0$
the ground state is degenerate, $\chi_F(0) = +\infty$, and $\chi_{F_c}(0) = 0$.

One might argue that in a practical numerical experiment, one only computes the ground state at a finite set
of points $s$ (referred to as ``grid points'' later) and it is unlikely that the location of a level crossing
will coincide with a grid point. Indeed, in our datasets, we deal with such finite grids and replace
$\chi_{F_c}$ by a finite difference: for a 1D grid that difference is given by
\begin{equation}
  \label{eq:finite.diff}
  \tilde{\chi}_{F_c}(s) = \frac{2}{\Delta s^2} \left(1 - F_c(s - \Delta s/2, s + \Delta s/2)\right).
\end{equation}
Replacing $F_c$ with $F$, one can write a similar expression for $\tilde{\chi}_F$. For intervals $[s - \Delta s/2, s + \Delta s/2]$ not containing $0$, $\tilde{\chi}_{F_c}(s) = \tilde{\chi}_F(s) = 0$. However, when such an interval contains $0$, $F(s - \Delta s/2, s + \Delta s/2) = 0$ and $\tilde{\chi}_F(s) = 2 / \Delta s^2$, but $F_c(s - \Delta s/2, s + \Delta s/2) = 1$ and $\tilde{\chi}_{F_c}(s) = 0$. This illustrates that a failure of the degeneracy condition in \cref{th:comput.basis} matters even if it occurs between grid points in a numerical experiment.

Note, however, that the lack of a peak in $\tilde{\chi}_{F_c}$ does not have to happen when a level crossing
occurs: if we pick $H(s) = s Z$ instead, then $\tilde{\chi}_{F_c}(s) = \tilde{\chi}_F(s) = 2/\Delta s^2$
for $0 \in (s - \Delta s/2, s + \Delta s/2)$. In the more realistic case of a many-body Hamiltonian,
$\tilde{\chi}_{F_c}$ is expected to capture the level crossing unless the probability distributions corresponding to
the states between which the transition occurs are identical (e.g., due to a symmetry).
Also note that ClassiFIM can be used to estimate $\chi_{F_c}$ regardless of whether $\chi_{F_c}$ matches $\chi_F$.

Among the four Hamiltonian families corresponding to the quantum datasets used in this work, Hubbard12 is the only one exhibiting an unavoided degeneracy on the grid points, specifically the 17 points with $\lambda_0 = 0$, i.e., the points on the boundary of the region. In the remaining three, this degeneracy is only expected at level crossings between the grid points.

\section{The FIM-estimation task and the
\texorpdfstring{C\lowercase{lassi}FIM}{ClassiFIM} method}
\label{sec:classifim}

In this section, we define the FIM-estimation task, outline its applications, and summarize the ClassiFIM method as well as its formal justification.

\subsection{The FIM-estimation ML task}
\label{sec:task}
The FIM-estimation task is based on a statistical manifold $(\mcM, P)$. We are given the following:
\begin{itemize}
\item A dataset of the form $\mcDtrain = \{(\vlambda_i, x_i)\}_{i=1}^{\abs{\mcDtrain}}$,
  where $\vlambda_i$ is a point in the parameter space $\mcM$ and $x_i$ is a sample from $P_{\vlambda_i}$.
\item The structure of the samples $x_i$. For example, if $x_i\in\{0,1\}^n$ is a bitstring representing a measurement of a quantum system on a lattice with $n$ sites, then there is a correspondence between the bitstring bits and the lattice sites, and the said structure is the lattice.
\item Some additional but optional information about the system. For example, the description or symmetries of the Hamiltonian may be provided for quantum systems on a lattice.
\end{itemize}

The task is to produce a representation of a function $\hat{g}$ --- an estimate of the (classical) FIM $g$ --- taking a point $\vlambda$ in the parameter space and returning a symmetric matrix corresponding to an estimate of the FIM from which we can construct a phase diagram. In particular, we require it to provide a precise and accurate estimate of the peak locations corresponding to critical points and reflect the qualitative features of the ground truth FIM phase diagram.

\subsection{Motivation and applications}
\label{ss:motivation}
As explained in \cref{sec:intro}, in the context of quantum systems, the task we set out to solve is to estimate the quantum FIM from \emph{classical data}. To reiterate, by classical data we mean data obtained from measurements in a single fixed basis, and despite the classical nature of the data, we still wish to estimate the \emph{quantum} fidelity susceptibility, e.g., in order to identify QPTs. As formally defined in \cref{sec:task}, the solution of the FIM estimation task provides an estimate of the \emph{classical} fidelity susceptibility $\chi_{F_c} = g/4$, but in many important cases it is expected to coincide with the quantum fidelity susceptibility $\chi_F$, as firmly grounded in the results we presented in \cref{s:fidelity.susceptibility} regarding the formal equivalence between the classical and quantum fidelity susceptibilities.

However, note that should random basis measurements be available (e.g., as in shadow tomography), the task of estimating $\chi_{F_c}$ from those measurements also fits into the definition of the FIM-estimation task: In that case, a sample $x_i$ would include the information about the basis in which the measurement was made.

Potential uses of solvers for the FIM-estimation task include, but are not limited to:
(i) phase transitions of quantum systems, whether in their ground state or in an excited state;
(ii) classical phase transitions in systems to which we have sampling access (e.g., from measurements or Monte Carlo simulations);
(iii) applications outside physics that involve access to random samples from a statistical manifold of interest.

The solvers for this task include ClassiFIM (described in \cref{ss:classifim} below), FINE~\cite{duy2022fisher}, the confusion-based method denoted $I_2$ in Ref.~\cite{arnold2023machine}, and the parameter estimation based method denoted $I_3$ in Ref.~\cite{arnold2023machine}. 

The datasets we introduce below represent (i) and (ii), and the context is given for each physical model used to obtain samples. In \cref{app:related-work} we provide a detailed literature survey of previous work related to our ML task.

\subsection{\texorpdfstring{C\lowercase{lassi}FIM}{ClassiFIM}}
\label{ss:classifim}

The companion paper~\cite{ClassiFIM-ML} proposes ClassiFIM as an unsupervised ML method to solve the task of estimating the FIM given a limited number of samples from a probability distribution. The method consists of three steps.

First, the original dataset $\mcDtrain = \{(\vlambda_i, x_i)\}_{i=1}^{\abs{\mcDtrain}}$ (which has no labels) is transformed into a new dataset, each entry of which can be viewed as a pair of points $\vlambda+\vdlambda/2$, $\vlambda -  \vdlambda/2$ in the parameter space and a sample $x$ from the probability distribution corresponding to one of these points.
The label, which a binary classifier (BC) is supposed to predict,
indicates which of the two points the sample $x$ corresponds to.
This dataset can be generated, e.g., via \cite[Algorithm 1]{ClassiFIM-ML},
or a vectorized version \cite[Appendix E.3]{ClassiFIM-ML}. Generation can
also be performed on the fly during training.

Second, one has to select and train a BC for this task. Specifically, the trained BC will output a function $\ln(p/(1-p))$, where $p = p(\vlambda, \vdlambda, x)$ is the estimated posterior probability that $x$ was taken from $\vlambda + \vdlambda /2$ (the prior of this probability without knowledge of $x$ is $1/2$). It is this BC that will contain a neural network (NN) in our specific realization.
This BC has to satisfy three technical conditions,
which are detailed in \cite[Section 4.3]{ClassiFIM-ML}.

Third, one can estimate the FIM from the predictions of this BC using \cite[Eq.~(2)]{ClassiFIM-ML}. Roughly, the estimate $\hat{g}$ is equal to an average of $\left(\left.\frac{d \ln(p/(1-p))}{d \vdlambda}\right|_{\vdlambda = 0}\right)^2$.

This approach is rigorously justified in the sense that if the input dataset
is infinite and the BC model is perfect, then the estimate $\hat{g}$ of the FIM
is accurate (i.e., $\hat{g} = g$; see \cite[Theorem 1]{ClassiFIM-ML}).

In practice, the dataset size is finite and the BC model
is imperfect, so it remains to be seen how well ClassiFIM performs in
practice. We present phase diagrams obtained using ClassiFIM
to answer this question numerically in \cref{sec:numerical-experiments}.

\section{Datasets}
\label{sec:validation-metrics-and-datasets}

In this section, we present two datasets (frustrated Ising ladder, Hubbard) to test and validate any algorithm that attempts to solve the task of estimating the FIM, and four datasets associated with two prior approaches to machine learning quantum phase transitions (nearest neighbor Ising, next-nearest neighbor Ising, Kitaev chain, XXZ) that were not open-source and required nontrivial modifications for our present purposes.
The datasets are available in \cite{public-datasets} and the code to generate them in \cite{ClassiFIM-gencode}.

\subsection{Generating datasets}
\label{subsec:data-sources}

In order to create inputs for FIM-estimation, one first needs to generate and collect a dataset $\mcDtrain = \{(\vlambda_i, x_i)\}_{i=1}^{\abs{\mcDtrain}}$ of samples taken from the probability distribution $P_{\vlambda}(x)$ at different parameter values.
For QPTs, each $x_i$ can be viewed as the outcome of a measurement of a quantum system, e.g., measurements in a fixed basis, classical shadow techniques, full tomography, etc.
For each of the datasets we present, we generated a larger dataset $\mcD$ with 140 samples per $\vlambda$ and randomly split it into $\mcDtrain$ ($90\%$) and $\mcDtest$ ($10\%$).

FIM-estimation assumes that we have access to a device that can perform this data generation and collection.
Here, we restrict ourselves to models for which $\mathcal{D}$ can be generated with a classical computer using a numerical scheme such as Lanczos, Markov chain Monte Carlo (MCMC), Quantum Monte Carlo (QMC), density matrix renormalization group (DMRG), or Fermionic Fourier Transform applied to an analytically known expression for the ground state in the momentum basis.
But in principle, this could be a ground state obtained using a quantum annealer \cite{king2022coherent}, analog quantum simulator \cite{Daley:2022vu}, or a gate-model quantum computer \cite{Kim:2023aa}, and measured in a fixed basis (the computational basis) at different points in the parameter space. Quantum annealers and analog quantum simulators are settings in which measurements are available in only one fixed basis, which was the original motivation for this work.

Importantly, since a direct classical simulation of the underlying quantum dynamics is typically inefficient, a quantum computational advantage could be expected in generating $\mcD$ directly using a quantum device \cite{Childs:2017aa}.
Combined with an efficient method of solving FIM-estimation (such as ClassiFIM \cite{ClassiFIM-ML}), we may thus expect a speedup over classical simulation methods in generating quantum phase diagrams.

A schematic illustrating the entire protocol when, e.g., a quantum annealer is the source of the data, is presented in \cref{fig:S1}.

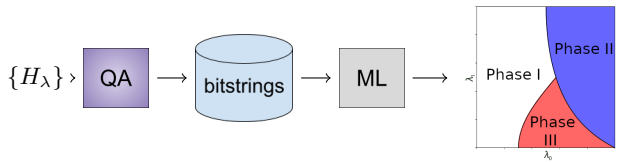
\begin{figure}[t]
\begin{center}
\pgfmathparse{\columnwidth/13.37cm}%
\edef\tikzscale{\pgfmathresult}%
\begin{tikzpicture}[scale=\tikzscale]
\newlength{\ux}
\newlength{\uy}
\setlength{\ux}{15mm}
\setlength{\uy}{11mm}
\draw (-3\ux, 0) node (H) {$\left\{H_\lambda\right\}$};
\draw (-1.9\ux, 0\uy) node (A2a)
{\includegraphics[width=0.1\columnwidth]{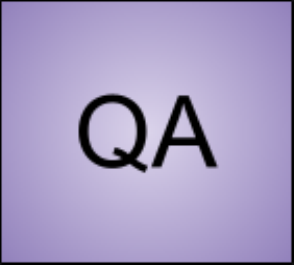}};
\draw (-0.15\ux, 0\uy) node (A2b)
{\includegraphics[width=0.15\columnwidth]{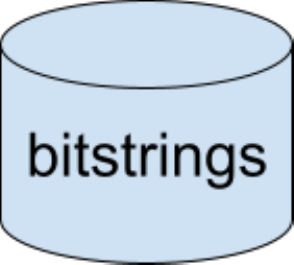}};
\draw (1.6\ux, 0\uy) node (A2c)
{\includegraphics[width=0.1\columnwidth]{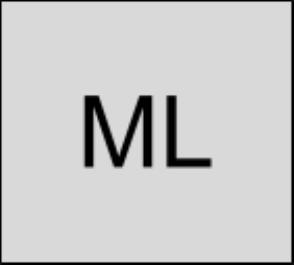}};
\draw (4\ux, 0\uy) node (B2)
{\includegraphics[width=0.28\columnwidth]{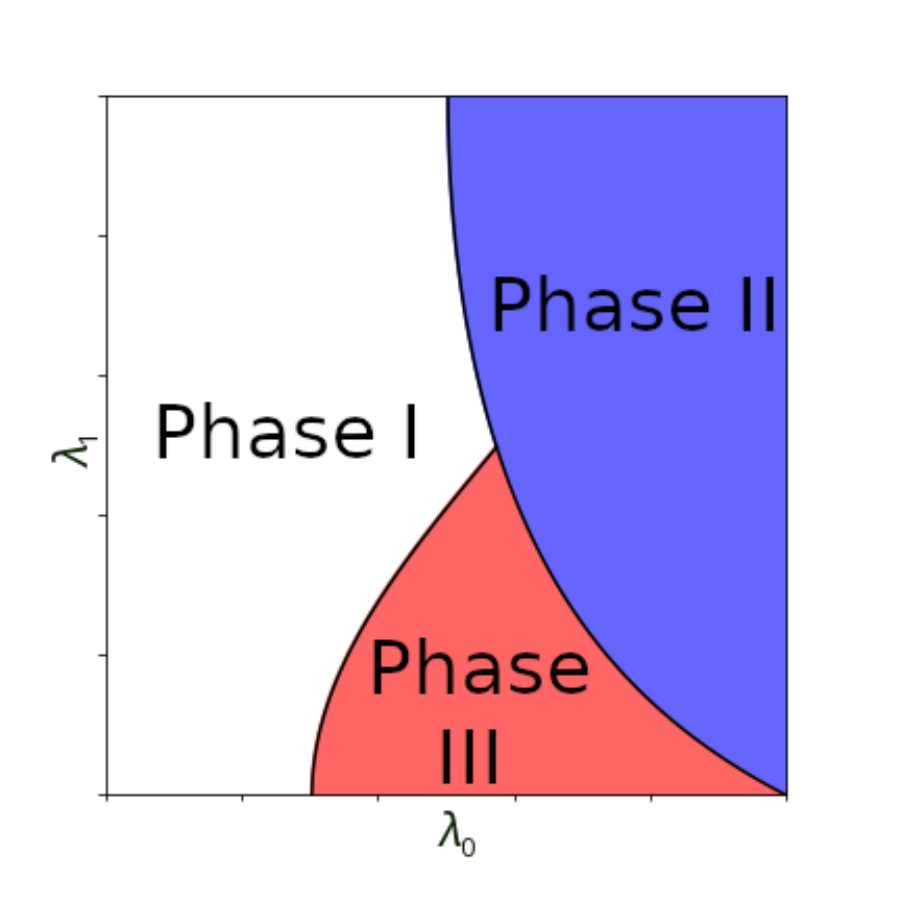}};
\draw [->] (H) -- (A2a.west);
\draw [->] (A2a) -- (A2b);
\draw [->] (A2b) -- (A2c);
\draw [->] (A2c) -- (B2);
\end{tikzpicture}
\end{center}
\caption{A setup for applications of ClassiFIM to quantum phase transitions. The Hamiltonian family $\{H_\lambda\}$ is time-evolved using quantum hardware such as a quantum annealer (QA), or its ground state is obtained by simulation on a classical computer. The system is measured in one or multiple bases that return classical data (samples), which are fed into a machine learning algorithm (ML). Our focus here is on the ML box, for which the ClassiFIM method has been developed \cite{ClassiFIM-ML}. The goal is to make predictions about the different phases characterizing the Hamiltonian family, as illustrated schematically in the last box.}
\label{fig:S1}
\end{figure}

We note that with a universal gate model quantum computer, the quantum fidelity between two states (from which the fidelity susceptibility, or FIM, can be computed) can be measured directly (e.g., via the SWAP test \cite{Barenco:1997aa}), and generating the bitstring dataset is, in principle, unnecessary. In practice, however, it remains relevant since the SWAP test requires storing two copies of the system's state and a multi-qubit gate between the two copies, which is far more resource-intensive than the single-copy measurement that is the starting point of the FIM-estimation task. Additionally, the SWAP test only provides the fidelity between two pure states, while the FIM-estimation task is valid for mixed states as well.
Similarly, if one has a unitary circuit $U(\vlambda)$ to prepare $\ket{\psi(\vlambda)}$, one can run $\bra{0} U(\vlambda')^{\dagger} U(\vlambda) \ket{0}$ to obtain the fidelity $F(\vlambda, \vlambda')$, but that approach requires doubling the circuit depth and may fail if the circuit includes non-unitary operations (e.g., measurements).

\subsection{Datasets used in this work}
\label{subsec:datasets}

\begin{table*}[t]
\centering
\begin{tabular}{cccccccccc}
\toprule
Name   &$\#x_i$& $\mcM$ &$\abs{\mcM'}$& System Type & PT Order & PT Type & State & Sampling method & Time \\
\toprule
Ising400 &$400$&$(0, 1]$& $1000$ & 2D Spin-1/2        & 2nd      & CPT     & Gibbs & MCMC          & $10$ \\
IsNNN400 &$400$& $(0, 1]^2$& $64 \times 64$& 2D Spin-1/2 & 1st, 2nd & CPT & Gibbs  & MCMC          & $20$ \\
FIL24    &$24$ & $[0, 1)^2$& $64 \times 64$& 2D Spin-1/2 & 1st, 2nd & QPTs & GS  & Lanczos         & $10$ \\
Hubbard12&$24$ & $[0, 1)^2$& $64 \times 64$& 2D Fermionic& unknown  & QPTs & GS  & Lanczos         & $10$ \\
Kitaev20 &$20$ & $(-4, 4)$ & $20000$       & 1D Fermionic& 2nd      & TQPT & GS  & FermFT          & $10$ \\
XXZ300   &$300$& $[0, 1]^2$& $64 \times 64$& 1D Spin-1/2 & 2nd      & TQPT & GS  & DMRG            & $125$\\
\bottomrule
\end{tabular}
\caption{We summarize the datasets generated and used in this work. The $\#x_i$ column is the number of bits in the bitstring, and $\abs{\mcM'}$ indicates the number of points and the shape of the grid $\mcM'$ of $\vlambda$'s. The ``Time'' column includes the suggested time limit in minutes for any method of FIM-estimation on the processor described in \cref{sec:numerical-experiments} (which can be rescaled as appropriate when training is done on other processors). ``GS'' means the probability distribution corresponding to the ground state of the Hamiltonian.}
\label{tab:datasets}
\end{table*}

A summary of the six primary datasets considered in this work is given in \cref{tab:datasets}. These datasets are generated from well-studied models on a lattice known to exhibit different types of phase transitions: CPTs (Ising, next-nearest-neighbor Ising), standard QPTs (frustrated Ising ladder, Hubbard model), and topological QPTs (Kitaev chain, XXZ chain). The dataset names are shortened forms of the corresponding models with a suffix $n$ indicating the number of lattice sites involved; e.g., in Hubbard12, the fermions are placed on a 12-site lattice. The precise Hamiltonian of each model and its physical significance are discussed next.

\subsubsection{Ising400}
The 2D Ising model is a simple classical model for magnetism and phase transitions~\cite{Sachdev:book}. The Hamiltonian of this model is given by
\begin{equation}
H = - J \sum_{\langle i, j \rangle} \sigma_i \sigma_j - h \sum_j \sigma_j
\end{equation}
where $\langle i, j \rangle$ denotes the underlying lattice of nearest neighbor connections, $\sigma_j$ can take on values $\pm 1$, and a state configuration consists of the value of each $\sigma_j$. For a given temperature $T$, the Gibbs thermal equilibrium state is given by
\begin{equation}
\rho_G = \frac{e^{- H / T} } { \Tr(e^{-H / T})},
\end{equation}
where we have chosen units where Boltzmann's constant $k_B = 1$. For the Ising400 dataset, we select the classical 2D Ising model on a $20 \times 20$ square lattice with periodic boundary conditions and a single parameter, $T$: we set $J=1$ and $h=0$. In the limit of infinite lattice size the model is known analytically to exhibit a phase transition at $T_c = 2 / \ln(1 + \sqrt{2})\simeq 2.269$ \cite{kramers1941statistics}. We choose the range $0\leq T \leq 4$, thus $\lambda_0 =T/4$.

\subsubsection{IsNNN400: Ising model with next-nearest-neighbor interactions}
\label{ass:ising-nnn-400}
Similarly to Ising400, IsNNN400 is an Ising model on a square $20\times 20$
lattice. This time, however, the couplings are antiferromagnetic and include
interactions along diagonals of the squares of the lattice. We denote the lattice by $\Lambda$,
the set of vertical and horizontal edges by $\NN(\Lambda)$, and the set of diagonal
edges by $\NNN(\Lambda)$. The Hamiltonian is given by
\begin{equation}
  H = J_1 \; \sum_{\mathclap{(i, j)\in \NN(\Lambda)}}\; Z_i Z_j + J_2 \; \sum_{\mathclap{(i,j) \in \NNN(\Lambda)}}\; Z_i Z_j.
\end{equation}
This Hamiltonian family was previously analyzed in Ref.~\cite{kalz2011analysis}, where Fig.~2 contains two glued phase diagrams. We pick a similar region in
the parameter space as the one displayed in Ref.~\cite[Fig.~2]{kalz2011analysis} with slight
adjustments to ensure that our axes are $\lambda_0, \lambda_1 \in (0, 1]$ and that $T, J_1, J_2$ depend analytically on $\lambda_0$ and $\lambda_1$:
we pick $J_1 = \lambda_0$, $J_2 = 1 - \lambda_0$, $T = 2.5 \lambda_1$.
Alternatively, since the Gibbs distribution depends only on $H/T$, one can interpret
the same phase diagram as $T = 1$, $J_1 = 0.4 \lambda_0 / \lambda_1$,
$J_2 = 0.4 (1 - \lambda_0) / \lambda_1$.

\subsubsection{FIL24: 2D version of the frustrated Ising ladder}

The frustrated Ising ladder is a simple model that exhibits a first-order QPT~\cite{laumann2012quantum}. The 1D model has a Hamiltonian given by
\begin{equation}
\label{eq:Hladder2}
H = (1-s) H_0 + s H_1,
\end{equation}
where
\bes
\begin{align}
\label{eq:Hladder0}
H_0 &= -\sum_{i=0}^{L-1} (X_{T_i} + X_{B_i})\\
\label{eq:Hladder1}
H_1 &= \sum_{i=0}^{L-1} \big(K Z_{T_i} Z_{T_{i+1}} - K Z_{T_i} Z_{B_i} - K Z_{B_i} Z_{B_{i+1}}\notag \\
&\qquad - K Z_{T_i} + \frac{U}{2} Z_{B_i}\big),
\end{align}
\ees
where $T_L = T_0$, $B_L = B_0$. This Hamiltonian can be identified with the lattice in \cref{fig:FIL1D.lattice}.

Similarly to the 1D version above, the 2D version that we study in this work has a top layer and a bottom layer. In FIL24 each layer is represented by the same
12-site lattice (see~\cref{fig:twelve-sites}) as used in Hubbard12
with NN edges of the top layer used in anti-ferromagnetic couplings and NN edges of the bottom layer used in ferromagnetic couplings. The NNN edges of the 12-site lattice are not used.
The Hamiltonian is given by \cref{eq:Hladder2} with
\bes
\begin{align}
\label{eq:FIL24.H0}
H_0 &= -\sum_{i\in\Lambda} (X_{T_i} + X_{B_i})\\
\label{eq:FIL24.H1}
H_1 &=
\sum_{i\in\Lambda}\left(
\frac{U}{2} Z_{B_i} - K Z_{T_i} Z_{B_i} - K Z_{T_i}\right)\notag \\
&\quad + \frac{K}{2}\sum_{(i,j)\in\text{NN}(\Lambda)} \left(
Z_{T_i} Z_{T_{j}}  - Z_{B_i} Z_{B_{j}}\right) .
\end{align}
\ees
The main difference relative to \cref{eq:Hladder0,eq:Hladder1}
is the factor of $1/2$ in the second term of \cref{eq:FIL24.H1}, which is
included to balance the fact that the ``horizontal'' degree of each site is 4 instead of 2 in the 1D frustrated Ising ladder.
For FIL24 we select $\lambda_0=s, \lambda_1=U$ and fix $K=1$.

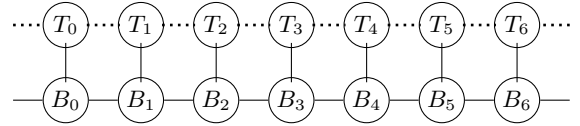
\begin{figure} 
\centering
\begin{tikzpicture}
\foreach \x in {0,...,6}{
\draw (\x,0) circle (3mm) node (n0x\x){$T_{\x}$};
\draw (\x,-1) circle (3mm) node (n1x\x){$B_{\x}$};
\draw (n0x\x) -- (n1x\x);
}
\foreach \x[count=\xnext from 1] in {0,...,5}{
\draw[dotted,line width=1pt] (n0x\x) -- (n0x\xnext);
\draw (n1x\x) -- (n1x\xnext);
}
\draw[dotted,line width=1pt] (-0.7,0) -- (n0x0);
\draw (-0.7,-1) -- (n1x0);
\draw[dotted,line width=1pt] (6.7,0) -- (n0x6);
\draw (6.7,-1) -- (n1x6);
\end{tikzpicture}
\caption{$L=7$ lattice for the 1D frustrated Ising ladder Hamiltonian family. Solid lines represent ferromagnetic couplings, dotted lines represent antiferromagnetic couplings.}
\label{fig:FIL1D.lattice}
\end{figure}

\subsubsection{Hubbard12: Hubbard Hamiltonian family on a 12-site lattice at half-filling}

\begin{figure} 
\centering
\begin{tikzpicture}
\def\dx{1cm}
\def\dy{1cm}
\clip (-3 * \dx / 2,7 * \dy/6) rectangle (7*\dx/2, -29*\dy/6);
\newcommand{\calcX}[1]{mod(#1, 3) * \dx}
\newcommand{\calcY}[1]{-#1 * \dy / 3}

\foreach \j in {0, 1, ..., 55} {
\pgfmathtruncatemacro{\jx}{mod(\j,7)}
\pgfmathtruncatemacro{\jy}{div(\j,7)}
\pgfmathtruncatemacro{\label}{mod(4 + \jx + 3 * \jy, 12)}
\tikzset{sitestyle/.style={circle, draw, minimum size=0.6cm, inner sep=0pt}}
\node[sitestyle] (n-\j) %
at ({(\jx - 2) * \dx}, {-(\j-16) * \dy / 3 + (\jy-2) * 4 * \dy / 3}) {\label};
}

\foreach \jx in {0, 1, ..., 5} {
\foreach \jy in {0, 1, ..., 7} {
\pgfmathtruncatemacro{\j}{\jx + 7 * \jy}
\pgfmathtruncatemacro{\nextj}{\j + 1}
\draw[blue] (n-\j) -- (n-\nextj);
}
}

\foreach \jx in {0, 1, ..., 6} {
\foreach \jy in {0, 1, ..., 6} {
\pgfmathtruncatemacro{\j}{\jx + 7 * \jy}
\pgfmathtruncatemacro{\nextj}{\j + 7}
\draw[blue] (n-\j) -- (n-\nextj);
}
}

\foreach \jx in {0, 1, ..., 5} {
\foreach \jy in {0, 1, ..., 6} {
\pgfmathtruncatemacro{\jnw}{\jx + 7 * \jy}
\pgfmathtruncatemacro{\jne}{\jnw + 1}
\pgfmathtruncatemacro{\jsw}{\jnw + 7}
\pgfmathtruncatemacro{\jse}{\jnw + 8}
\draw[red] (n-\jnw) -- (n-\jse);
\draw[red] (n-\jne) -- (n-\jsw);
}
}
\fill[white, opacity=0.7, even odd rule] %
(current bounding box.south west) %
rectangle (current bounding box.north east) %
(-\dx/2, 2*\dy/3) -- (5*\dx/2, -\dy/3) -- (5*\dx/2, -13*\dy/3) %
-- (-\dx/2, -10*\dy/3) -- cycle;
\end{tikzpicture}
\caption{12-site lattice $\Lambda$ with two edge types: NN edges (blue) and NNN edges (red).
NN edges connect node $j$ with nodes $(j\pm 1) \bmod 12$ and $(j\pm 3) \bmod 12$.
NNN edges connect node $j$ with nodes $(j\pm 2) \bmod 12$ and $(j\pm 4) \bmod 12$.}
\label{fig:twelve-sites}
\end{figure}
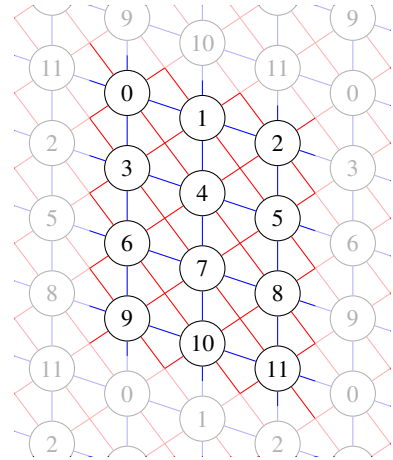

The Hubbard model exhibits a transition from a normal to a superconducting phase, usually represented on a phase diagram of temperature and doping~\cite{gull2013superconductivity}. It is a fermionic model on a 2D lattice, with doping roughly equivalent to deviations of the density of fermions from half-filling. The dataset we construct corresponds to a single point on the phase diagram: zero temperature and zero doping. We explore different phases by varying nearest-neighbor and next-nearest-neighbor hopping. While phases are usually defined in the thermodynamic limit, we restrict ourselves to a numerically tractable 12-site lattice. The lattice is described in \cref{fig:twelve-sites}. The reason we chose the boundary condition as shown is so that the lattice is bipartite (which would not have been the case for the regular boundary conditions on a $3\times 4$ square lattice).
The Hubbard Hamiltonian family acts on a Hilbert space describing configurations
of electrons on this grid with two possible spin values:
$\sigma \in \{\uparrow,\downarrow\}$.
The Hubbard Hamiltonian is given by
\begin{equation}
\label{eq:Hubbard12.H}
H = u H_{\text{int}} + t H_{\text{NN}} + t' H_{\text{NNN}},
\end{equation}
where
\bes
\begin{align}
H_{\text{int}} &= \sum_{i\in\Lambda} (n_{j\uparrow}-1/2) (n_{j\downarrow}-1/2)\\
H_{\text{NN}} &= -\;\sum_{\mathclap{(i,j)\in\text{NN}(\Lambda),\sigma}}\;c_{i\sigma}^{\dagger}c_{j\sigma},
\quad H_{\text{NNN}} = -\;\sum_{\mathclap{(i,j)\in\text{NNN}(\Lambda),\sigma}}\;c_{i\sigma}^{\dagger}c_{j\sigma},
\end{align}
\ees
where $n_{j\sigma} = c_{j\sigma}^{\dagger} c_{j\sigma}$ is the number operator.
Here $\NN(\Lambda)$ is the set of all pairs $(i,j)$ of lattice
sites connected via nearest-neighbor (NN) edges: e.g., both $(0, 1)$ and $(1, 0)$ are in $\NN(\Lambda)$. Similarly, $\NNN(\Lambda)$ is the set of
pairs of lattice sites connected via next-nearest-neighbor (NNN) edges.

We note that the ground state is not changed when the Hamiltonian is multiplied by a positive real number.
We parameterize the region $u \geq 0$, $t \geq 0$, $t' \geq 0$ (up to overall scaling; the sign of $t$ can be freely changed by redefining $c$ on one of the sublattices)
by two real parameters:
\begin{equation}
\label{eq:lambda01}
\lambda_0 = \frac{16 t}{u + 16 t}, \quad \lambda_1 = \frac{16t'}{u + 16t'}.
\end{equation}
The coefficient $16$ is intended to place $u$ and $t$ on the same
scale: $H_{\textnormal{int}}$ has a single term per site
with absolute value at most $1/4$, while $H_{\textnormal{NN}}$ and $H_{\textnormal{NNN}}$ each have four terms per site. Hence, in \cref{eq:lambda01} we multiply the parameters $u, t, t'$ by the same coefficients as a rough estimate of
the scale of the corresponding Hamiltonian terms in
\cref{eq:Hubbard12.H}. When converting $\vlambda$ to $(u, t, t')$
we select $u, t, t'$ satisfying \cref{eq:lambda01} and the normalization
condition $u/4 + 4(t+t') = 1$.

\subsubsection{Kitaev20: Kitaev chain} 

The Kitaev chain is a 1D model defined in terms of fermionic creation and annihilation operators~\cite{Kitaev:2001aa}. It exhibits a TQPT and Majorana modes. The general Hamiltonian is given by
\begin{align}
\hat{H} = -t\sum_{i=0}^{L-1} \left(  \hat{c}_{i+1}^{\dagger}\hat{c}_{i}
+ \hat{c}_{i+1}\hat{c}_{i} + h.c. \right)
- \mu\sum_{i=0}^{L-1} \hat{c}_{i}^{\dagger}\hat{c}_{i}, \label{eq:KitaevH}
\end{align}
where $t>0$ controls the hopping and the pairing of spinless fermions alike and $\mu$ is a chemical potential. The ground state of this model has a quantum phase transition from a topologically trivial ($|\mu|>2t$) to a nontrivial state ($|\mu|<2t$) as the chemical potential $\mu$ is tuned across $\mu = \pm2t$. This model was studied in Ref.~\cite{vanNieuwenburg2016LearningPT} for $L = 20$. We likewise choose $L = 20$, which defines the Kitaev20 model. We set $t=1$ and investigate the range $-4 \leq \mu \leq 4$ using 20,000 grid points.

\subsubsection{XXZ300: Bond-alternating XXZ model}
\label{ss:XXZdef}
The bond-alternating XXZ model is a 1D spin chain studied in Ref.~\cite{Elben} that is known to exhibit a TQPT. A 300 qubit instantiation of this model was considered in Ref.~\cite{Huang:22} with the following Hamiltonian:
\begin{equation}
\label{eq:XXZ-H0}
\bal
H_0 &= 0.1 Z_0 + \sum_{i=1: \text{odd}}^{299} (X_i X_{i-1} + Y_i Y_{i-1} + \delta Z_i Z_{i-1}) \\
&\quad + \sum_{i=2:\text{even}}^{298} J' (X_i X_{i-1} + Y_i Y_{i-1} + \delta Z_i Z_{i-1}).
\eal
\end{equation}
This is known as bond-alternating since the interaction strength alternates between $1$ and $J'$. The first term is not part of the XXZ model and is included to break the ground state degeneracy. However, we found that when using DMRG, in some regions of the parameter space \cref{eq:XXZ-H0} allows for the energies of the ground and first excited state to be indistinguishable despite this first term.
In order to fix this, we work with the Hamiltonian
$H = H_0 + 2 \cdot 10^{-8} (\sum Z_i)^2$, where
the additional term is added so that DMRG avoids these degeneracies.
We call the dataset we generated from this Hamiltonian XXZ300. We choose $\lambda_0 =\frac{2}{5}J',~\lambda_1 = \frac{2}{7}\delta$ to lie in the interval $(0,1)$ to match Ref.~\cite[Fig.~4B]{Huang:22}.

\subsection{Choice of sample size}

For convenience, we rescaled the physical parameters in most datasets, so that $\lambda_i \in [0, 1]$ in most datasets
(see column $\mcM$ of \cref{tab:datasets}
for the exact range for each dataset).
The Ising400 and Kitaev20 models are defined by a single parameter $\lambda_0$, whereas FIL24, Hubbard12, IsNNN400 and XXZ300 are defined by two: $\lambda_0$ and $\lambda_1$. We note that for the Ising400 and IsNNN400 datasets, which arise from classical models, $\lambda_0$ is proportional to the temperature, but for the other models the temperature is set to zero. Temperature is one of the essential differences between classical and quantum phase transitions, but ClassiFIM is insensitive to the meaning of these parameters.

For each dataset, we (i) computed the ground truth FIM and (ii) generated 140 samples for a $64 \times 64$ grid of $\vlambda$ values (or a line of $\abs{\mcM'}$ values for Ising400 and Kitaev20) and split the resulting dataset $\mcD$ into $90\%$ $\mcDtrain$ and $10\%$ $\mcDtest$. For Ising400 and IsNNN400, directly sampling the classical Gibbs state using MCMC is sufficient for both (i) and (ii) (see \cref{as:Ising}). For the remaining (quantum) models, we first needed to prepare the ground state of $H_n(\vlambda)$.
For FIL24 and Hubbard12, we employed the general-purpose Lanczos-Arnoldi method with several nontrivial system-specific modifications (see \cref{ass:fill-hubbard-datasetgen}). For Kitaev20, we used the inverse fermionic Fourier transform to convert the known analytical ground state (in Fourier space) back into real space (see \cref{ass:kit20-dataset-gen}).  For XXZ300, we employed a DMRG scheme similar to Ref.~\cite{Huang:22} but with several modifications to more reliably sample the ground state (see \cref{ass:xxz300-dataset-gen}).

After preparing the ground state with these methods, it is straightforward to efficiently estimate the ground truth FIM and sample measurements in the computational basis, i.e., the simultaneous eigenbasis of the Pauli operators $Z_i$. For this reason, we defaulted to computational basis samples, which, while not necessary (i.e., sampling in another basis would have worked equally well), are both backed by our theoretical findings and empirically sufficient; see \cref{sec:numerical-experiments}.

All datasets and methods for their generation are open-source and can be found in Ref.~\cite{public-datasets,ClassiFIM-gencode,ClassiFIM-code}.

\begin{figure*}[t]
\centering
\includegraphics[width=0.99\linewidth]{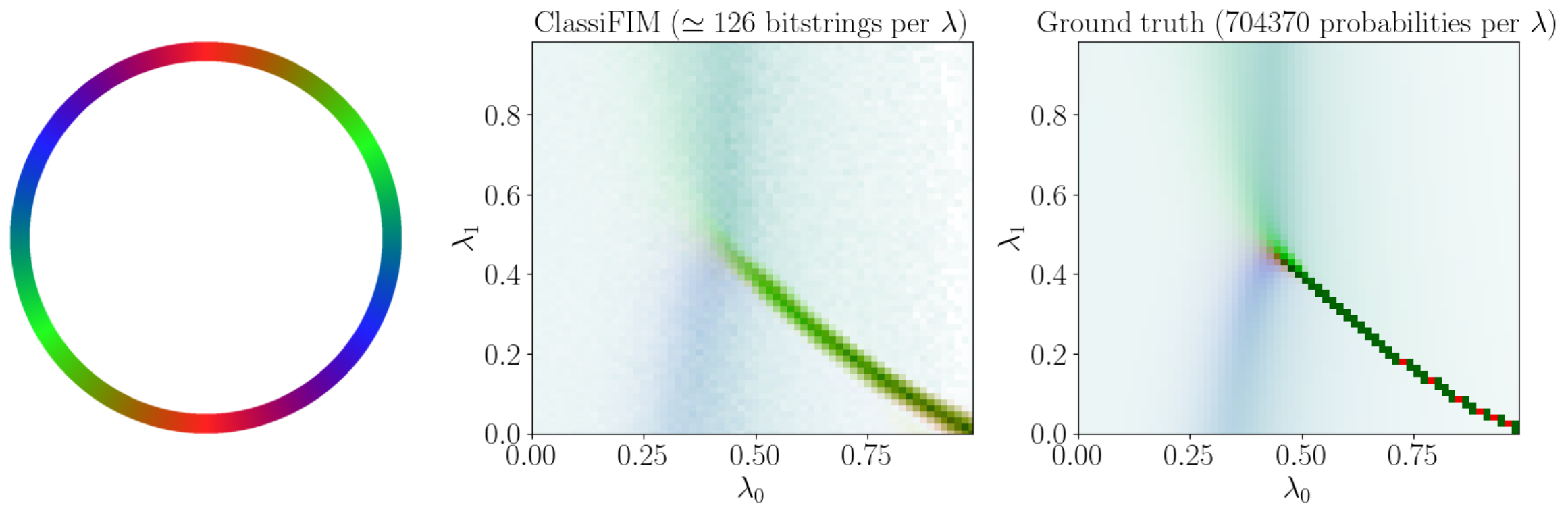}
\includegraphics[width=0.99\linewidth]{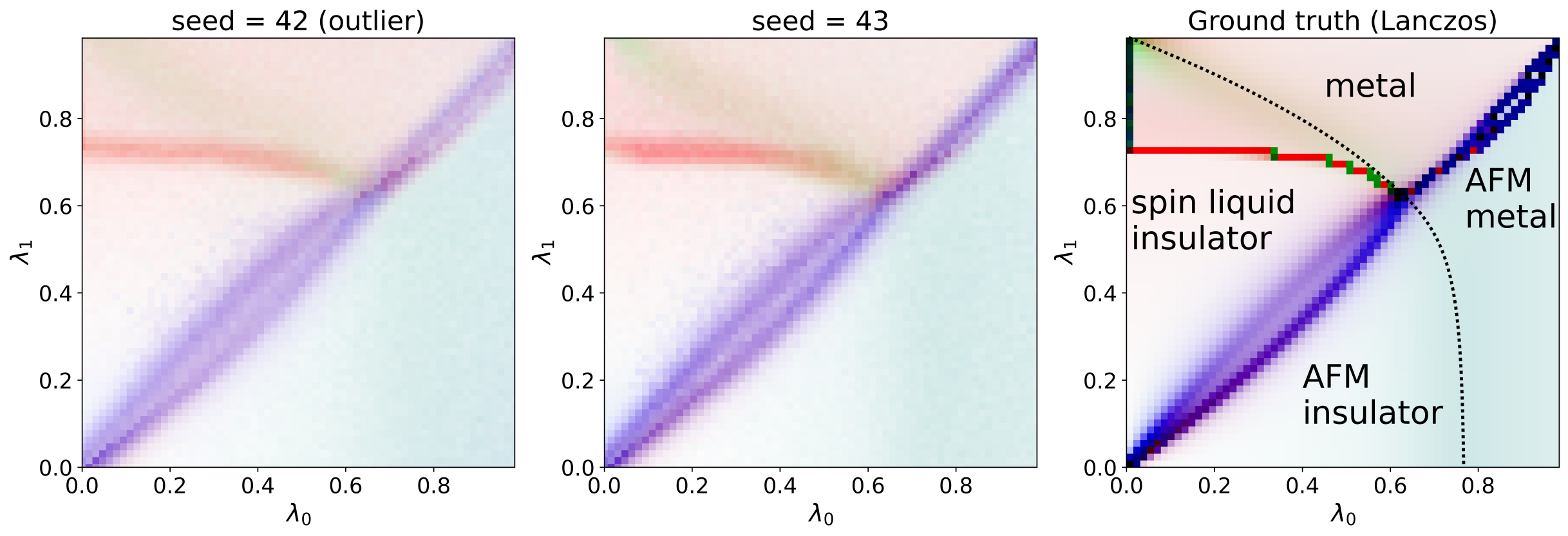}
\caption{Phase diagrams for FIL24 and Hubbard12. In general, darker colors correspond to higher values of the FIM. In all diagrams, the colored lines indicate locations with a higher FIM, delineating putative phase boundaries.\\
Top row: FIL24.
Left: The color ring represents the metrics $g_{\mu\nu}(\vlambda)$ and $\hat{g}_{\mu\nu}(\vlambda)$. Each color on the ring indicates how the rank-$1$ metric $g(\vlambda) = \vecv \vecv^T$ is colored, where $\vecv$ is a length $C = 54.0$ vector from the center of the circle to the point on the ring. Hue is determined by the direction of the vector $\vecv$. Black (not shown on the color ring) represents $g(\vlambda) \geq C^2 I$(where $I$ is the identity matrix), while white represents $g(\vlambda) = 0$. In general, darker colors on the diagrams correspond to higher values of the FIM. 
Middle: ClassiFIM prediction ($g_{\mu\nu}(\vlambda)$). 
Right: ground truth FIM ($\hat{g}_{\mu\nu}(\vlambda)$ as computed by Lanczos-Arnoldi diagonalization.\\
Bottom row: Hubbard12. 
Left and center:  ClassiFIM prediction of $g_{\mu\nu}(\vlambda)$ for two different random seeds. As indicated by the lighter colors, the diagram on the left represents an outlier, where an instability occurred during training, affecting the quality of the predictions. See also \cref{fig:supp-phase1D}.
Right: ground truth FIM ($\hat{g}_{\mu\nu}(\vlambda)$ as computed by Lanczos-Arnoldi diagonalization.
}
\label{fig:phase2d}
\end{figure*}

\begin{figure*}[t]
\centering
\includegraphics[width=0.495\linewidth]{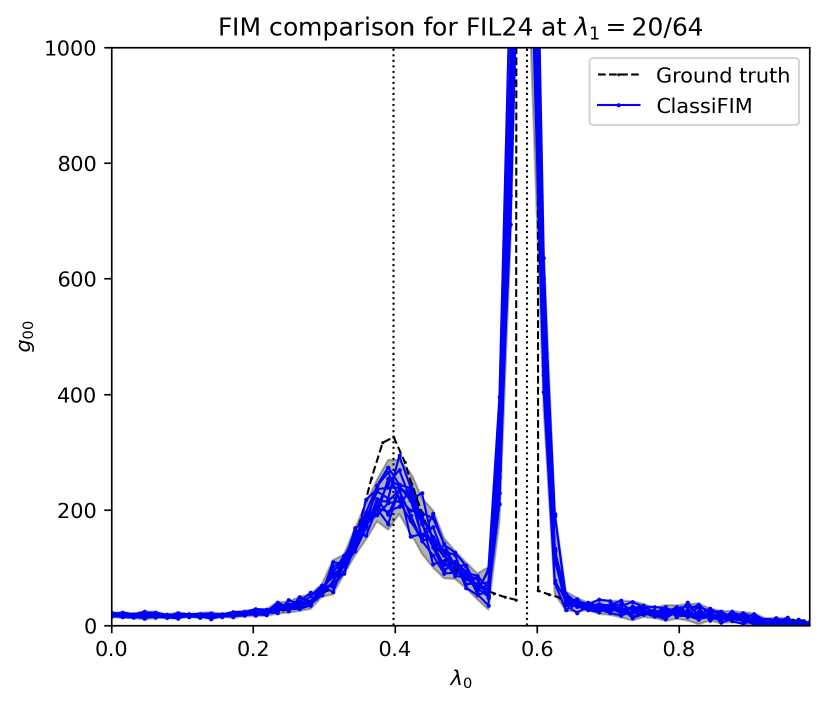}
\includegraphics[width=0.495\linewidth]{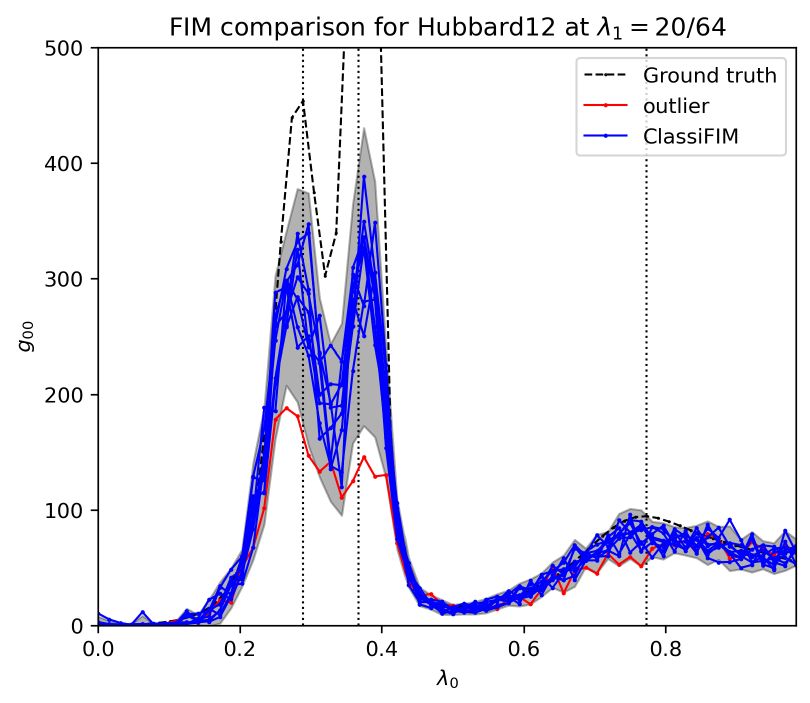}
\caption{1D sections of the phase diagrams from \cref{fig:phase2d}. Left: FIL24. Right: Hubbard12. Shaded regions correspond to two standard deviations. The peaks correspond to the locations of putative quantum phase transitions, and are in good agreement with the dotted vertical lines, which mark the locations of the ground truth peaks. Sharp peaks likely correspond to first-order phase transitions while the more subtle peaks may correspond to higher-order phase transitions.
The outlier curve in the right plot corresponds to seed=$42$, where an instability was encountered during training (see \cref{fig:phase2d}, bottom left). }
\label{fig:supp-phase1D}
\end{figure*}

\section{Numerical experiments}
\label{sec:numerical-experiments}

In this section, we demonstrate that ClassiFIM can successfully estimate ground truth FIM diagrams and that these diagrams correctly encode known physics.
In more detail, in this section, we discuss the physics of the ground truth FIM found for the systems used for the dataset generation in the previous section and report on the results of applying ClassiFIM along with the other methods discussed earlier to these datasets. To implement ClassiFIM, we designed a neural network (NN) as described in Ref.~\cite{ClassiFIM-ML}. 
For all datasets, the NN included a convolutional neural network (CNN)-like set of layers for processing input bitstrings according to the lattice geometry; e.g., for FIL24 and Hubbard12, the information could propagate along the edges of the 12-site lattice. The time in minutes we gave ClassiFIM to estimate the FIM for each dataset on a single NVIDIA GeForce GTX 1060 6GB GPU is given in \cref{tab:datasets}. These times were chosen as follows. First, we found that scaling the times linearly with the system size was sufficient. For example, since FIL24 was given 10 minutes, XXZ300 (with 12.5 times more qubits) was given 125 minutes. This estimate comes from machine learning intuition on the cost of training NN of various sizes, but it is not a rigorous complexity estimate. 
Second, for the QPT datasets, we gave ClassiFIM substantially less time than it took to generate the datasets and ground truth FIM (e.g., this took about 3000 CPU-hours for Hubbard12 and about 12000 CPU-hours for XXZ300; see \cref{ass:compute} for more details).

\subsection{ClassiFIM performance by phase diagrams for FIL24 and Hubbard12}
\label{sec:phase-diagrams}

A significant benefit of our method is the ability to generate visual FIM phase diagrams. To illustrate this, \cref{fig:phase2d} shows the FIM phase diagram generated by ClassiFIM compared to the ground truth computed by exact diagonalization (right column) for the FIL24 (top) and Hubbard12 (bottom) models. It is evident that ClassiFIM and ground truth are in close qualitative agreement, in particular regarding the phase boundaries. The FIL24 FIM phase diagram agrees with previously established phase diagrams; e.g., see Refs.~\cite[Fig.~2]{takada2021phase} and \cite[Fig.~3]{laumann2012quantum}, after changing coordinates to match their conventions. 

While first-order phase transitions correspond to singularities of the FIM in the thermodynamic limit and to sharp peaks of the FIM for finite-size systems (i.e., finite-size precursors of phase transitions), higher-order phase transitions --- corresponding to a discontinuity in higher than the first derivative of the (free or ground state) energy \cite{GU:2010vv,Sachdev:book} --- may manifest as subtle peaks for finite-size systems, which are more difficult to identify. This is what we observe in \cref{fig:supp-phase1D}, where sharp peaks correspond to first-order phase transitions, and more subtle, flatter peaks likely correspond to higher-order phase transitions. 

The FIL24 phase diagram in the top row of \cref{fig:phase2d} is consistent with earlier 1D results for a ladder geometry~\cite{laumann2012quantum} and indeed features both first and second-order phase transitions. This is seen most clearly in \cref{fig:supp-phase1D}, where the large (small) peak near $\lambda_0=0.6$ ($0.4$) corresponds to a first-order (second-order) QPT. We find that for this model, the phase structure of a two-layer 2D lattice is inherited from the 1D physics of the ladder. 

The bottom row of \cref{fig:phase2d} can be mapped to the $t-t'$ phase diagram of the Hubbard model found in Ref.~\cite[Fig.~1]{kyung2006mott}. We identify the bottom right ($t>t'$) phase with an unfrustrated antiferromagnet (AFM), while the two phases in the $t'>t$ regime contain frustration on sublattices connected by $t'$ terms. Mapping the phase diagram \cite{kyung2006mott} we identify the smaller $t'$ phase as a spin liquid and the higher $t'$ phase as a featureless metal. The minor features along the $t=t'$ line are likely boundary-condition dependent and do not directly correspond to an island of superconductivity found in Ref.~\cite{kyung2006mott}. Besides the transitions mentioned above, a barely visible peak (\cref{fig:supp-phase1D}, right) can be found in the FIM data going from the top left corner all the way into the unfrustrated phase, ending in a vertical line $\lambda_0\approx 0.77$. References \cite{turkowski2021one,arovas2022hubbard,gull2013superconductivity,tanaka2019metal,terletska2011quantum} suggest that the metal-insulator transition indeed continues into the AFM phase (lower triangle). Since it is a dynamical transition, we did not a priori expect it to be visible in static observables or in the quantum fidelity susceptibility. It is a surprising result that the signatures of the metal-insulator transition appear in our data, albeit weakly. 

We remark that a fixed amount of training data per $\vlambda$ ($126$ computational basis measurements) was sufficient to produce an accurate FIM phase diagram (and hence detect QPTs) for each dataset. This is especially surprising for XXZ300 (discussed in detail in \cref{sec:technical-comparison} below) which (i) exhibits a TQPT with no local order parameter and (ii) has a Hilbert space dimension of $2^{300}$. This is consistent with the result of Ref.~\cite{Huang:22} but less demanding as they used $500$ classical shadows per $\vlambda$ for the same model.

Finally, we note that we used the same neural network architecture and training procedure for FIL24 as for Hubbard12 without any additional fine-tuning. Despite representing completely different physical systems (24 spins and 12 fermions, respectively), this strategy yielded good ClassiFIM results for both FIL24 and Hubbard12. This suggests that the success of ClassiFIM does not rely on hyperparameter choices specific to a dataset, and the main work needed to adapt ClassiFIM to a new dataset is to adjust the neural network according to the format, size, and symmetries of the input samples.

\subsection{Comparison with Refs.~\cite{vanNieuwenburg2016LearningPT,Huang:22}}
\label{sec:technical-comparison}

\begin{figure*}
\includegraphics[width=0.99\linewidth]{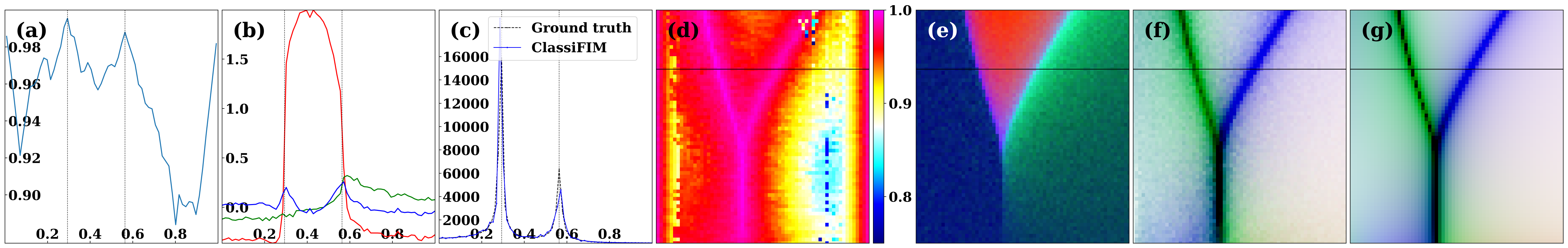}
\caption{Phase diagrams for the XXZ300 dataset, generated by W~\cite{vanNieuwenburg2016LearningPT} (a,d), SPCA~\cite{Huang:22} (b,e), and ClassiFIM (c,f). Panels (a)-(c) are horizontal slices through the 2D phase diagrams (d)-(f). Panel (g) depicts the ground truth FIM $g_{\mu \nu}$ obtained using DMRG. The vertical lines in (a-c) mark the positions of the ground truth peaks in $g_{00}$. In (a) and (d) the value between 0 and 1 is the W-shaped confusion~\cite{vanNieuwenburg2016LearningPT} that is expected to peak at the phase transition. In (b) and (e), the first three principal components obtained from classical shadows \cite{Huang:22} are plotted as the red, green, and blue curves in (b) and a color in RGB format in (e). Following Ref.~\cite{Huang:22}, we identify the phases in (e) as trivial (blue), symmetry-broken (red), and topological (green); thus the vertical phase transition line separating the blue and green phases is a TQPT.
The component $\hat{g}_{00}$ of the estimated FIM is plotted in panel (c), while the color scheme of panels (f) and (g) is the same as in \cref{fig:phase2d}. The phase diagrams in panels (d)-(g) depict the square $0\leq \lambda_0, \lambda_1 \leq 1$, where the rescaled parameters $\lambda_0 =\frac{2}{5}J',~\lambda_1 = \frac{2}{7}\delta$ are defined in \cref{ss:XXZdef}.
}
\label{fig:xxz-1x4}
\end{figure*}

We now compare ClassiFIM and the methods proposed
in Refs.~\cite{vanNieuwenburg2016LearningPT,Huang:22} while focusing on the XXZ300 dataset.
Various technical details are presented in \cref{as:peakrmse}.

First, we observe that the methods of Refs.~\cite{vanNieuwenburg2016LearningPT,Huang:22}
can be applied to our datasets (which contain measurements in a fixed basis):
Ref.~\cite{vanNieuwenburg2016LearningPT} already demonstrated how to apply
the Confusion Scheme they introduced, which we refer to as W,
to the bitstring data by applying it to the classical Ising model.
We discuss the application of the method of Ref.~\cite{Huang:22},
which we refer to as SPCA, to the bitstring data in \cref{as:spca-bitstring}.

To obtain a ground truth result for the XXZ300 dataset, we implemented DMRG,
which gives access to the state in matrix product state (MPS) form;
from here we can extract state overlaps and the fidelity susceptibility
in terms of finite differences (see \cref{ass:xxz300-dataset-gen}).
These procedures result in the phase diagrams shown in \cref{fig:xxz-1x4}.

The most notable aspect is that the phase diagrams \cref{fig:xxz-1x4}(d)-(g)
display a remarkable visual agreement.
Despite this, the phase diagrams have different meanings and cannot be quantitatively compared.
In particular, the W and SPCA plots do not have a rigorous quantitative connection to QPTs,
unlike the ClassiFIM phase diagram. A rigorous quantitative connection is possible using the methods of \cite{arnold2023machine,duy2022fisher}. The former presents a modification of the W method, which provides a lower bound for the FIM, and another method ($I_3$), which also produces an estimate of the FIM; the latter provides a method for estimating the FIM. A comparison with these methods is beyond the scope of this work.

However, all three methods (W, SPCA, and ClassiFIM) are designed to have a phase boundary
at the locations of phase transitions, as is visible for W, SPCA,
and ClassiFIM in \cref{fig:xxz-1x4}(a)-(c), respectively.
Since the peaks in the FIM are known to correspond to an exact phase boundary,
we can use \cref{fig:xxz-1x4}(g) as the ground truth.
We can then compare the peak locations of each method with the ground truth.

In \cref{as:peakrmse} we introduce a PeakRMSE metric that measures the ability of
the methods to predict the locations of phase boundaries
(which appear as peaks in the FIM and W plots).
The computation of this metric requires modifications to all methods,
so that all of them output a list of guesses
of the locations of phase transitions for each 1D slice of the phase diagram
within a given computational budget using the $\mcDtrain$ dataset.
Such modifications were performed in \cite[Sec.~5 and App.~F]{ClassiFIM-ML} for
all three methods. We refer to these modified implementations
as mod-W, mod-SPCA, and mod-ClassiFIM.
The resulting PeakRMSE values are shown in \cref{tab:comparison};
the main finding is that mod-ClassiFIM and mod-SPCA
emerge as the most accurate methods overall.

\begin{table}[t]
  \centering
  \begin{tabular}{llllllll}
    \toprule
    Statistical Manifold  & Method        & PeakRMSE \\
    \midrule
    IsNNN400 & mod-ClassiFIM & $\mathbf{0.017}(4)$\\
             & mod-SPCA      & $\mathbf{0.016}(1)$\\
             & mod-W         & $0.035(7)$ \\
    \midrule
    Hubbard12& mod-ClassiFIM & $\mathbf{0.035}(7)$\\
             & mod-SPCA      & $\mathbf{0.039}(1)$\\
             & mod-W         & $0.076(2)$ \\
   \midrule
    FIL24    & mod-ClassiFIM & $0.028(8)$ \\
             & mod-SPCA      & $\mathbf{0.017}(1)$\\
             & mod-W         & $0.026(4)$\\
    \midrule
    XXZ300   & mod-ClassiFIM & $\mathbf{0.0078}(5)$\\
             & mod-SPCA      & $0.012(2)$\\
             & mod-W         & $0.011(3)$\\
    \bottomrule
  \end{tabular}
  \caption{Comparison with two other methods using peak position accuracy (RMSE). Statistically significant best performance is in bold. Methods mod-ClassiFIM and mod-SPCA emerge as the most accurate among those compared here. With the exception of FIL24, ClassiFIM captures the peak locations on par with or better than mod-W or mod-SPCA.}
  \label{tab:comparison}
\end{table}

\section{Summary and Outlook}
\label{as:limits-and-future}

In this work, we addressed the problem of detecting quantum or classical phase transitions via unsupervised machine learning. We chose to focus on the Fisher information metric due to its relationship to fidelity susceptibility and utilized ClassiFIM~\cite{ClassiFIM-ML}: an ML algorithm that takes classical data as input. This data can be generated numerically or experimentally using a quantum processor. To justify the use of classical data for detecting quantum phase transitions, we established a number of formal results relating the classical and quantum fidelity susceptibilities. Along the way, we also relaxed the differentiability conditions assumed previously to be needed in order for the fidelity susceptibility to be well-defined.

Using a number of models exhibiting quantum (including topological) and classical phase transitions, we demonstrated that ClassiFIM generates accurate phase diagrams and correctly predicts the locations of phase boundaries based on significantly smaller datasets than those required by leading alternative methods, in particular the W and SPCA methods in their original, unmodified form~\cite{vanNieuwenburg2016LearningPT,Huang:22}. The fact that ClassiFIM achieves this based on purely classical data is particularly noteworthy; the datasets required for the W and SPCA methods are significantly more complex to obtain. Moreover, the theorems presented here and in the companion paper~\cite{ClassiFIM-ML} establish rigorous connections to phase transitions, which is not guaranteed for the W and SPCA methods (but is guaranteed, e.g., for the modification of the W method described in Ref.~\cite{arnold2023machine}).

Next, we briefly discuss the limitations of ClassiFIM and directions for future work.

The first and most obvious limitation is related to the fact that the FIM peak positions are known in some cases not to match those of phase transitions in the thermodynamic limit accurately. In such cases, even if the FIM estimates generated by ClassiFIM are accurate, the FIM peak positions are similarly not guaranteed to detect a phase transition accurately.

Second, we considered a neural network with local features. If two distributions are only globally distinguishable, it is possible that ClassiFIM would fail to estimate the underlying FIM, despite its success in the two TQPT cases we studied (Kitaev20 and XXZ300).

A third limitation is that we did not rigorously study the resource scaling of ClassiFIM. In particular, we observed that, regardless of the model, a fixed number of samples per $\vlambda$ was sufficient for ClassiFIM to estimate the FIM accurately. Since the models ranged from $24$ to $300$ qubits, this suggests that the number of samples needed does not grow rapidly with system size, but we currently do not have a rigorous proof or general guarantee. Similarly, we observed that scaling the training time linearly with the size of the samples (i.e., the number of qubits) is sufficient, but do not have a general guarantee of this either.

Future work may address the following three goals. First, ClassiFIM's capacity to characterize QPTs needs to be extrapolated to the thermodynamic limit, which can be done via finite-size scaling. Second, it is desirable to characterize the application of ClassiFIM to data obtained from an actual quantum processor. To do so meaningfully will require hundreds or even thousands of qubits. This may soon be possible with recent progress using quantum annealers and optical lattices. Finally, it remains to be determined whether ClassiFIM is competitive with other methods for solving the FIM estimation task, in particular Refs.~\cite{arnold2023machine, duy2022fisher}.

\begin{acknowledgments}
This material is based upon work supported by the Defense Advanced Research Projects Agency (DARPA) under Agreement No. HR00112190071 and Agreement No. HR001122C0063. NE was supported by the U.S. Department of Energy
(DOE) Computational Science Graduate Fellowship under
Award Number DE-SC0020347. We thank Itay Hen, Utkarsh Mishra, and Lalit Gupta for many useful discussions.
\end{acknowledgments}

\appendix
\section*{Appendix}

Here, we provide additional details and proofs in support of the main text.
We also provide access to the code and data in support of the experiments described in the main text;
see \cite{public-datasets,ClassiFIM-gencode,ClassiFIM-code} for details on how to use the code and download
the data.

\section{Basic quantum background and notation}
\label{as:background}

This section provides a brief introduction to the quantum terminology used throughout our paper for the benefit of readers from different communities.

\subsection{Fidelity}
\label{ass:fidelity}

We use Dirac bra-ket notation and
refer to elements of the Hilbert space with a unit norm, i.e., 
$\norm{\ket{\psi}} = \bk{\psi}{\psi} = 1$, as pure states. The notation
$\bk{\phi}{\psi}$ denotes the inner product between $\ket{\phi}$ and
$\ket{\psi}$, where the bra vector $\bra{\phi}$ is the complex-conjugate
transpose of the ket vector $\ket{\phi}$. Two states $\ket{\psi}$ and
$\ket{\phi}$ are equivalent if they differ only by a global phase, i.e.,
$\ket{\phi} = e^{i\phi} \ket{\psi}$ for some $\phi \in \mathbb{R}$. One can measure the
similarity between two pure states $\ket{\phi}$ and $\ket{\psi}$ by their fidelity, defined as
\begin{equation}
\label{eq:def.F.pure}
F(\ket{\phi}\!, \ket{\psi}) = \abs{\bk{\phi}{\psi}}.
\end{equation}
This measure is always within $[0,1]$ and is unaffected by any global phase difference.

Quantum states, more generally, are density operators (or density matrices) $\rho$, defined as unit-trace, positive semi-definite operators acting on the Hilbert space $\mathcal{H}$. In this formulation, pure states are rank-$1$ projectors, i.e., the set of density matrices of the outer product form $\rho = \ketbra{\phi}$.

The classical fidelity between two discrete probability distributions $p$ and $q$ is defined as:
\begin{equation}
\label{eq:def.Fc}
  F_c(p, q) = \sum_z \sqrt{p_z q_z}.
\end{equation}
The quantum generalization is the Uhlmann fidelity, which also generalizes \cref{eq:def.F.pure}
for the fidelity between two pure states to the fidelity between two mixed states:
\begin{equation}
\label{eq:def.F.mixed}
  F(\rho, \sigma) \equiv \|\sqrt{\rho}\sqrt{\sigma}\|_1 =\Tr\left(\sqrt{\sqrt{\sigma}\rho\sqrt{\sigma}}\right),
\end{equation}
where $\|\cdot\|_1$ denotes the trace norm, i.e., the sum of the singular values.
Uhlmann's theorem~\cite{uhlmann1976transition}
(see also~\cite[Theorem 9.4]{nielsen-chuang-2010})
states that for any two mixed states $\rho$ and $\sigma$ one has
\begin{equation}
\label{eq:uhlmann}
  F(\rho, \sigma) = \max_{\ket{\psi}, \ket{\phi}} \left| \braket{\psi|\phi} \right|,
\end{equation}
where the maximum is taken over all purifications $\ket{\psi}$ and $\ket{\phi}$
on $\mathcal{H}\otimes\mathcal{H}$. In particular, it follows that
\cref{eq:def.F.mixed} is symmetric in $\rho$ and $\sigma$ and
$F(\rho, \sigma) \leq 1$. For a detailed discussion of the properties of the fidelity
see~\cite[Ch.~13]{bengtsson2006geometry}.

\cref{eq:def.F.pure,eq:def.Fc,eq:def.F.mixed} are consistent in the following sense. 
\begin{enumerate}
\item If $\rho = \ketbra{\psi}$ and $\sigma = \ketbra{\phi}$ are pure states,
then $F(\rho, \sigma) = F(\ket{\psi}, \ket{\phi})$.
\item If $\rho = \diag(p)$ and $\sigma=\diag(q)$ are diagonal in the same basis,
then $F(\rho, \sigma) = F_c(p, q)$.
\item If $\ket{\psi} = \sum_x \sqrt{p_x} \ket{x}$ and $\ket{\phi} = \sum_x \sqrt{q_x} \ket{x}$ have non-negative amplitudes in the same basis, then $F(\ket{\psi}, \ket{\phi}) = F_c(p, q)$.
\item In general, if $p_x = \rho_{xx}$, $q_x = \sigma_{xx}$ (in some basis $\{\ket{x}\}$),
then $0 \leq F(\rho, \sigma) \leq F_c(p, q) \leq 1$, $F(\rho, \sigma) = 1$ iff $\rho = \sigma$, and $F(\rho, \sigma) = 0$ iff $\rho\sigma = 0$.
\end{enumerate}

\subsection{Hamiltonians and observables}

A quantum system is characterized by its Hamiltonian $H_\lambda$:
a $d\times d$ Hermitian matrix depending on a parameter or set of parameters $\lambda$.
Most of the systems we consider in this work contain $n$ degrees of freedom that are either fermions or spin-$1/2$ particles (qubits).
Each of the $n$ orbitals or spins is characterized by a binary variable:
an orbital occupation number $c_i^\dag c_i$ (where $c_i$ and $c_i^\dagger$ are fermionic annihilation and creation operators, respectively)
or the eigenvalue of a spin operator, e.g., $\sigma^z_i$ (one of the three Pauli matrices).
Here, $c_i^\dag c_i$ and $\sigma^z_i$ are also represented by $d\times d$ Hermitian matrices, with $d = 2^n$,
and the notation implies a tensor product with the identity matrix on the remaining $n-1$ operators with index not equal to $i$.
The ground state of the Hamiltonian $H$ is its eigenvector $|\phi\rangle =\sum_z \phi_z |z\rangle$ with the lowest eigenvalue.
Here we have chosen the ``computational basis'' $\{|z\rangle\}$ of the simultaneous eigenvectors of all the $c_i^\dag c_i$ or $\sigma^z_i$. The vector components $\{\phi_z\}$ are called ``probability amplitudes'', or just ``amplitudes'', and the set $\{|\phi_z|^2\}$ forms a probability distribution with normalization $\sum_z |\phi_z|^2=1$.

An observable $O$ can be any $d\times d$ Hermitian matrix. Its expectation value is $\langle \phi| O |\phi\rangle =\sum_{z,z'} \phi_z^* O_{zz'}\phi_z$. The matrix components $O_{zz'} = \bra{z}O\ket{z'}$ are the ``matrix elements'' of $O$ in the basis $\{|z\rangle\}$, i.e., $O = \sum_{z,z'} O_{zz'} \ket{z}\!\!\bra{z'}$.

\subsection{Quantum phase transitions}

The textbook definition of a phase transition~\cite{Sachdev:book} concerns a Hamiltonian with a simple analytic dependence on a parameter $\lambda$, e.g., $H = H_0 + \lambda H_1$.
A QPT occurs if the ground state energy is non-analytic at some point $\lambda_c$.
A first-order QPT is characterized by a finite discontinuity in the first derivative of the ground state energy with respect to $\lambda$, in the thermodynamic limit $n\to\infty$.
A second-order (or continuous) QPT is similarly characterized by a finite discontinuity, or divergence, in the second derivative of the ground state energy, assuming the first derivative is continuous.
These characterizations are the zero temperature limits of the classical definition of the corresponding phase transitions, given in terms of the free energy.
Finding the ground state of a quantum system is a hard problem in general, but there are various methods that can result in good approximations in many cases of practical interest.
Quantum annealing~\cite{Hauke:2019aa} or adiabatic quantum computing~\cite{Albash-Lidar:RMP} is one such method that can be used both in the setting of gate-model quantum computers and via specialized quantum annealers~\cite{king2018observation,b2020probing,Weinberg:2020aa,Nishimura2020,king2022coherent,King:22}.

An order parameter is a quantity that takes on different values in different phases and typically becomes non-zero when a system undergoes a QPT to an ordered phase.
It is a measure of the degree of order across the transition.
Thus, an alternative characterization of QPTs is that a QPT occurs when the order parameter of a quantum system becomes discontinuous or singular in the thermodynamic limit.
I.e., a QPT occurs when $\lim_{n\to\infty}\langle \phi| O|\phi\rangle$ or its derivatives, for some observable $O$ (whose expectation value is an order parameter) is discontinuous with $\lambda$.

For example, the transverse field Ising model $H =\sum_{i=1}^n  \lambda\sigma^x_i -\sigma^z_i \sigma^z_{i+1}$ undergoes a phase transition at $\lambda=1$:
the ground state is a ferromagnet for $\lambda <1$ with a non-zero order parameter $O =(\sum_i \sigma^z_i)^2$.
The square here indicates that we are interested in the absolute value of the magnetization $\sum_i \sigma^z_i$, while its sign selects between nearly degenerate ground states in the ferromagnetic phase, consisting of spins that are all aligned either up or down.
The magnetization is a \emph{local} order parameter, constructed as an average of local terms $\sigma^z_i$.
In the presence of a small symmetry-breaking perturbation $\epsilon \sigma^z_1$, the magnetization order parameter can be deduced by measuring a single spin locally and taking the limit $\epsilon \to 0$.

In contrast to the transverse field Ising model, topological phases of matter do not possess a local order parameter.
Instead, a non-local string-like operator needs to be used to distinguish between approximately degenerate ground states.
However, to detect a TQPT, it is often enough to follow the definition above and consider the ground state energy.
One of the prototypical examples of topological order---the toric code~\cite{Kitaev:97}---exhibits a TQPT in the presence of an external magnetic field $\sum_{i=1}^n  \lambda\sigma^x_i$.
The discontinuity at this TQPT was demonstrated in terms of both a non-local (string-like) order parameter~\cite[Fig.~4]{hamma2008entanglement} and a local quantity:
the second derivative of the ground state energy~\cite[Fig.~3(a)]{hamma2008entanglement}.
It is an open question whether there are realistic systems where a TQPT is fully hidden from local observables such as the energy.

Generally, the energy may not be accessible, the order parameter may not be known, or there are multiple competing order parameters.
This situation motivates our approach:
it suggests the need for generic indicators of phase transitions, such as the FIM derived from the probability distribution $p_z = \langle \phi| P_z|\phi\rangle$, where $P_z$ is a projector on the eigenvalues of $c_i^\dag c_i$ or $\sigma^z_i$ corresponding to the bitstring $z$ (not to be confused with the $z$-axis, the superscript in $\sigma^z$).

Finally, we note that another type of PTs, called dynamical phase transitions, can, in principle, be hidden from the ground state entirely, in terms of both local and global features, as they characterize the excitation spectrum, not just a single quantum state.
Nevertheless, as we demonstrate in our Hubbard model results in \cref{sec:phase-diagrams}, a textbook example of a dynamical phase transition---a metal-insulator transition \cite{kyung2006mott,turkowski2021one,arovas2022hubbard,gull2013superconductivity,tanaka2019metal,terletska2011quantum}---turns out to be accompanied by changes in the ground state itself.

In contrast to the thermodynamic limit above, the FIM-estimation task (\cref{sec:task}) is set in the context of a finite system of fixed size.
Such a finite-size system may have features hinting at the existence of a phase transitions:
e.g., peaks of the FIM at locations in the phase diagram close to where an infinite-size system experiences phase transitions.
Such features are called finite-size precursors of phase transitions,
but, for simplicity, throughout this work, we often refer to the locations of these finite-size precursors simply as phase transitions.

\section{Fisher Information Metric}
\label{as:fim}

We introduced the Fisher Information Metric (FIM) in \cref{eq:fim}.
Here we clarify the conditions needed for the FIM to be well-defined.

Given the score function formula $s_{\mu}(x;\vlambda) = \partial{\lambda_{\mu}} \log P_{\vlambda}(x)$ [\cref{eq:sff}], four conditions are naively required in order for the FIM $g_{\mu\nu}(\vlambda)$ to be well-defined:
(i) the sample space should form a disjoint union of one or more manifolds;
(ii) for each $\vlambda$ the probability distribution should be absolutely
continuous on each of those manifolds, so that the density function
$P_{\vlambda}(\bullet)$ is well-defined;
(iii) $P_{\vlambda}(x) \neq 0$ for all $x$;
(iv) $P_{\vlambda}(x)$ is differentiable with respect to $\vlambda$ for all $x$.
However, a more general treatment of the FIM is possible, which does not require
any of these conditions, at the cost of making the proofs more cumbersome. 

Assume that we have a manifold $\mcM$ of parameters $\vlambda$ and a measurable space
$(\Omega, \mathcal{F})$, where $\Omega$ is the set of all possible samples, and
$\mathcal{F}$ is a sigma-algebra on that space. Furthermore, assume that for every
$\vlambda\in\mcM$ we have a probability measure $\sigma_{\vlambda}$.
The classical fidelity (also known as the Bhattacharyya distance) between two probability measures $\sigma$ and $\sigma'$ is
defined as the following generalization of \cref{eq:def.Fc}:\footnote{Some sources define $F_c$ to be the square of r.h.s. of \cref{eq:Fc.def} instead, e.g., \cite[Def.~9.2.4]{wilde2011classical}.}
\begin{equation}
\label{eq:Fc.def}
F_c(\sigma, \sigma') = \int_{\Omega} \sqrt{d\sigma d\sigma'}.
\end{equation}
We note that $F_c$ is always well-defined, where the integral on the r.h.s. is defined using the Radon-Nikodym derivative with respect to any measure with respect to which both $\sigma$ and $\sigma'$ are absolutely continuous (e.g., $\sigma + \sigma'$). When $F_c(\sigma_{\vlambda_0}, \sigma_{\vlambda_0 + \vdlambda})$ has the form
\begin{equation}
\label{eq:Fc.Taylor}
F_c(\sigma_{\vlambda_0}, \sigma_{\vlambda_0 + \vdlambda}) =
1 + \frac{1}{8} \sum_{\mu,\nu} g_{\mu\nu}(\vlambda_0)
\delta\lambda_{\mu} \delta\lambda_{\nu}
+ o(\norm{\vdlambda}^2)
\end{equation}
for some $g_{\mu\nu}(\vlambda_0)$,
we say that the FIM is well-defined. In that case, we can always make it symmetric with respect to the interchange of indices $\mu$ and $\nu$ by replacing it with $(g_{\mu\nu}(\vlambda_0) + g_{\mu\nu}(\vlambda_0))/2$. The resulting  symmetric $g_{\mu\nu}(\vlambda_0)$ is the FIM.

In fact, there are multiple inequivalent choices for a more general treatment, that become equivalent when the four conditions above are all satisfied.
These choices are based on the fact that, as in \cref{eq:Fc.Taylor}, the FIM
coincides (under certain regularity conditions,
up to a constant factor) with the first non-trivial term in the
Taylor expansion of multiple measures of distance or similarity between probability
distributions, including the Kullback-Leibler
divergence, Jensen-Shannon divergence, and Bhattacharyya distance (our choice)~\cite[Sec.~3.5]{amari2016information}.

\section{Proofs of
\texorpdfstring{\cref{th:comput.basis} and \cref{th:avggc}}{Theorems \ref*{th:comput.basis} and \ref*{th:avggc}}}
\label{ass:two.theorems}

For the following two lemmas we fix an orthonormal basis $\{\ket{z}\}_{z\in\mcS}$
(the ``computational basis'')
for the finite-dimensional Hilbert space $\mcH$ and write pure states and operators
as vectors and matrices in that basis. $\mcM$ denotes a manifold.
\begin{mylemma}
  \label{lm:thm1.a}
  Let $H(\vlambda')$ be a continuously differentiable family of Hamiltonians
  with real matrix elements and a non-degenerate ground state
  defined in the neighborhood of $\vlambda \in \mcM$.
  Then its ground state can be chosen to have real amplitudes
  and to be continuously differentiable in some neighborhood of
  $\vlambda$.
\end{mylemma}
We note that the ``real matrix elements'' condition of this lemma is commonly
satisfied for many Hamiltonians of interest. For example, if one has a system
of spins with one-body and two-body interactions, then there are nine possible
types of terms in such interactions, written in terms of the Pauli matrices as:
$X$, $Y$, $Z$, $XX$, $XY$, $XZ$, $YY$, $YZ$, $ZZ$. Six of these have real
matrix elements, i.e., if such a Hamiltonian only uses the
$X$, $Z$, $XX$, $XZ$, $YY$, and $ZZ$ terms,
then it has real matrix elements. For example, the Hamiltonians currently implemented in the D-Wave quantum annealing devices
only include the $X$, $Z$, and $ZZ$ terms~\cite{King:22}, as long as one ignores additional terms that arise due to geometric phases~\cite{Vinci:2017aa}.

\begin{proof}
  First, let us fix any $\vlambda'$ and show that the ground state can be picked
  to have real amplitudes. Indeed, let $\ket{\psi}$ be any ground state
  with ground state energy $E_0$. We can represent $\ket{\psi}$ as
  \begin{equation}
    \ket{\psi} = \cos(\theta) \ket{\psi_0} + i\sin(\theta) \ket{\psi_1},
  \end{equation}
  where both $\ket{\psi_0}$ and $\ket{\psi_1}$ are normalized and have
  real amplitudes. We have
  \begin{align}
  \label{eq:C2}
    E_0 &= \braket{\psi|H(\lambda)|\psi}\notag\\
    &= \cos^2(\theta) \braket{\psi_0|H(\lambda)|\psi_0}
      + \sin^2(\theta) \braket{\psi_1|H(\lambda)|\psi_1} \notag\\
    &\geq E_0 (\cos^2(\theta) + \sin^2(\theta))
    = E_0 ,
  \end{align}
where the inequality is due to the fact that $E_0$ is the ground state of $H(\lambda)$. Since \cref{eq:C2} must be an equality, we have 
\bes
  \begin{align}
    &\braket{\psi_0|H(\lambda)|\psi_0} = E_0
      \textrm{ or } \cos(\theta) = 0\\
   & \braket{\psi_1|H(\lambda)|\psi_1} = E_0
      \textrm{ or } \sin(\theta) = 0.
 \end{align}
 \ees  
 Thus, at least one of $\psi_0$ or $\psi_1$ is the ground state
  with real amplitudes.

  Now let us show that the ground state can be picked to be continuously
  differentiable. First, applying the implicit function theorem to the
  equation $\det(H(\vlambda') - E_0(\vlambda')) = 0$ yields that
  $E_0(\vlambda')$ is continuously differentiable. Then,
  applying the implicit function theorem to the system of equations obtained
  from $(H(\vlambda') - E_0(\vlambda')) \ket{\psi(\vlambda')}=0$ by replacing
  one of the redundant rows with the equation
  \begin{equation}
    \norm{\ket{\psi(\vlambda')}}^2 = 1,
  \end{equation}
  we find that $\ket{\psi(\vlambda')}$ can indeed be chosen to be
  continuously differentiable. Indeed, this is a system of $\dim(\mcH)$ real equations with $\dim(\mcH)$ real unknowns, where the l.h.s. and r.h.s. are differentiable and the Jacobian is invertible (since the ground state is non-degenerate). Hence, by the implicit function theorem the resulting $\ket{\psi(\vlambda')}$ is continuously differentiable. Note that in general, the ground state $\ket{\psi(\vlambda')}$ need not have real coefficients. The use of the implicit function theorem above rules out the cases when the global phase makes $\ket{\psi(\vlambda')}$ non-continuously differentiable. Even under the additional constraint of $\ket{\psi(\vlambda')}$ being real, the sign can arbitrarily change at some $\vlambda$ making it non-continuously differentiable.
\end{proof}

\begin{mylemma}
  \label{lm:thm1.b}
  Let $\ket{\psi(\vlambda')}$ be a family
  of pure states with real amplitudes defined in the neighborhood
  of $\vlambda \in \mcM$ and differentiable at $\vlambda' = \vlambda$.
  Let $p(\vlambda')$ be the associated probability distribution, obtained after measurement of $\ket{\psi(\vlambda')}$ in the computational basis.
  Then
  \begin{equation}
    \label{eq:lm.thm1.b}
    \chi_{F_c}(\vlambda) = \chi_F(\vlambda).
  \end{equation}
\end{mylemma}
\begin{proof}
  First, note that both sides of \cref{eq:lm.thm1.b} are well-defined.
  Indeed, $p_z(\vlambda') = \abs{\braket{z|\psi(\vlambda')}}^2$
  is differentiable at $\vlambda$ and twice differentiable if $p_z(\vlambda) = 0$.

  Because both sides of \cref{eq:lm.thm1.b} are equivariant with respect to a change
  of coordinates $\vlambda'$, it is sufficient to show the equality for the $00$-th
  component of both. Using \cref{eq:explicit.chi.pure} we compute
  \begin{equation}
    (\chi_F(\vlambda))_{00}
    = \real(\braket{\partial_0\psi(\vlambda)|\partial_0\psi(\vlambda)})
      - 0 \cdot 0
    = \norm{\ket{\partial_0\psi(\vlambda)}}^2.
  \end{equation}
  Using \cref{eq:explicit.chifc} we compute
  \bes
  \begin{align}
    (\chi_{F_c}(\vlambda))_{00}
   & = \sum_{z \in \mcS_{+}} \left.
        \frac{
          \left(\partial_{0} \braket{z|\psi(\vlambda')}^2\right)^2
        }{
          4 \braket{z|\psi(\vlambda')}^2
        } \right|_{\vlambda' = \vlambda}\\
&\quad      + \sum_{z\in\mcS_0}
        \frac12 \left.\partial_{0}^2
        \braket{z|\psi(\vlambda')}^2\right|_{\vlambda' = \vlambda}\notag \\
    &= \sum_{z \in \mcS} \braket{z|\partial_0\psi(\vlambda)}^2
    = \norm{\ket{\partial_0\psi(\vlambda)}}^2.
  \end{align}
  \ees
\end{proof}

\cref{th:comput.basis} now
follows directly from \cref{lm:thm1.a} and \cref{lm:thm1.b}.
Next, we prove \cref{th:avggc}

\begin{proof}[Proof of \cref{th:avggc}]
  Let us fix $\vlambda \in \mcM$ and evaluate both sides of
\cref{eq:Emeas}
  at that point using \cref{eq:explicit.chifc,eq:explicit.chi.pure}.
  The second term on the r.h.s. of \cref{eq:explicit.chifc} is only relevant when one of the
  vectors in the measurement basis is orthogonal to $\ket{\psi(\vlambda)}$
  --- a subset of measure $0$ in the space of all measurement bases.
  Thus, we can safely discard it in \cref{eq:Emeas}.
  Since both sides of \cref{eq:Emeas}
  are equivariant with respect to the choice of coordinates $\vlambda$,
  it is sufficient to prove the equality for the $00$-th component of both.
  Thus, it remains to show that
  \begin{align}
    \mathbb{E}_{\textnormal{meas}}
    \sum_{z' \in \mcS'}
        \frac{
          \left(\partial_{0} p'_{z'}(\vlambda)\right)^2
        }{4 p'_{z'}(\vlambda)} 
    &= \frac12 \norm{\ket{\partial_{0} \psi(\vlambda)}}^2 \notag \\
        &\ \ - \frac12 \abs{\braket{\partial_{0} \psi(\vlambda) | \psi(\vlambda)}}^2.
      \label{eq:avggc.1}
\end{align}
  One can simplify the l.h.s. by noting that the expectation of each term
  in the sum is the same. Thus, it remains to average over $\ket{\varphi}$
  s.t. $\norm{\ket{\varphi}} = 1$:
  \begin{equation}
    \label{eq:avggc.2}
    \textnormal{l.h.s. \cref{eq:avggc.1}}
    = \dim(\mcH) \mathbb{E}_{\ket{\varphi}}
      \frac{
        \left(\real\left(
        \braket{\varphi|\partial_{0} \psi(\vlambda)}
        \braket{\psi(\vlambda)|\varphi}\right)\right)^2
      }{\abs{\braket{\varphi|\psi(\vlambda)}}^2}.
  \end{equation}
  Let us denote
    $\ket{\psi_0} = \ket{\psi(\vlambda)}$,
    $\psi_\parallel = \braket{\psi_0|\partial_0\psi(\vlambda)}$,
    $\ket{\psi_\perp} = \ket{\partial_0\psi(\vlambda)} - \psi_\parallel\ket{\psi_0}$,
    $\varphi_\parallel = \braket{\psi_0|\varphi}$,
    $\ket{\varphi_\perp} = \ket{\varphi} - \varphi_\parallel\ket{\psi_0}$.
    Note that $\real(\psi_\parallel) = 0$.  With this notation
  \begin{equation}
    \textnormal{r.h.s. \cref{eq:avggc.1}} = \frac12 \norm{\ket{\psi_\perp}}^2,
  \end{equation}
  and
  \begin{equation}
    \label{eq:avggc.5}
    \braket{\varphi|\partial_{0} \psi(\vlambda)} \braket{\psi(\vlambda)|\varphi} =
    \abs{\varphi_\parallel}^2 \psi_\parallel
    + \varphi_\parallel \braket{\varphi_\perp|\psi_\perp}.
  \end{equation}
  Substituting \cref{eq:avggc.5} into \cref{eq:avggc.2}
  and averaging over the arguments of $\varphi_\parallel$
  we obtain
  \begin{equation}
    \label{eq:avggc.6}
    \textrm{l.h.s. \cref{eq:avggc.1}}
    = \frac{\dim(\mcH)}{2}\mathbb{E}_{\ket{\varphi}}
      \abs{\braket{\varphi_\perp|\psi_\perp}}^2
    = \frac12 \norm{\psi_\perp}^2.
  \end{equation}
\end{proof}

\section{Proof of
\texorpdfstring{\cref{lm:explicit.chi-1}}{Theorem \ref*{lm:explicit.chi-1}}}
\label{app:proof-lm:explicit.chi-1}

\begin{proof}[Proof of \cref{lm:explicit.chi-1}]
  Note that \cref{lm:explicit.chi-2} is a generalization of \cref{lm:explicit.chi-1}, and since the proof of the former does not rely on the latter, the reader may choose to skip the proof of \cref{lm:explicit.chi-1}. However, this proof is simpler, so we present it here for clarity. 

Due to the equivariance of the definition of $\chi_{F}$ with respect to
  a coordinates change $\vlambda' \mapsto \vlambda' - \vlambda$,
 we can prove the statements in the lemma for
  $\vlambda = \veczero$  without loss of generality. Let $\ket{\psi_0} = \ket{\psi(\veczero)}$,
  $\ket{\delta \psi} = \ket{\psi(\vlambda')} - \ket{\psi_0}$.
  We know that
  \begin{equation}
    \label{eq:explicit.chi:part1:1}
    \ket{\delta \psi} = \sum_{\mu} \ket{\partial_\mu\psi(\veczero)} \vlambda_{\mu}
      + \ket{r},
  \end{equation}
  where $\ket{r} = o(\abs{\vlambda'})$.
  We also know that $\braket{\psi(\vlambda')|\psi(\vlambda')} = 1$.
  On the other hand,
  \bes
    \label{eq:explicit.chi:part1:2}
  \begin{align}
    \braket{\psi(\vlambda')|\psi(\vlambda')} &=
    1 + 2\real{\braket{\psi_0|\delta\psi}} + \braket{\delta\psi|\delta\psi}\\
    &= 1 + 2\sum_\mu\real{\braket{\psi_0|\partial_\mu\psi(\veczero)}}\lambda'_{\mu}
      + 2\real{\braket{\psi_0|r}}\notag \\
      &\quad + \sum_{\mu,\nu}\braket{\partial_\mu\psi(\veczero)|\partial_\nu\psi(\veczero)}
        \vlambda'_{\mu} \vlambda'_{\nu}
      + o(\abs{\vlambda'}^2).
  \end{align}
  \ees
  Thus,
  \begin{equation}
    \label{eq:explicit.chi:part1:3}
    \real{\braket{\psi_0|\partial_\mu\psi(\veczero)}} = 0
  \end{equation}
  and
  \begin{equation}
    \label{eq:explicit.chi:part1:4}
    \real{\braket{\psi_0|r}}
      = - \frac12 \sum_{\mu,\nu}
        \braket{\partial_\mu\psi(\veczero)|\partial_\nu\psi(\veczero)}
        \vlambda'_{\mu} \vlambda'_{\nu}
    + o(\abs{\vlambda'}^2).
  \end{equation}
    Note that for real $x,y$ around $x=y=0$ we have
  \begin{equation}
    \abs{1 + x + iy} = 1 + x + \frac{y^2}{2} + O(x^2 + y^4).
  \end{equation}
  Using this fact along with \cref{eq:explicit.chi:part1:1,eq:explicit.chi:part1:3,eq:explicit.chi:part1:4},
  we compute \cref{eq:def.F.pure}:
 \bes
    \label{eq:explicit.chi:part1:5}
  \begin{align}
    F(\veczero, \vlambda') &= \abs{\braket{\psi(\veczero)|\psi(\vlambda')}}
    = \abs{1 + \braket{\psi_0|\delta \psi}}\\
   & = \abs{1 + \sum_\mu\braket{\psi_0|\partial_\mu \psi(\veczero)}\vlambda'_\mu
      + \braket{\psi_0|r}} \\
    &= 1 + \real{\braket{\psi_0|r}} - \frac12 \left(
    \sum_\mu\braket{\psi_0|\partial_\mu \psi(\veczero)}\vlambda'_\mu\right)^2\notag\\
     &\quad + o(\abs{\vlambda'}^2) \\
   & = 1 - \frac12 \sum_{\mu\nu} (\chi_{F}(\veczero))_{\mu\nu} \vlambda'_\mu \vlambda'_\nu
      + o(\abs{\vlambda'}^2),
  \end{align}
  \ees
  where $(\chi_{F}(\veczero))_{\mu\nu}$ is given by \cref{eq:explicit.chi.pure}.
\end{proof}

\section{Proof of
\texorpdfstring{\cref{lm:explicit.chi-2}}{Theorem \ref*{lm:explicit.chi-2}} and of \cref{lm:chif.mixed.bound}}

\subsection{Proof of \texorpdfstring{\cref{lm:explicit.chi-2}}{Theorem \ref*{lm:explicit.chi-2}}}
\label{app:proof-lm:explicit.chi-2}

The main idea of our approach in proving \cref{lm:explicit.chi-2} is to use the residue formula for the
action of a holomorphic function $f$ on a bounded linear operator $A$:
\begin{equation}
    \label{eq:holomorphic}
    f(A) = -\frac{1}{2\pi i} \oint f(z) (A - z)^{-1} dz.
\end{equation}
Here, the function $f$ should be defined and holomorphic on the boundary
and inside the contour, and the spectrum of the operator
$A$ should be inside the integration contour
[in the proof below we apply \cref{eq:holomorphic} to the function $f(z) = \sqrt{z}$].
As a reminder, the residue
theorem states that if a function $g(z)$ is holomorphic and defined
on the integration contour and inside it except for finitely many points
$z_1, \dots, z_n$ inside the contour, then
\begin{equation}
    \label{eq:res.theorem}
    \frac{1}{2\pi i} \oint g(z) dz = \sum_{j=1}^{n} \Res(g, z_k), 
\end{equation}
where $\Res(g, z_k)$ is the coefficient $c_{-1}$ in the Laurent series expansion
of $g(z)$ around the point $z=z_k$:
\begin{equation}
    \label{eq:laurent}
    g(z) = \sum_{l=-\infty}^{\infty} c_{l} (z - z_k)^{l}.
\end{equation}
If $g(z) = \tilde g(z) (z - z_k)^{-m-1}$ for $m \geq 0$ where $\tilde g(z)$
does not have a singularity at $z=z_k$ then
\begin{equation}
  \label{eq:res.explicit}
  \Res(g, z_k)
  = \frac{1}{m!} \left.\frac{\partial^m \tilde g(z)}{\partial z^m}\right|_{z = z_k}.
\end{equation}

\begin{proof}[Proof of \cref{lm:explicit.chi-2}]
  As in the proof of \cref{lm:explicit.chi-1}, without loss of generality set $\vlambda = 0$.
  Denote $\rho_0 = \rho(\veczero)$, $\delta \rho = \rho(\vlambda') - \rho_0$.
  We know that
  \begin{equation}
    \delta \rho = \sum_{\mu} \vlambda'_\mu \partial_\mu\rho(\veczero) + r,
  \end{equation}
  where $r = o(\abs{\vlambda'})$. We also know that
  \begin{equation}
    \Tr\left(\partial_\mu\rho(0)\right) = 0, \qquad \Tr(r) = 0.
  \end{equation}
  From the definition,
  \begin{equation}
    F(\rho_0, \rho_0 + \delta \rho) =
    \Tr \sqrt{\rho_0^2 + \sqrt{\rho_0} \delta \rho \sqrt{\rho_0}}.
  \end{equation}
  Let $\rho_{0{+}} = P_{{+}} \rho_0 P_{{+}}^{\dagger}$,
  $\delta\rho_{+} = P_{{+}} \delta\rho P_{{+}}^{\dagger}$,
  $r_{+} = P_{{+}} r P_{{+}}^{\dagger}$. One can see that the expression under the
  square root acts nontrivially only on $\rho_0(\mathcal{H})$, hence the trace
  can be computed in that subspace:
  \begin{equation}
    \label{eq:F.Tr.sqrt}
    F(\rho_0, \rho_0 + \delta \rho) = \Tr \sqrt{
      \rho_{0{+}}^2 + \sqrt{\rho_{0{+}}} \delta \rho_{+} \sqrt{\rho_{0{+}}}
    }.
  \end{equation}
  For $\vlambda' = \veczero$ the expression under the square root is equal to
  $\rho_{0{+}}^2$ and has only positive eigenvalues.
  Thus, for $\vlambda'$ in some neighborhood of $\veczero$ the spectrum of
  the expression under the square root lies in $(c_1, c_2)$
  for some $c_1,c_2$ satisfying $0 < c_1 \leq c_2 < \infty$.
  In that neighborhood the square root is
  an analytic function and can be expressed as an integral with the
  corresponding resolvent over a contour surrounding $[c_1, c_2]$:
\bes
      \label{eq:oint}
\begin{align}
       \label{eq:oint-1}
   &\sqrt{
      \rho_{0{+}}^2 + \sqrt{\rho_{0{+}}} \delta \rho_{+} \sqrt{\rho_{0{+}}}
    }\\
    &\qquad= \frac{-1}{2\pi i} \oint \sqrt{z} \left(
      \rho_{0{+}}^2 + \sqrt{\rho_{0{+}}} \delta \rho_{+} \sqrt{\rho_{0{+}}} - z
    \right)^{-1} dz \notag\\
      \label{eq:oint-2}
    &\qquad = I_0 + I_1 + I_2 + o(\abs{\vlambda'}^2),
  \end{align}
  \ees
  where 
  \begin{align}
    I_l &= \frac{(-1)^{l+1}}{2\pi i} \oint \sqrt{z}
      \left(\rho_{0{+}}^2 - z\right)^{-1}\times \notag \\
      &\qquad \left(
        \sqrt{\rho_{0{+}}} \delta \rho_{+} \sqrt{\rho_{0{+}}} \left(\rho_{0{+}}^2 - z\right)^{-1}
      \right)^{l} dz.
  \end{align}
  In \cref{eq:oint-2} we used the series expansion $(a + b)^{-1} = (-1)^l a^{-1} \sum_{l=0}^{\infty} (ba^{-1})^{l}$ (which holds when $\norm{ba^{-1}} < 1$) applied to $a = \rho_{0{+}}^2 - z$ and $b = \sqrt{\rho_{0{+}}} \delta \rho_{+} \sqrt{\rho_{0{+}}}$.
  In order to evaluate the metric $\chi_F(\veczero)$, we only need to compute
  the diagonal elements of $I_0, I_1, I_2$ while discarding any terms of order
  $o(\abs{\vlambda'}^2)$.
  We pick the basis where $\rho_{0{+}}$ is diagonal with diagonal elements
  $\xi_1 \geq \xi_1 \geq \dots \geq \xi_{n_{+}} > 0$. Computing $I_0$ and $(I_1)_{jj}$ using \cref{eq:res.explicit} we obtain
  \bes
  \begin{align}
    I_0 &= \rho_{0{+}},\\
    \label{eq:I1jj}
    (I_1)_{jj} &= \frac12 (\delta \rho_{+})_{jj}\notag \\
    &= \frac12 \sum_{\mu} \vlambda'_\mu \partial_{\mu} (\rho_{+}(\veczero))_{jj}
      + \frac12 (r_{+})_{jj}.
  \end{align}
  \ees
  To evaluate the diagonal entries of $I_2$ we note that for the contour
  \begin{center}
    \begin{tikzpicture}
      \tikzstyle{bullet}=[circle, fill,minimum size=4pt, inner sep=0pt,
        outer sep=0pt];
      \draw (0, 0) node[style=bullet,label=90:$0$] {};
      \draw (2, 0) node[style=bullet,label=0:$a$] {};
      \draw (3, 0) node[style=bullet,label=0:$b$] {};
      \draw[-latex] (2.5, 0.5) arc (-270:90:1.5 and 0.5);
    \end{tikzpicture}
  \end{center}
  where $a\le b$ are positive real numbers, we have
  \begin{equation}
    \label{eq:oint.aab}
    \frac1{2\pi i} \oint \frac{\sqrt{z}}{(z-a)^2(z-b)} dz = -\frac{1}{2\sqrt{a}(\sqrt{a} + \sqrt{b})^2}.
  \end{equation}
  One can derive \cref{eq:oint.aab} using \cref{eq:res.explicit}.
  We can now apply \cref{eq:oint.aab} to the summands of
  $(I_2)_{jj}$ (where $a = \xi_j^2$, $b = \xi_k^2$), and evaluate $(I_2)_{jj}$:
\bes
    \label{eq:I2jj}
  \begin{align}
    &(I_2)_{jj} = -\sum_{k=1}^{n_{+}} \frac{\xi_j \delta \rho_{{+},jk}
      \xi_k \delta \rho_{{+},kj}}{2\xi_j (\xi_j + \xi_k)^2}
    = -\sum_{k=1}^{n_{+}} \frac{\abs{\delta \rho_{{+},jk}}^2 \xi_k}{2(\xi_j + \xi_k)^2} \\
   &\quad = -\sum_{k=1}^{n_{+}} \sum_{\mu,\nu} \frac{
        \real\left(
          (\partial_\mu\rho_{+}(0))_{jk}(\partial_\nu\rho_{+}(0))_{kj}
        \right) \xi_k
      }{2(\xi_j + \xi_k)^2} \vlambda'_\mu\vlambda'_\nu\notag\\
&\qquad      + o(\abs{\vlambda'}^2).
  \end{align}
  \ees
  Now we are ready to evaluate
  \begin{equation}
    \label{eq:F=I0+I1+I2}
    F(\rho_0, \rho_0 + \delta \rho) =
      \Tr(I_0 + I_1 + I_2) + o(\abs{\vlambda'}^2),
  \end{equation}
  where the trace on the r.h.s. is taken over $\rho(\veczero)\mcH \subset \mcH$.
  Note that
  \begin{equation}
    \label{eq:TrI0}
    \Tr(I_0) = \Tr(\rho_0) = 1.
  \end{equation}
  The first term of the r.h.s. of \cref{eq:I1jj} cannot have a non-zero contribution
  to $\Tr(I_1)$ due to the fact that $\rho(\vlambda')$ is non-negative and
  has $\Tr(\rho(\vlambda')) = 1$. For the second term, notice that $\Tr(r) = 0$,
  hence $\Tr(r_{+}) + \Tr(P_0 r P_0) = 0$, giving
  \begin{equation}
    \label{eq:TrI1}
    \Tr(I_1) = - \Tr(P_0 r P_0^{\dagger}) / 2 = - \Tr(P_0 \rho(\vlambda') P_0^{\dagger}) / 2,
  \end{equation}
  where in the second equality we observed that  $\Tr(P_{0}\partial_\mu\rho(\veczero)P_{0}^\dagger) = 0$
  [otherwise $\rho(\vlambda')$ cannot be
  positive semidefinite in the neighborhood of $\vlambda'=\veczero$].
  Substituting \cref{eq:TrI0,eq:TrI1,eq:I2jj} into \cref{eq:F=I0+I1+I2}
  and comparing the result with 
  \begin{align}
\label{eq:def.chi.Fmixed}
&  F\left(\rho(\vlambda), \rho(\vlambda')\right) = \\
 & \     1 - \frac{1}{2} \sum_{\mu,\nu} (\vlambda' - \vlambda)_{\mu} (\vlambda' - \vlambda)_{\nu} \left(\chi_{F}(\vlambda)\right)_{\mu\nu} + o\left(\abs{\vlambda' - \vlambda}^2\right)  \notag ,
\end{align}
  we obtain
  \begin{align}
(\chi_{F}(\veczero))_{\mu\nu}
    &= \sum_{j,k} \frac{
        \real\left(
          (\partial_\mu\rho_{+}(\veczero))_{jk}
          (\partial_\nu\rho_{+}(\veczero))_{kj}
        \right) \xi_k
      }{(\xi_j + \xi_k)^2}\notag \\
    \label{eq:explicit.chi:part2:17}
    &\qquad + \frac12 \left.\partial_\mu \partial_\nu \Tr(P_0 \rho(\vlambda')
    P_0^{\dagger}) \right|_{\vlambda'=\veczero}.
\end{align}
  Averaging \cref{eq:explicit.chi:part2:17}
  with itself with $j$ and $k$ interchanged, we obtain \cref{eq:explicit.chi.mixed}.

\end{proof}

\subsection{Proof of \texorpdfstring{\cref{lm:chif.mixed.bound}}{Prop \ref*{lm:chif.mixed.bound}}}
\label{app:proof-lm:explicit.chi-3}

We need the following lemma:
\begin{mylemma}
  \label{lm:block-diagonal}
  Let
  \begin{equation}
    \label{eq:lm:block-diagonal}
    M = \begin{pmatrix}A & B \\ B^{\dagger} & B^{\dagger} A^{-1} B + C\end{pmatrix}
  \end{equation}
  be a finite-dimensional block diagonal matrix over $\mathbb{C}$ with
  positive definite $A$.
  Then $M \geq 0$ iff $C \geq 0$.
  In that case $\rank(M) = \rank(A) + \rank(C)$.
\end{mylemma}
\begin{proof}
  The lemma follows from the decomposition
  \begin{equation}
    M =
      \begin{pmatrix}A & B \\ 0 & 1\end{pmatrix}^{\dagger}
      \begin{pmatrix}A^{-1} & 0 \\ 0 & C\end{pmatrix}
      \begin{pmatrix}A & B \\ 0 & 1\end{pmatrix}.
  \end{equation}
If $M \geq 0$ then 
\beq
\begin{pmatrix}A^{-1} & 0 \\ 0 & C\end{pmatrix} =  \begin{pmatrix}A^{-1} & -A^{-1} B \\ 0 & 1\end{pmatrix}^\dag M \begin{pmatrix}A^{-1} & -A^{-1} B \\ 0 & 1\end{pmatrix} \geq 0,
\eeq 
so $C \geq 0$. If $C \geq 0$ (and we know $A^{-1} \geq 0$ from the precondition of the Lemma), then $\begin{pmatrix}A^{-1} & 0 \\ 0 & C\end{pmatrix} \geq 0$, so $M \geq 0$.
  
\end{proof}

\begin{proof}[Proof of \cref{lm:chif.mixed.bound}]
Here we use the facts and notation from the proof of \cref{lm:explicit.chi-2}.
We are interested in deriving a bound for $\Tr(P_0 r P_0^{\dagger}) = \Tr(P_0 \delta \rho P_0^{\dagger}) = \Tr(P_0 \rho(\vlambda') P_0^{\dagger})$
from \cref{eq:TrI1} using \cref{lm:block-diagonal}. Consider a basis where
$\rho(\vlambda) = \diag(\xi_1,\dots,\xi_{\dim{\mcH}})$ with $\xi_{j} = 0$
for $j > n_{+}$. In that basis
  \begin{equation}
    \label{eq:rhovlambdap.blocks}
    \rho(\vlambda') = \begin{pmatrix}
      \rho_{+}(\vlambda') & P_{+} \delta\rho P_0^{\dagger} \\
      P_{0} \delta\rho P_{+}^{\dagger} & P_0\delta\rho P_0^{\dagger}
    \end{pmatrix}.
  \end{equation}
That is, $\rho(\vlambda')$ (which we know is positive-semidefinite) can
be written in the form \cref{eq:lm:block-diagonal} with
\bes
\begin{align}
  A &= \rho_{0{+}} + \delta\rho_{+},
  \quad
  B = P_{+} \delta\rho P_0^{\dagger}\\
  C &= P_0\delta\rho P_0^{\dagger} - 
  P_{0} \delta\rho P_{+}^{\dagger}
  (\rho_{+}(\vlambda'))^{-1}
  P_{+} \delta\rho P_0^{\dagger}.
\end{align}
\ees
From \cref{lm:block-diagonal} we know that $C\geq 0$, implying
\begin{equation}
  \label{eq:chif.mixed.bound:proof.4}
  \Tr(P_{0} \delta\rho P_{0}^{\dagger})
  \geq \Tr\left(P_{0} \delta\rho P_{+}^{\dagger}
    (\rho_{+}(\vlambda'))^{-1}
    P_{+} \delta\rho P_0^{\dagger}\right).
\end{equation}
Taking the second term in the Taylor expansion near $\vlambda' = \vlambda$
we obtain \cref{eq:chi.mixed2.bound}.

The second statement of \cref{lm:chif.mixed.bound} can be obtained
by noticing that the strict inequality in \cref{eq:chif.mixed.bound:proof.4}
is only possible if $C \neq 0$, i.e., $\rank(C) > 0$. Using \cref{lm:block-diagonal}
we conclude that $\rank(\rho(\vlambda')) > n_{+} = \rank(\rho(\vlambda))$.

The third statement can be obtained by substituting \cref{eq:chi.mixed2.bound} as an equality into \cref{eq:explicit.chi.mixed}.
\end{proof}

\section{Previous work related to our ML task}
\label{app:related-work}

A wide body of research exists on the application of ML to the study of the physics of many-body systems. 
Recently, after this work was concluded, we became aware that the FIM estimation task has already been addressed in Refs.~\cite{arnold2023machine,duy2022fisher}.
These works present methods that solve the FIM estimation task using approaches distinct from ClassiFIM but achieve the same goal.
In particular, \cite{arnold2023machine} discusses the motivation of FIM estimation in the context of estimation of QPTs and provides three methods for doing so, two of which are unsupervised (denoted as $I_2$ and $I_3$ there), i.e., that do not require prior partial knowledge of the phase diagram. The unsupervised ML method FINE from \cite{duy2022fisher}, while not formulated in the context of PTs, is directly applicable to the study of PTs.

Focusing next on earlier work that was concerned with phase transitions, Refs.~\cite{vanNieuwenburg2016LearningPT,Lidiak2020UnsupervisedML} considered a prediction of topological phase transitions without the use of order parameters, the former relying instead on data that are deliberately labeled incorrectly, the latter on unsupervised learning using diffusion maps.
Ref.~\cite{abram2022inferring} used the uncertainty in the inverse problem (finding the model parameters from the observed states) as an indicator for the location of the phase transition. Uncertainty was also used to infer the locations of the transitions in the setting of unsupervised learning~\cite{Li2023MachineLP}. Knowing the phase labels for some of the observed states, one can train an algorithm to distinguish phases on held-out states~\cite{Guo2022LearningPT}. Other works applied various classical machine learning methods to quantum phases of matter data, obtained both numerically~\cite{Karsch2022AML} and experimentally~\cite{Kming2021UnsupervisedML}. The quantum variational eigensolver algorithm was combined with classical support vector machines to classify different phases of the quantum transverse field Ising model~\cite{uvarov2020machine}. 
Two frameworks have been developed for quantum state reconstruction and evolution using neural network ansatzes for representing quantum states: NetKet~\cite{carleo2019netket,vicentini2022netket} and QuCumber~ \cite{beach2019qucumber}. In particular, Ref.~\cite{ch2017machine} used CNNs to map the N\'eel phase boundary between the ordered phase and the disordered high-temperature phase in the 3D Fermi-Hubbard model using data generated via diffusion quantum Monte Carlo (DQMC). These earlier machine-learning applications did not utilize the connection between phase transitions and fidelity susceptibility.

Methods have also been proposed to compute the fidelity susceptibility (or quantum FIM) directly via QMC simulations~\cite{wang2015fidelity}. Unlike our method, which assumes access to sample measurements (i.e., $\mathcal{D}$), QMC methods only require a physical model description (i.e., a Hamiltonian along with initial conditions). When applicable, QMC methods are thus preferred in a simulation context, but there are two well-known situations where QMC might fail. First, QMC samples Gibbs states, not ground states.
Second, QMC may fail to equilibrate, e.g., if it has a sign problem~\cite{Troyer:2005aa} or for non-stoquastic Hamiltonians~\cite{Marvian:2019aa,klassen2019hardness}. On the other hand, our method is compatible with many data sources, including both QMC and methods that can sample ground states without suffering from a sign problem such as DMRG. Our approach is also compatible with quantum computing platforms that could potentially sample systems that are intractable for all known classical numerical methods (e.g., a 3D model with a severe sign problem), as explained in \cref{subsec:data-sources}.

The FIM has wide applications in statistics and ML.
In ML, the FIM has been studied with respect to the weights of a neural network as parameters~\cite{karakida2019universal,martens2020new,kirkpatrick2017overcoming,amari1998natural}. In contrast, we focus on the FIM with respect to a small (up to $5$) number of parameters $\vlambda$. Ref.~\cite{ghimire2021reliable} focuses on improving estimates of the KL divergence ($\DKL$), a non-symmetric measure of distance between two probability distributions. The FIM can be recovered as the first term in the Taylor expansion of $\DKL$, so the estimation of $\DKL$ can be applied to the task we define by comparing the distributions corresponding to nearby points in the parameter space. Generally, most prior ML work focuses on estimating a distance between just two distributions given samples, while our method is designed specifically for phase diagrams, where the number of distributions is significantly larger (e.g. $4096$ for a $64\times 64$ grid used in our 2D datasets).

Ref.~\cite{sriperumbudur2009integral} studies several notions of distance between two probability distributions, including $\phi$-divergences and integral probability metrics (IPMs), and contributes estimators for several integral probability metrics (IPMs). While the FIM can be recovered as the first non-trivial term in the Taylor expansion of $D_\phi$ ($\phi$-divergences) for a wide class of functions $\phi$, it cannot be recovered in a similar way from IPMs. Thus, the methods for estimating IPMs provided by Ref.~\cite{sriperumbudur2009integral} do not seem to be applicable to the task of estimating FIM from samples. On the other hand, our task can be seen as related to the simultaneous estimation of a divergence $D_{\phi}$ between multiple pairs of similar probability distributions. Such divergence measures can often be estimated from estimates of density ratios~\cite{sugiyama2012density}. Meta-learning can be used when a given density ratio estimation task is one of a set of tasks~\cite{kumagai2021meta,finn2017model}.

When the only parameter is time, a closely related problem is change point detection \cite{aminikhanghahi2017survey,truong2020selective}. The task of estimating divergence measures from samples (usually between two distributions) has been studied, e.g., in Ref.~\cite{sugiyama2013direct,ghimire2021reliable,nguyen2010estimating}. Our task can be seen as related to the simultaneous estimation of a divergence measure between multiple pairs of similar probability distributions. However, we note that some divergence measures, such as the integral probability metrics (IPMs) studied in Ref.~\cite{sriperumbudur2009integral}, are not invariant under reparameterizations of the sample space and do not have the FIM as their limit. Other aspects of such divergence measures can be found in Refs.~\cite{sugiyama2012density,kumagai2021meta,finn2017model}.

\section{Generating physical model datasets and ground truth}
\label{as:groundtruth}

Here, we discuss the numerical details of generating the datasets defined in \cref{subsec:datasets}. For each of these models, we consider a grid of $\vlambda$ values (either $64 \times 64$ or just $64$ points depending on the model). The first step in generating a dataset for a model is to prepare a fiducial state of this model for each $\vlambda$ value. This is typically an expensive and non-trivial task, but given this fiducial state, computing the ground truth FIM and generating sample measurements in the computational basis is straightforward and efficient. 

For DMRG (and also Clifford simulations, though not considered in this work), it is also possible to straightforwardly perform random local rotations before measuring in the computational basis and, hence, to generate a classical shadow dataset. This does not hold for Lanczos (or QMC, though not considered in this work either), which would require rotating the initial Hamiltonian first and, hence, rerunning the fiducial state generation procedure for each measurement. This is prohibitively expensive and hence is not considered in this work.

\subsection{Ising400 and IsNNN400}
\label{as:Ising}

\subsubsection{FIM from classical MCMC}
\label{ass:FIM-from-MCMC}

Suppose we are estimating the FIM for a statistical manifold $(\mcM, P)$, where we are able to get samples for any $\vlambda \in \mcM$ and query the values $\tilde{P}_{\vlambda}(x) = Z(\vlambda) P_\vlambda(x)$, where $P_\vlambda(x)$ is the probability to get $x$ according to the distribution $P_{\vlambda}$ and $Z(\vlambda)$ is the partition function (unknown). In this subsection we explain, how to estimate FIM using such access. For example, this procedure can be used when $P_\vlambda$ corresponds to Gibbs distribution for some classical Hamiltonian family $H_{\vlambda}$ and $T=1$ (cases with $T \neq 1$ can be reduced to $T=1$ by rescaling the Hamiltonian), where samples can be obtained using MCMC, such as Ising400 and IsNNN400. In this case one can take $\tilde P_\vlambda(x) = \exp(-H_\vlambda(x))$, but for numerical stability we use
\begin{equation}
  \tilde P_\vlambda(x) = \exp(-(H_\vlambda(x)-E_{\vlambda})),
\end{equation}
where $E_{\vlambda}$ is the mean energy of samples $x$ obtained for that $\vlambda$.

Suppose we have two points in the parameter space: $\vlambda_{+}$ and $\vlambda_{-}$
with $\vdlambda = \vlambda_{+} - \vlambda_{-}$, $\vlambda_0 = (\vlambda_{+} + \vlambda_{-})/2$. We have
\begin{multline}
    g(\vlambda_0)(\vdlambda) / 4 \simeq
    2\sum_{x}\left(1-\sqrt{P(x|\vlambda_{-})P(x|\vlambda_{+})}\right)\\
= 2\mathbb{E}_{x\sim q} \left(1 - \frac{2}{r_x + r_x^{-1}}\right),
\end{multline}
where $g$ is the FIM (with respect to $\vlambda$) to be estimated, $q$ is a probability distribution given by $2q_x = P_{\vlambda_{+}}(x) + P_{\vlambda_{-}}(x)$, and $r_x = \sqrt{P_{\vlambda_{+}}(x) / P_{\vlambda_{-}}(x)}$. Thus, in order to estimate $g(\vlambda_0)(\vdlambda)$ we need to get estimates of $r_x$ at random samples at $\vlambda_{+}$ and $\vlambda_{-}$. We have
\begin{equation}
    r_x = \sqrt{\frac{P_{\vlambda_{+}}(x)}{P_{\vlambda_{-}}(x)}} =
    \sqrt{\frac{\tilde{P}_{\vlambda_{+}}(x)}{\tilde{P}_{\vlambda_{-}}(x)}}
    \sqrt{\frac{Z_{\vlambda_{-}}}{Z_{\vlambda_{+}}}}.
\end{equation}
The ratio $\sqrt{Z_{\vlambda_{-}} / Z_{\vlambda_{+}}}$ can be estimated similarly:
\begin{equation}
\label{eq:ising-zpzm}
    \frac{Z_{\vlambda_{-}}}{Z_{\vlambda_{+}}}
    = \mathbb{E}_{x\sim P_{\vlambda_{+}}} \frac{\tilde P_{\vlambda_{-}}(x)}{\tilde P_{\vlambda_{+}}(x)}
    = \left(\mathbb{E}_{x\sim P_{\vlambda_{-}}} \frac{\tilde P_{\vlambda_{+}}(x)}{\tilde P_{\vlambda_{-}}(x)}\right)^{-1}.
\end{equation}
In fact, \cref{eq:ising-zpzm} provides two alternative ways to estimate this ratio, and we use the square root in an attempt to reduce the variance of the estimate.

The procedure explained here was used to estimate the FIM for both the Ising400 and IsNNN400 datasets, which we discuss next.

\subsubsection{Ising400}
\label{ass:ising400-dataset-gen}

\begin{figure*}
\centering
\includegraphics[width=0.99\linewidth]{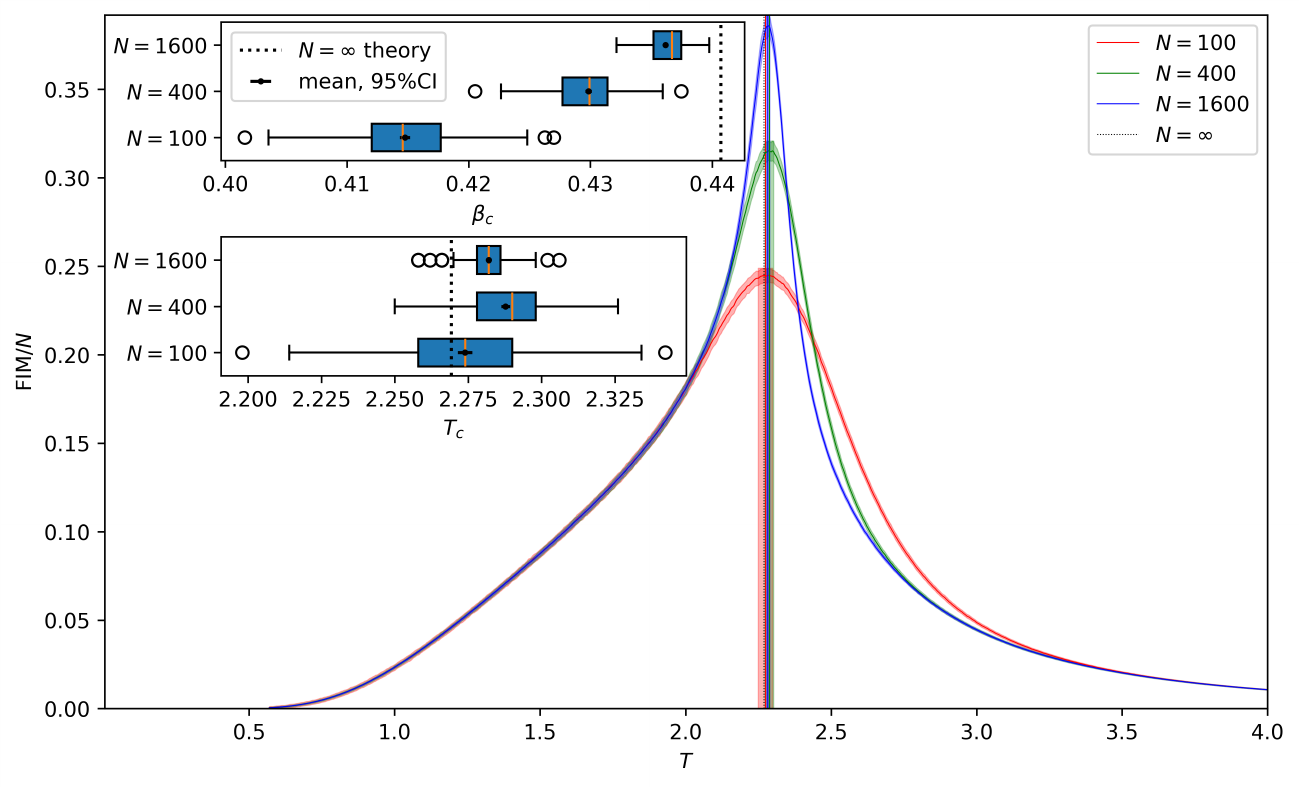}
\caption{Comparison of the theoretical $T_c = 2 / \ln(1 + \sqrt{2})$ for the infinite grid
with the maximum of the FIM. For the result shown here, we repeat MCMC simulation for three sizes: $10\times10$ ($N=100$), $20\times20$ ($N=400$), and $40\times40$ ($N=1600$). The FIM is scaled by the number of spins. For each grid, we produce $11$ independent FIM estimations (with different seeds) and plot the mean and a color band showing $\mathrm{mean} \pm \mathrm{std}$. The insets show whisker plots as a function of $T$ (bottom) and inverse temperature $\beta$ (top), for each of the three sizes. For each mean maximum of the FIM, we performed a two-sided t-test; the observation that the mean of the FIM maximum as a function of $T$ ($\beta$) is greater (less) than $T_c$ ($\beta_c$) is statistically significant for each of the lattice sizes. As expected, the theoretical critical value is approached as $N$ increases.}
\label{fig:ising-fim-gt-vs-theory}
\end{figure*}

For Ising400, the fiducial state is the Gibbs state at temperature $T$, and $T \in (0, 4]$ is the parameter we vary to induce a classical phase transition.
For each dataset, we run $70$ independent MCMC simulations, each starting from
a random ($T=\infty$) state, setting $T=4.0$ and equilibrating for $70$ MCMC steps
(each step involves updating all $20 \times 20$ sites, thus can be split into $400$ substeps). Then for each of $1000$ equally spaced temperatures in the range $(0, 4]$ (going from highest to lowest) we adjust the temperature, run $\floor{30d} + \floor{9d}$ MCMC steps, sample, run another $\floor{9d}$ MCMC steps, and sample again. Here
\begin{equation}
\label{eq:ising400-dT}
d = d(T) = 1 + (0.1 + (T - T_c)^2)^{-1},
\end{equation}
where $T_c = \frac{2}{\ln\left(1 + \sqrt{2}\right)}$, is subtracted to capture longer equilibration times for temperatures close to the critical value. When sampling, we randomly flip (or not) and shift the lattice to further reduce possible auto-correlations. In addition, we record the energies $H(x)$ and estimate the FIM.
For consistency with the other datasets, where all $\lambda_\mu \in [0,1]$, we estimate the FIM with respect to $\lambda = T/4$.

Compared to the other datasets, the generation of Ising400 is extremely fast: all $10$ datasets, together with the FIM estimates, were obtained in under $80$s on a single personal workstation. The MCMC code is implemented from scratch in C++ (using stdlib) and is made available to Python code using ctypes.
For details on how FIM was estimated for Ising400, see \cref{ass:FIM-from-MCMC}.

As this is perhaps the simplest of our models, where the phase transition location is known analytically, \cref{fig:ising-fim-gt-vs-theory} presents preliminary data on finite size scaling, though a full investigation of this aspect is left for future work. 
We use MCMC simulation to approximate the FIM ground truth multiple times (to also obtain error bars via bootstrapping) for 3 sizes: $10\times10$, $20\times20$, and $40\times40$. We doubled the number of iterations for $40\times40$ to ensure equilibration. The resulting FIM scaled by the number of spins experiences a scaling collapse in the low-temperature range, shown in \cref{fig:ising-fim-gt-vs-theory}. We limited our investigation of the scaling near the phase transition to obtaining averages and error bars of the peak location at different sizes (inset of \cref{fig:ising-fim-gt-vs-theory}). Note that the averages behave differently for $T_c$ and $\beta_c$, yet both appear to approach the theoretical value given by $T_c = 2 / \ln(1 + \sqrt{2})$ in the thermodynamic limit.

\subsubsection{IsNNN400}
\label{ass:ising-nnn-400-dataset-gen}
We use $64 \times 64$ grid $\mcM'$ in the parameter space $\mcM$. For each of the vertical
lines of that grid, and for each of the datasets, we run 70 independent MCMC simulations. Each simulation involves starting a parallel tempering (PT) array for all 64 values of $\lambda_1$. Initially we set temperatures for Ising instances in the PT array lower than those dictated by $\lambda_1$ to allow the system to cool down using MCMC steps for individual instances (involving spin updates and line updates) and parallel tempering steps. Then we restore the original temepratures and run additional MCMC iterations to ``re-heat'' the system. Then we sample two configurations for each $\lambda_1$ with 32 MCMC steps between them.

\subsection{FIL24, Hubbard12, and FIL1Dn datasets}
\label{ass:fill-hubbard-datasetgen}

For FIL24 and Hubbard12, the fiducial state is the ground state on a $64\times 64$ grid of $\vlambda$ values. For these models, we diagonalized the corresponding Hamiltonian
using the Lanczos-Arnoldi algorithm with implicit restarts \cite{sorensen1997implicitly}
as implemented in ARPACK-NG \cite{lehoucq1998arpack} available in Python
via SciPy \verb!scipy.sparse.linalg.eigsh!.
This allowed us to obtain
a ground state distribution (technically, a probability distribution
corresponding to a low temperature of $10^{-7}$; see below). This was the most resource-intensive part, consuming about 3000 CPU-hours for the Hubbard12 and FIL24 Hamiltonian families.
From this probability distribution we generated
both the ground truth FIM and
datasets $\mcD$ obtained by selecting $140$ samples from each of the $64^2$-sized
probability distributions.
The training dataset $\mcDtrain$ is a random $90\%$ sample of this dataset,
and was used for training and producing estimates $\hat{g}$. The remaining $10\%$
were only used for the final evaluation.

\subsubsection{Computing the ground state with Lanczos}
\label{ass:eigdeg}

For each of two Hamiltonian families (Hubbard12 and FIL24) we used ARPACK-NG via the \verb!scipy.sparse.linalg.eigsh! function to generate the four lowest eigenstates
$\{\ket{\psi_j(\vlambda)}\}_{j=0,1,2,3}$
and the corresponding eigenvalues $\{E_j(\vlambda)\}_{j=0,1,2,3}$
for each $\vlambda=(\lambda_0, \lambda_1)$ on an $r \times r$ grid
\begin{equation}
\mcM' = \{(l_0/r, l_1/r): l_0,l_1 \in \{0, 1, \dots, r-1\}\},
\end{equation}
where we selected $r=64$. We note that if one is interested in a higher resolution estimate in some region of the parameter space (e.g., near peaks of $\hat{g}$), one can zoom in on that region and collect a new dataset. We then used a truncated Boltzmann distribution for the probabilities,
i.e., we set
\begin{equation}
p_z(\vlambda) = \frac{1}{C(\vlambda)} \sum_{j=0}^{3} \abs{\braket{z|\psi_j}}^2 e^{-\beta (E_j(\vlambda)-E_0(\vlambda))},
\end{equation}
where $C(\vlambda)$ is the normalization constant chosen to ensure $\sum_{z} p_z(\vlambda) = 1$.
We used $\beta = 10^{7}$. The temperature $1/\beta = 10^{-7}$ is typically lower than
the energy gap between the ground state and the first excited state,
so the Boltzmann distribution we obtain is typically dominated by the ground state.
A notable exception to this observation is a set of $18$ points in the parameter space of Hubbard12, where the ground state is degenerate up to numerical precision, and the result of the above procedure is highly non-deterministic.
All these points are located on the boundary of the parameter space
along the $\lambda_0=0$ line and therefore should not have a significant
effect on the phase diagram: ClassiFIM was not designed to predict
phase boundaries on the boundary of the parameter space.

We use double-precision floating point numbers (float64) for eigenstates,
eigenvalues, and probabilities. Single precision
might be sufficient for our purposes, but it would require additional
effort to verify that the use of lower precision does not
significantly affect the probabilities or the phase diagrams. The effect of
lower precision could be significant, for example, if the gap between
the ground state and the first excited state is between $10^{-10}$ and $10^{-7}$:
in this case, Lanczos-Arnoldi with double precision would easily
distinguish the states, but Lanczos-Arnoldi with single precision
would not be able to do so, resulting in the ``ground state''
returned being an arbitrary linear combination of the true ground state
and the true first excited state.

\subsubsection{Dimension reduction}
The elements of the Hilbert spaces
corresponding to Hubbard12 and FIL24 can be encoded using $n=24$ qubits.
Thus, without any additional tricks, the above
procedure would work with sparse $2^{n}$-dimensional vectors.
We used two different dimensionality reduction techniques to generate the Hubbard12 and FIL24 datasets.

The Hamiltonian conserves the number of particles of each spin, $N_{\uparrow}$ and $N_{\downarrow}$.
This allows us to restrict the problem to the subspace $\mathcal{H}_0$ with
$N_{\uparrow} = N_{\downarrow} = \left|\Lambda\right|/2$. Since we are working
at half-filling (i.e., the number of particles of each spin is fixed to $n/2$),
we only need to work with $24$-bit bitstrings, each half of which has exactly
$n/2=6$ bits set to $1$. The total number of such bitstrings (classical configurations) is
$\binom{n}{n/2}^2 = 924^2 = 853776$, equal to
the Hilbert space dimension $\dim(\mathcal{H}_0)$. This size allows us to use the Lanczos-Arnoldi method with implicit
restarts, as implemented in ARPACK-NG~\cite{ARPACK-NG}, to identify low-energy eigenstates of the Hamiltonian for any
combination of parameters. Each state can be represented
by $853776$ double-precision floating point numbers. We created lookup tables
to convert between the indices $v \in \{0, 1, \dots, 853775\}$ and the
bitstrings $z \in \{0, 1\}^{24}$ and precomputed the sparse matrices
used in the Hamiltonian in the ``$v$'' representation.

For FIL24 we note that all off-diagonal elements of the Hamiltonian are
non-positive (i.e., the Hamiltonian is stoquastic~\cite{Bravyi:2006aa}).
This allows us to reduce the dimensionality of the ground
state problem to be solved using Lanczos by considering
orbits of bitstrings under the symmetry group $D_{12}$ (the 24-element Dihedral group, which is the symmetry group of the 12-site lattice) instead of bitstrings themselves.
The number of such orbits is
the number of bracelets (turnover necklaces) of 12 beads of 4 colors
$\textnormal{A032275}(12) = 704370$~\cite{Bower1998}.
Thus, we reduced the dimensionality of the problem to be solved with Lanczos from $2^{24}$ to $704370$. If
bitstrings $x_1$ and $x_2$ lie on the same orbit of $D_{12}$ and the
ground state $\ket{\psi_0} = \sum_{x\in \mcS} a_x \ket{x}$, then
$a_{x_1} = a_{x_2}$. This allows us to reduce the dimensionality
by indexing the orbits $o_v$ of $D_{12}$ by an index $v \in \{0, 1, \dots, 704369\}$
and storing $b_{v}$ for each orbit $o_v$ instead of $a_x$ for each $x \in o_v$, where
\begin{equation}
b_{v} = \sqrt{\abs{o_v}} a_{x},\quad \textrm{where}\quad x \in o_v.
\end{equation}
Then
\begin{equation}
\ket{\psi_0} = \sum_{v=0}^{704369} b_{v} \ket{o_v},\quad\textrm{where}\quad
\ket{o_v} = \frac{1}{\sqrt{\abs{o_v}}}\sum_{x \in o_v} \ket{x}.
\end{equation}
Similarly to Hubbard12, we precompute the lookup tables to convert between
$x$ and $v$ and precompute a sparse matrix representation of the Hamiltonian
in the $\ket{o_v}$ basis. When using the resulting probability distribution
over the orbits to sample bitstrings, we first sample $v$ according to the
probability distribution, then sample $x$ uniformly from the orbit $o_v$.

\subsubsection{Computing datasets for the ``estimating FIM'' task from Lanczos ground state}
For each of the Hamiltonian families (Hubbard12 and FIL24) we generate a dataset
by randomly sampling 140 bitstrings from each point $\vlambda \in \mcM$,
thus producing a dataset $\mcD$ with
$4096 \cdot 140 = 573440$ pairs $(\vlambda, z)$.
The training dataset $\mcDtrain$
is a random $90\%$ sample of that dataset.

For each of the Hamiltonian families we perform the above dataset
generation 10 times using 10 different random seeds to assess the
stability of the presented metrics.

\subsubsection{Computing the ground truth FIM from the Lanczos ground state}
\label{ass:ground truth}
We compute the ground truth FIM, $g$, by
the finite difference method. For each $l_0 \in \{0, 1, \dots, r-2\}$ and
each $l_1 \in \{0, 1, \dots, r-1\}$ we can estimate
\begin{equation}
\label{eq:gs-chifc.00}
\bal
&[g((l_0 + 1/2) / r, l_1/r)]_{00}
\approx \\
&\qquad 8 r^2 \{1 - \sum_{v}[p_v(l_0/r, l_1/r) p_v((l_0+1)/r, l_1/r)]^{1/2}\},
\eal
\end{equation}
where $p$ is the probability distribution obtained from the Lanczos-Arnoldi method. The quantity within the square root is the product of two probability distributions obtained at nearby points $(l_0/r,l_1/r)$ and $((l_0+1)/r,l_1/r)$, respectively. A detailed derivation of \cref{eq:gs-chifc.00} is provided in Ref.~\cite{ClassiFIM-ML}.
While the finite difference estimate is not exact, we argue that
for a set of estimates on a finite parameter grid $\mcM'$, it is
better than the true value of the FIM for our purposes.
The reason is that the true value may experience a very sharp peak
(or even a Dirac delta in the case of an unavoided crossing
between the two lowest eigenvalues of $H$) at some point in the parameter space.
If the size of the peak is smaller than the grid spacing, then
one can miss it entirely by considering only the values of $g$
at the grid points.

Similarly to
\beq
(g((l_0 + \tfrac12) / r, l_1/r))
\left(\begin{smallmatrix}1\\0\end{smallmatrix}\right)
= (g((l_0 + \tfrac12) / r, l_1/r))_{00}
\eeq
in \cref{eq:gs-chifc.00}, we can estimate
\bes
\begin{align}
&g(l_0/r, (l_1 + \tfrac12)/r)
\left(\begin{smallmatrix}0\\1\end{smallmatrix}\right)
= (g(l_0/r, (l_1 + \tfrac12)/r))_{11}\\
& g((l_0 + \tfrac12)/r, (l_1 + \tfrac12)/r)
\left(\begin{smallmatrix}1\\1\end{smallmatrix}\right)\notag\\
&\qquad= (g((l_0 + \tfrac12)/r, (l_1 + \tfrac12)/r))_{{+}{+}},
\end{align}
\ees
and
\begin{align}
&(g((l_0 + \tfrac12)/r, (l_1 + \tfrac12)/r))
\left(\begin{smallmatrix}1\\-1\end{smallmatrix}\right)\notag\\
&\qquad = (g((l_0 + \tfrac12)/r, (l_1 + \tfrac12)/r))_{{-}{-}}.
\end{align}
Note that these estimates are centered at
different points in the parameter space, which is
inconvenient for the purpose of computing the distance ranking error
or plotting 2D phase diagrams.
Let $\mcM''$ be the set of centers of the grid squares in $\mcM'$.
Two of the four estimates are at points of $\mcM''$.
We average the neighboring estimates of the other two to obtain estimates
at $\mcM''$. We then combine these four estimates to obtain all components
of $g$ at $\mcM''$.

\subsection{Kitaev20}
\label{ass:kit20-dataset-gen}

The Kitaev chain has an efficient exact solution in the momentum basis. We implemented the fermionic Fourier transform to convert it to the position basis.
However, we found that implementing the fermionic Fourier transform efficiently is a non-trivial
task and plan to report our findings in future work.
Again, all numerical algorithms can be found in Ref.~\cite{public-datasets,ClassiFIM-gencode,ClassiFIM-code}. To test our implementation, we compared it with the Lanczos-based solution at one point in the parameter space.

\subsection{XXZ300}
\label{ass:xxz300-dataset-gen}

For XXZ300, the fiducial state is the ground state, so as a first step, we must prepare the ground state for each of the $64 \times 64$ values of $\vlambda$. The underlying model is a 300 qubit 1D spin chain, so we use DMRG which is well suited to this task. At the implementation level, we utilized a fork of DMRGpy library \cite{JoseLado2023} which is a python wrapper on top of ITensor~\cite{itensor} with some additional features. The exact details of our code can be found in Ref.~\cite{ClassiFIM-gencode}.

Our DMRG procedure was inspired by Refs.~\cite{Elben, Huang:22} but deviates from it in some aspects. In particular, the direct implementation of Refs.~\cite{Elben, Huang:22} does not result in the ground state. The procedure we used is the following:
\begin{itemize}
\item Pick three initial candidate states: two random
      states with maximum link (or bond) dimension $16$, and one of the two N{\'e}el states:
      $\ket{0101\ldots01}$, which is aligned with the pinning term $0.1Z_0$ in \cref{eq:XXZ-H0} (the rightmost ``$1$'' corresponds to the $-1$ eigenvalue of $Z_0$).
\item Run DMRG with 1 sweep and \verb!maxlinkdim=16! for each of the three states. Discard one of the states: either find a state with a dot-product with another, lower energy, state close to $\pm 1$, or, if such state does not exist, discard the highest energy state.
\item Run 5 more DMRG iterations with increasing linkdim for the remaining two states.
\item Project the resulting states to eigenvalues of $\otimes_i Z_i$, discarding projections of sufficiently low norm, or close to other states.
\item Pick the state with the lowest energy after one more DMRG sweep.
\item Refine this state using 10 more sweeps.
\item Run excited state DMRG using a similar strategy (re-using the best states discarded in the previous steps). If that state achieves a lower energy, swap it with the ground state and refine using 10 more sweeps.
\end{itemize}

We repeated the generation procedure twice with slightly different choices of details such as the number of sweeps, and with different random seeds. The verification results are illustrated in \cref{afig:dmrgpy}.

\begin{figure*}
\centering
\includegraphics[width=0.99\linewidth]{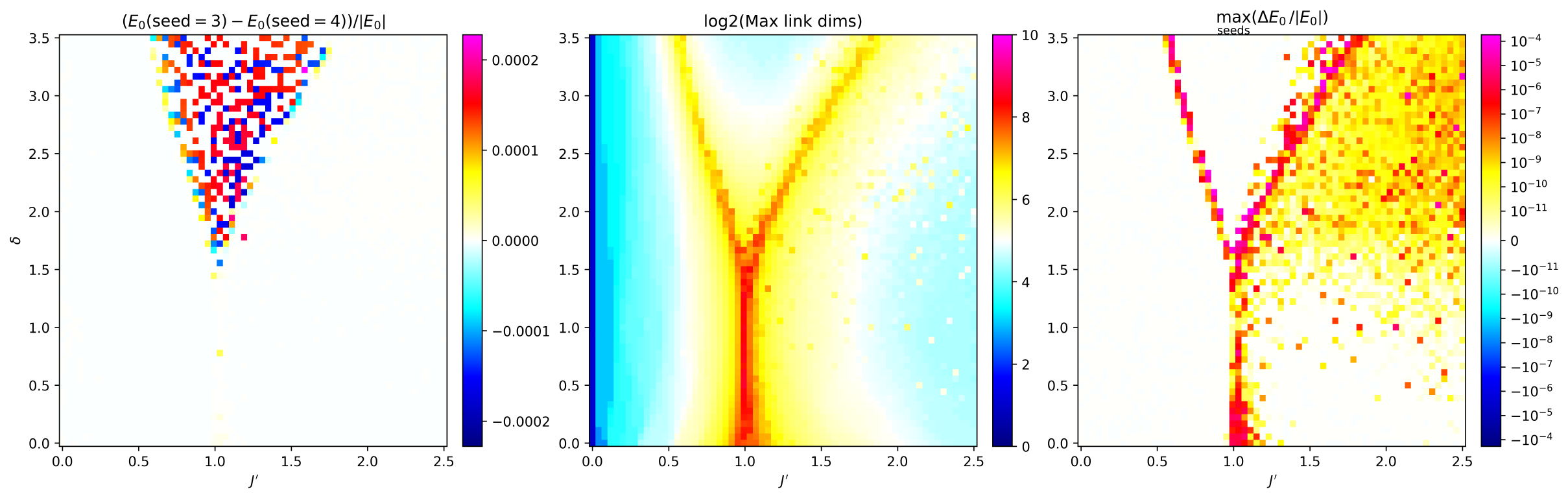}
\includegraphics[width=0.99\linewidth]{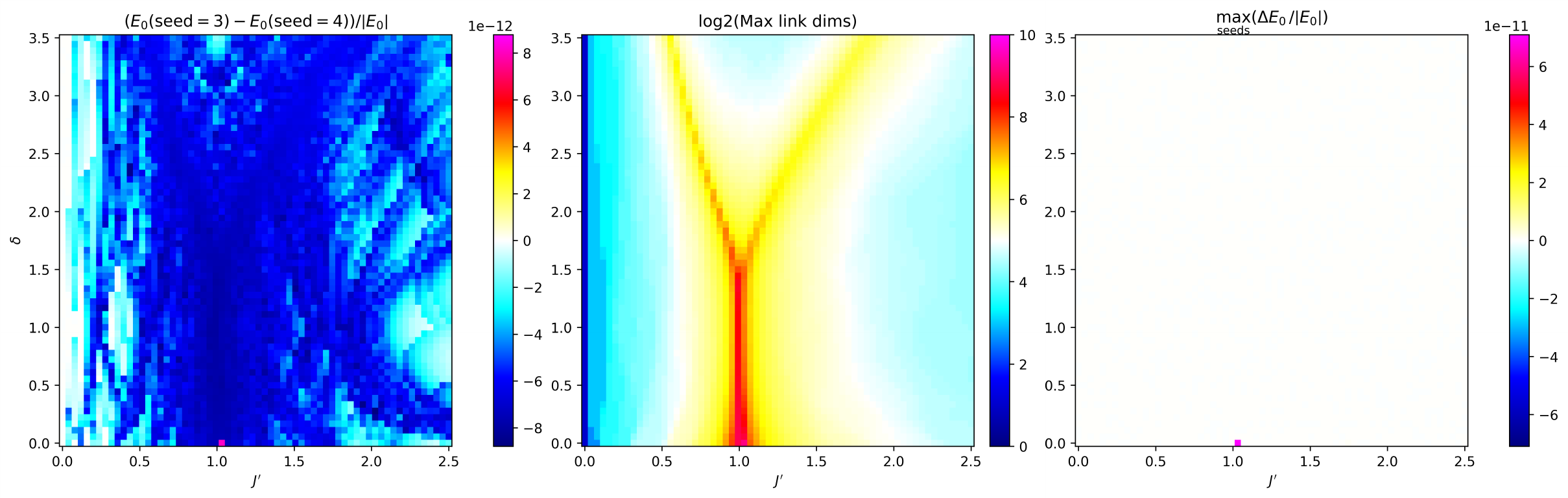}
\caption{Verification of DMRG results. Top: intermediate results from running a simpler procedure without excited state DMRG, using N\'eel states, and projecting to eigenvalues of $\otimes_i Z_i$. Bottom: final results. Left panels: difference between ``ground state'' energies found by different seeds. Middle: $\log_2$ of maximum link dimensions used in DMRG states. Right: improvement in energy made in the last sweep.}
\label{afig:dmrgpy}
\end{figure*}

The final ground state from DMRG is a matrix product state (MPS). Contracting such an MPS is straightforward, and we can use this fact to compute the ground truth FIM and produce samples in the computational basis. At the same time, applying local unitaries on each site is also straightforward, so performing random local measurements would be efficient as well.

\section{Computational resources}
\label{ass:compute}
Here we provide an estimate of the computational resources spent on this project. Most resources were spent on ground state generation, specifically Hubbard12 (${\sim} 3 \cdot 10^3$ CPU-hours) and XXZ300 (${\sim} 1.2 \cdot 10^4$ CPU-hours). This includes initial trial-and-error and generation of the ground states twice (with different random seeds) to verify reproducibility. The net time for generation of dataset and ground truth FIM without verification and initial failed attempts was about $25\%$ of the total. In these estimates, ``CPU-hour'' is one hour on one logical core: e.g., using all cores of one 6-core CPU with hyperthreading for 1 hour is counted as 12 CPU hours. In total, including preliminary experiments not covered here, we spent about $2 \cdot 10^4$ CPU-hours.

Since GPU performance varies by orders of magnitude between different GPU models, to provide resource estimates for training we scale time according to GPU performance by estimating how much time a similar computation would have taken on single NVIDIA GeForce GTX 1060 6GB GPU. According to ``GPU-hour'' defined this way, we spent about $2\cdot 10^2$ GPU hours on model training (including the training of mod-W for comparison).

\section{Application of SPCA to bitstring data}
\label{as:spca-bitstring}
SPCA is a kernel PCA method applied to shadow tomography data
that produces an embedding where different points in the parameter space
correspond to different phases~\cite{Huang:22}.
Classical shadows offer a succinct classical description of quantum many-body states that can be used to accurately predict a wide range of properties with rigorous performance guarantees~\cite{Huang:2020wo}.
Here we review SPCA and describe how to apply it to bitstring data.

\subsection{Review of SPCA}
\label{ass:spca-review}
Ref.~\cite{Huang:22} considers input data consisting of
$T$ random Pauli basis measurements of the state $\rho_\vlambda$
at each point $\vlambda$: each measurement is described by
a corresponding product state
$\ket{s^{(t)}_{\vlambda}} = \bigotimes_{i=1}^n \ket{s_{\vlambda,i}^{(t)}}$,
where each
$\ket{s_{\vlambda,i}^{(t)}} \in \{\ket{0}, \ket{1}, \ket{+}, \ket{-}, \ket{i+}, \ket{i-}\}$
is one of the six Pauli eigenstates describing both the measurement basis ($Z$, $X$, or $Y$)
and its outcome ($+1$ or $-1$) of a single qubit.
Here $t = 1,\dots, T$ indexes different measurements of the same state $\rho_\vlambda$.
A corresponding unbiased approximation to $\rho_{\vlambda}$ is given by
\begin{equation}
  \label{eq:st-sigma-vlambda-t}
  \sigma_{\vlambda}^{(t)} = \bigotimes_{i=1}^{n} \sigma_{\vlambda,i}^{(t)},
\end{equation}
where
\begin{equation}
  \label{eq:st-sigma-vlambda-i-t}
  \sigma_{\vlambda,i}^{(t)} = 3 \ket{s_{\vlambda,i}^{(t)}}\!\!\bra{s_{\vlambda,i}^{(t)}} - I,
\end{equation}
and one can compute a more precise estimate of $\rho_{\vlambda}$
by averaging over $T$ measurements (\cite[Eq.~(1)]{Huang:22}).

The procedure of unsupervised learning of the phase diagram given such
data consists of two steps. The first step is to compute a PCA kernel matrix
(\cite[Eq.~(6)]{Huang:22}):
\begin{equation}
  \label{eq:spca-st-kernel}
  K(\vlambda, \vlambda') = \exp\!\left(
    \frac{\tau}{T^2} \sum_{t,t'=1}^{T}\!\!\exp\!\left(
      \frac{\gamma}{n} \sum_{i=1}^{n}\Tr\left(
        \sigma_{\vlambda,i}^{(t)} \sigma_{\vlambda',i}^{(t')}
      \right)\right)\right),
\end{equation}
where $\tau=1.0$ and $\gamma=1.0$ are hyperparameters.%
\footnote{The referenced paper had a misprint, which we correct here:
the last superscript should be $^{(t')}$ but is mistakenly spelled as
$^{(t)}$ in \cite[Eq.~(6)]{Huang:22}.}
The second is to apply kernel PCA to that matrix to obtain the first few principal components.
Phases are then expected to appear as clusters in the space of the first few principal components.

\subsection{Application to bitstring data}
\label{ass:spca-bitstring}
There are two differences between the FIM-estimation datasets $\mcDtrain$
and shadow-tomography data assumed by \cite{Huang:22}. First,
instead of having $T$ samples for every $\vlambda$,
the number of samples corresponding to a given $\vlambda$ in $\mcDtrain$
is variable ($126$ on average). This is easily accommodated by replacing
the coefficient $\tau / T^2$ in \cref{eq:spca-st-kernel} with
$\tau / (T_{\vlambda} T_{\vlambda'})$ and adjusting the sum limits accordingly.
Second, in our data the samples are bitstrings $x_j$
which can be grouped by $\vlambda$ and renamed to $x_{\vlambda}^{(t)}$
to make the notation similar to $\ket{s^{(t)}_{\vlambda}}$ considered above.
Instead of describing a (possibly mixed) state $\rho_{\vlambda}$ they describe
a probability distribution $P_{\vlambda}(\bullet)$ which can be written as
a vector of probabilities $p_{\vlambda} \in \mathbb{R}^{2^n} = (\mathbb{R}^2)^{\otimes n}$.
To make the analogy clearer,
we can describe $p_{\vlambda}$ as a diagonal density matrix
$\rho_{\vlambda} = \sum_{x} p_{\vlambda}(x) \ket{x}\!\!\bra{x}$.
A corresponding unbiased approximation of $\rho_{\vlambda}$ is given by
\cref{eq:st-sigma-vlambda-t}, where $\sigma_{\vlambda,i}^{(t)}$ is now defined
as
\begin{equation}
  \label{eq:bs-sigma-vlambda-i-t}
  \sigma_{\vlambda,i}^{(t)} = \ket{x_{\vlambda,i}^{(t)}}\!\!\bra{x_{\vlambda,i}^{(t)}}.
\end{equation}
Correspondingly, the kernel matrix $K$ can now be computed by the same
expression \cref{eq:spca-st-kernel} (with the correction mentioned above),
which simplifies to
\begin{multline}
  \label{eq:spca-bs-kernel}
  K(\vlambda, \vlambda') = \exp\Biggl[
    \frac{\tau}{T_{\vlambda} T_{\vlambda'}}
    \sum_{t=1}^{T_{\vlambda}} \sum_{t'=1}^{T_{\vlambda'}}\\
      \exp\Biggl(\frac{\gamma}{n} \sum_{i=1}^{n} I\left(
        x_{\vlambda,i}^{(t)} = x_{\vlambda',i}^{(t')}
      \right)\Biggr)\Biggr],
\end{multline}
where $I(\bullet)$ is the indicator function.

\subsection{Discussion}
In \cref{ass:spca-bitstring}, we introduced an interpretation of SPCA in the context of its application to bitstring data.
Another possible interpretation is that bitstring data should be used to estimate the kernel matrix $K$ that would be obtained if SPCA were applied to shadow tomography data corresponding to the diagonal density matrix $\rho_{\vlambda}$.
By averaging the inner $\exp$ in \cref{eq:spca-st-kernel},
one can see that \cref{eq:spca-bs-kernel}
with $\tau$ replaced by $\tau \exp(-\gamma)$ can serve as an approximation of such $K$.
This approximation becomes exact as $n, T \to \infty$.

Yet another interpretation is that bitstring data should be used to approximate $K$ from \cref{eq:spca-st-kernel} for a pure state with positive coefficients
$\rho_{\vlambda} = \ket{\psi_{\vlambda}}\!\!\bra{\psi_{\vlambda}}$,
where $\ket{\psi_{\vlambda}} = \sum_{x} \sqrt{p_{\vlambda}(x)} \ket{x}$.
However, it is unclear whether this can be done efficiently.
The following example demonstrates that an expression similar to \cref{eq:spca-bs-kernel} cannot serve as such an approximation.

Assume both $n$ and $T$ are large.
Consider three points $\vlambda_1, \vlambda_2, \vlambda_3$,
where $P_{\vlambda_1}$ assigns probability $1$ to the bitstring $00\dots0$,
$P_{\vlambda_2}$ assigns probability $1$ to the bitstring $00\dots011\dots1$
with the first $(1-\alpha)n$ bits set to $0$ and the remaining $\alpha n$ bits set to $1$
for some $\alpha \in [0, 1]$ (which is a multiple of $1/n$),
and $P_{\vlambda_3}$ assigns equal probability $1/2^n$ to all bitstrings.
Furthermore, let us denote the kernel and hyperparameters in \cref{eq:spca-st-kernel}
as $K_{ST}$, $\tau_{ST}$, and $\gamma_{ST}$ (``ST'' for shadow tomography),
and the kernel and hyperparameters in \cref{eq:spca-bs-kernel}
as $K_{BS}$, $\tau_{BS}$, and $\gamma_{BS}$ (``BS'' for bitstring).
Then we have
\begin{subequations}
\begin{align}
  K_{BS}(\vlambda_1, \vlambda_2)
    &= \exp\left(\tau_{BS} e^{(1-\alpha)\gamma_{BS}}\right),\\
  K_{BS}(\vlambda_3, \vlambda_3)
    &= \exp\left(\tau_{BS} e^{\gamma_{BS}/2}(1 + O(1/\sqrt{nT})\right),
\end{align}
\end{subequations}
where $O(1/\sqrt{nT})$ denotes a term with standard deviation
bounded by $\text{const} / \sqrt{nT}$.
In other words, $K_{BS}(\vlambda_3, \vlambda_3) \approx K_{BS}(\vlambda_1, \vlambda_2)$
for $\alpha = 1/2$.
On the other hand, the shadow tomography kernel for a state with positive coefficients
can be estimated by observing that all three such states are product states:
\begin{subequations}
\begin{align}
  K_{ST}(\vlambda_1, \vlambda_2)
    &= \exp\left(\tau_{ST} e^{-\alpha\gamma_{ST}}(1+\epsilon)\right),\\
  K_{ST}(\vlambda_3, \vlambda_3)
    &= \exp\left(\tau_{ST} (1+\epsilon)\right),
\end{align}
\end{subequations}
where $1+\epsilon$ is short for $1 + O(1/n + 1/\sqrt{nT})$.
In other words, $K_{ST}(\vlambda_3, \vlambda_3)
\approx K_{ST}(\vlambda_1, \vlambda_1) \approx K_{ST}(\vlambda_1, \vlambda_2)$
for $\alpha = 0$ (in contrast to $\alpha=1/2$ in the bitstring case).
This is because $\vlambda_3$ corresponds to a pure state
$\ket{+}^{\otimes n}$, which is equivalent up to 1-local Cliffords to the state
$\ket{0}^{\otimes n}$ corresponding to $\vlambda_1$, thus resulting in the
same distribution of the kernel matrix element.
Therefore, the kernels $K_{ST}$ and $K_{BS}$ are different
and cannot be approximated by each other.

\section{Comparison with \texorpdfstring{Ref.~\cite{vanNieuwenburg2016LearningPT}}{W}}
\label{ass:comparison-W}

Recall that we conducted numerical experiments for one of the three models presented in Ref.~\cite{vanNieuwenburg2016LearningPT}, namely, the Kitaev chain with lattice size $L=20$ [see \cref{eq:KitaevH}]. To prevent confusion, we first recall what Kitaev20 means in our work as compared to the input data used in Ref.~\cite{vanNieuwenburg2016LearningPT} for the same model. In both works, information is obtained from the ground state of this model for different choices of the parameter $\mu$ for $t = 1$ fixed. In our work, the Kitaev20 datasets $\mcDtrain$ are constructed by measuring the ground state in the computational basis ${\approx}126$ times (see \cref{ass:kit20-dataset-gen}). In \cite{vanNieuwenburg2016LearningPT}, the input dataset consisted of the top $10$ eigenvalues of the reduced density matrix of the ground state corresponding to $10$ sites. This is also commonly known as a truncated entanglement spectrum from a bipartition of the ground state. Either way, this dataset can be interpreted as a probability distribution over a set of $11$ possible samples: $\{\#0, \#1, \dots, \#9, \text{``not in top 10''}\}$.

We note that the comparison presented here refers to the original W method as presented in Ref.~\cite{vanNieuwenburg2016LearningPT} and does not include the improvements suggested in Ref.~\cite{arnold2023machine}.

\subsection{The entanglement spectrum data is too powerful}
\label{ass:dataset-too-powerful}

\begin{figure}[t]
\centering
\includegraphics[width=\columnwidth]{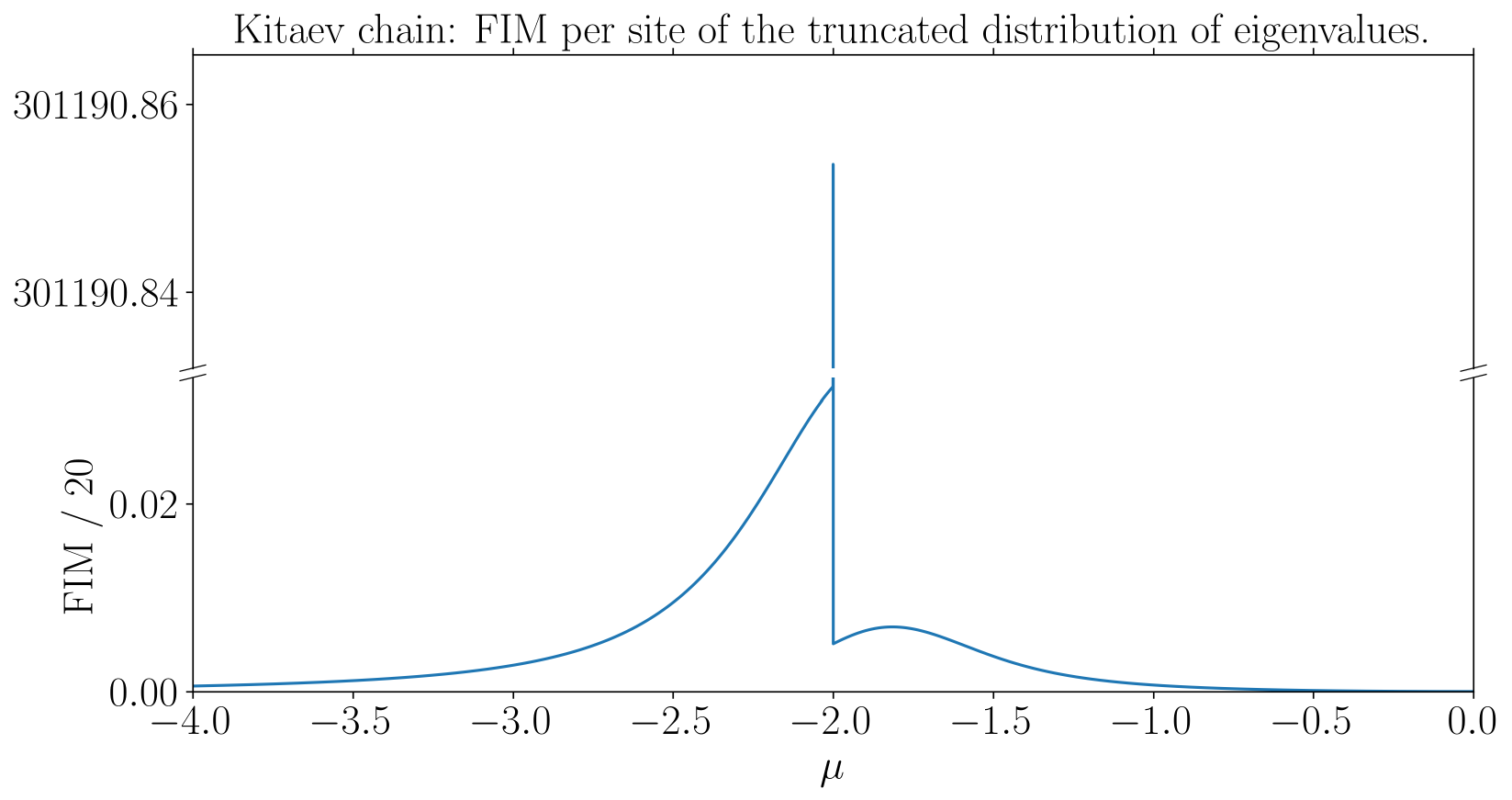}
\caption{The FIM of the distribution obtained by interpretation of the top 10 eigenvalues of the reduced density matrix. The parameter $t$ of the Kitaev chain is assumed to be equal to 1.}
\label{fig:kitaev20-spectrum-fim}
\end{figure}

In \cref{fig:kitaev20-spectrum-fim}, we plot the result of computing the FIM directly from the truncated entanglement spectrum dataset used in Ref.~\cite{vanNieuwenburg2016LearningPT}. Generating this plot does not require the use of ClassiFIM or the ML algorithm described in Ref.~\cite{vanNieuwenburg2016LearningPT}, and yet, it correctly predicts the location of the phase transition. Hence, ML is not needed, and we conclude that the truncated entanglement spectrum is too powerful.

\begin{figure}[t]
\centering
\includegraphics[width=\columnwidth]{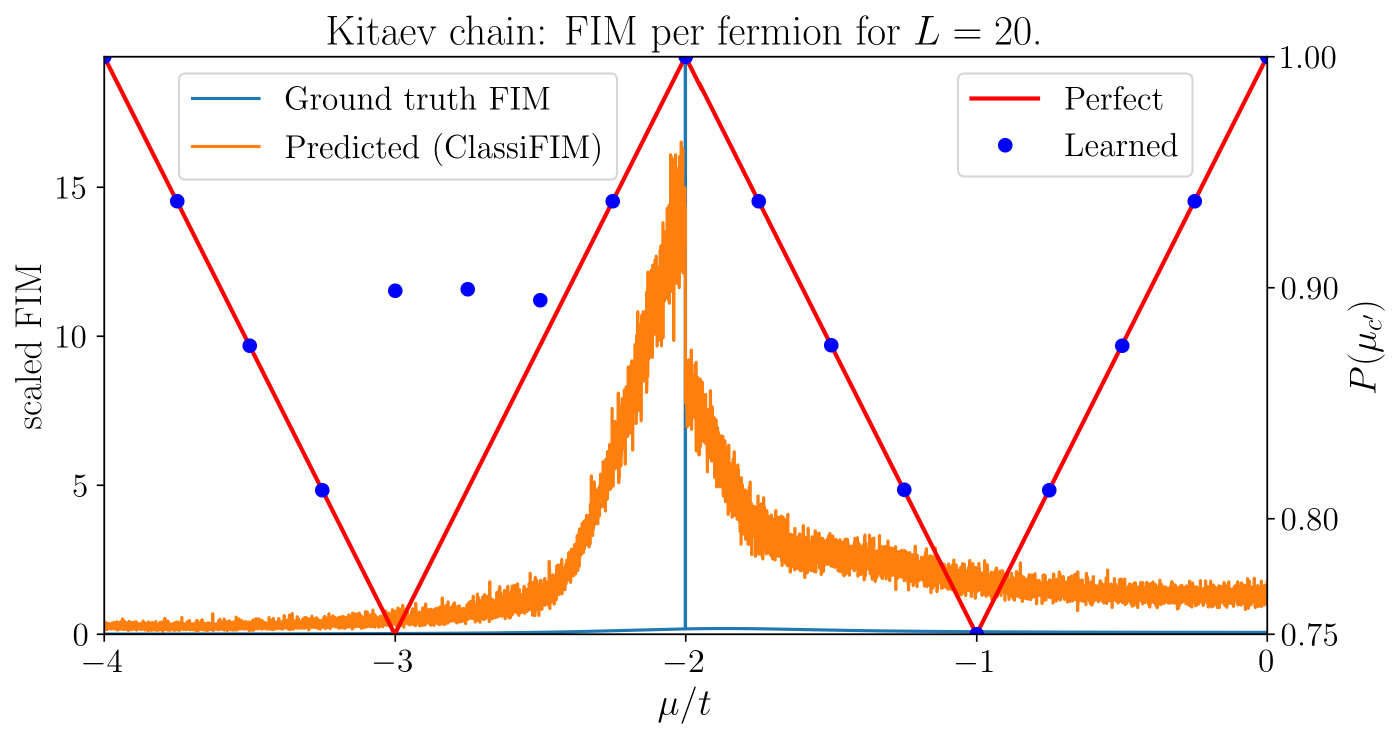}
\caption{Left axis: The FIM in our approach and its estimate computed from Kit1D20. Right axis: ideal and numerically computed value of the indicator \cite{vanNieuwenburg2016LearningPT}.}
\label{fig:KitFid}
\end{figure}

\subsection{Details of our comparison with the ``W" method}
\label{ass:w-comp-details}

Ref.~\cite{vanNieuwenburg2016LearningPT} presented two main ideas related to identifying phase transitions: (i) In order to refine the location of a phase transition, one can train a binary classifier that predicts one of two phases trained on the data sampled from the regions where the phase is known, and then plot average values of predictions in the unknown region, thus obtaining the refined location of the phase transition as a crossover point where the number of samples labeled by each of the two classes become equal to each other; (ii) (Confusion Scheme) In order to draw a phase diagram, one can train multiple neural networks on incorrectly labeled data and expect the accuracy of such a neural network to form a W-like shape with the middle peak pointing to the location of a phase transition. Our method, ClassiFIM, provides an alternative to (ii): method (i) can be applied on top of any method after an approximate phase diagram is obtained. Therefore, a comparison of ClassiFIM with Ref.~\cite{vanNieuwenburg2016LearningPT} means a comparison of ClassiFIM with the Confusion Scheme. Because of how this scheme works, we call it the ``W" method.

In \cref{fig:KitFid}, we compare the output of ClassiFIM to the W method for Kitaev20. Both methods correctly predict $\mu/t = -2$ as the critical point of the TQPT. Let us summarize the meaning of the plot for each method. 

The left vertical axis (``scaled FIM'') pertains to ClassiFIM, as do the orange and blue solid lines, which represent the ClassiFIM output and the ground truth FIM, respectively. The blue solid line is almost zero everywhere except at $\mu/t = -2$ where it is sharply peaked. Evidently, the approximate FIM from ClassiFIM has a much broader peak, but its maximum is still correctly located around $\mu/t = -2$. Hence, ClassiFIM correctly predicts the location of the TQPT. The W method would ideally yield the red curve, whose numerical values correspond to the right vertical axis ($P(\mu_{c'})$), representing the classification accuracy of mislabelled samples. The actual implementation provided in Ref.~\cite{vanNieuwenburg2016LearningPT} yielded the blue points, predicting that the phase transition occurs at $\mu/t = -2$.

ClassiFIM exhibits two technical advantages that may not be apparent from \cref{fig:KitFid}. First, our datasets are different. For this fermionic system, the ``computational basis measurement" dataset used in ClassiFIM corresponds to measuring the presence or absence of a fermion at each site $140$ times per parameter value. This is a much simpler measurement protocol than obtaining the truncated entanglement spectrum that is used in the W method. Moreover, we have already shown in \cref{fig:kitaev20-spectrum-fim} that the W dataset is powerful enough to predict the QPT location \emph{without} ML. No such procedure is possible with the Kitaev20 bitstring dataset.

Second, the W scheme was applied for $\mu \in [-4, 0]$.
We observe that this scheme applied to $[-4, 4]$ instead would fail,
i.e., no peak in the middle would be observed.
Instead of a W-like shape, the Confusion Scheme would return a V-like shape.
The reason is that the truncated entanglement spectrum is symmetric about $\mu = 0$.
On the other hand, ClassiFIM works correctly whether the range is $[-4, 0]$ or $[-4, 4]$.
In fact, the above ClassiFIM data was indeed trained on data from $\mu \in [-4, 4]$,
and the estimated FIM in $[0, 4]$ (not shown) roughly matches the mirror image
of that in $[-4, 0]$.
When we applied ClassiFIM to $\mu \in [-4, 0]$ we found the curve to more
sharply peaked and the peak location to be unchanged.

\section{PeakRMSE: Quantitative comparison in machine learning of phase transitions}
\label{as:peakrmse} 

As mentioned in the Introduction, there is a lack of standardization in the field of unsupervised ML of QPTs that prevents direct quantitative comparisons between methods.
This is also clear from the discussion of Ref.~\cite{vanNieuwenburg2016LearningPT} in the previous section [\cref{ass:comparison-W}]; \cref{fig:KitFid} is ultimately unsatisfactory since any quantitive comparison of peaks is mired by using a different dataset and strategy from the start. While some methods such as ClassiFIM and those presented in Refs.~\cite{arnold2023machine,duy2022fisher} solve the FIM estimation task, W and SPCA do not.
Nevertheless, we have found that it is possible to compare methods for solving the FIM estimation task with the methods producing other types of phase diagrams (e.g. Refs.~\cite{vanNieuwenburg2016LearningPT, Huang:22}) indirectly using a metric ``PeakRMSE'' we introduce,
applicable when ground truth locations of finite size precursors to phase
transitions are known. For simplicity, we refer to these precursors as ``peaks''
below, even though they might be other features such as, e.g., level crossings, depending
on the type of phase diagram used.

\subsection{Motivation behind PeakRMSE}
\label{ass:peakrmse-motivation}
Given different methods for building a phase diagram estimate different quantities
we must focus on something common they all allow to predict: the phase boundaries.
E.g. the phase boundary predicted by W consists of peaks with accuracy close to 1,
phase boundary from FIM is indicated by local peaks of FIM,
and SPCA ``colors'' different phases in different colors, so phase boundaries
are the boundaries of the colored regions.

Yet, W can only be applied to 1D parameter space,
so when the full phase diagram has more dimensions
(such as our 2D datasets Hubbard12, FIL24, IsNNN400, XXZ300),
we have to restrict our attention to 1D slices of such dataset:
horizontal slices keep $\lambda_1$ constant and sweep $\lambda_0$ from $0$ to $1$,
and vertical slices keep $\lambda_0$ constant and sweep $\lambda_1$ from $0$ to $1$.

For each fixed $\lambda_1$, we know from the exact FIM plot how many peaks are expected.
Suppose it is one peak at $\lambda_0^*(\lambda_1)$. We now can use method-specific
procedure to extract the most likely peak $\hat\lambda_0(\lambda_1)$.
 We can then compute a squared error between these peak predictions,
 $(\lambda_0^*(\lambda_1) - \hat\lambda_0(\lambda_1))^2$.
When two peaks are present for a given $\lambda_1$, we simply ask the method
for its two highest peaks.
We can then return the square root of mean of such squared errors as PeakRMSE.

In practice, naively following this strategy results in a few issues,
illustrated by the following examples:
\begin{itemize}
  \item Consider an output of a method for a single 1D slice. Such output
    might be noisy and contain a large number of peaks. If we allow the method
    to return all of them as predictions, the PeakRMSE might be small
    (because to any location of the slice, including the ground truth peak,
    there is a nearby predicted peak) even though the method's prediction
    is not useful.
  \item Imagine that, to avoid the above issue, we restricted each method to
    return at most as many peaks as there are ground truth peaks in the slice.
    Imagine, further, that the ground truth contains a single peak in the slice,
    but an additional peak is present just outside the slice's boundary.
    A method misjudging the location of such boundary peak might be penalized
    harshly by PeakRMSE even though it made only a small error.
  \item Similarly to the above example, imagine that ground truth has
    two similar subtle peak-like features in a slice, but only
    one of them qualifies as a peak according to the thresholds we set.
    A method which successfully predicts both of them, but only allowed
    to output one, might misjudge the subtlety of the two locations
    slightly and output the wrong one. Similarly to the above,
    such small error might be penalized harshly by PeakRMSE if those
    locations are far apart.
\end{itemize}
We address these issues by introducing the notion of ``inner peak''
and ``outer peak'' in the ground truth data. Inner peaks are the peaks
passing the thresholds we set and located away from the boundary of the slice,
and outer peaks are
the peak-like features that are close to the boundary of the slice
or only pass a lowered set of thresholds for the definition of a peak.

\subsection{Definition of PeakRMSE}
Consider a method for building a phase-diagram like picture.
It may be solving a FIM-estimation task or some other task.
The procedure to measure PeakRMSE is as follows:
\begin{enumerate}
  \item Split the ground truth phase diagram into 1D slices:
    e.g. 2D phase diagrams are split into all vertical and horizontal slices.
  \item For each slice $s$ determine
    (i) $n_s$ --- the number of all ground truth peaks in the slice
      (both inner and outer);
    (ii) $n'_s$ --- the number of inner ground truth peaks in the slice ($n'_s \leq n_s$);
    (iii) $\{y_{s,j}\}_{j=1}^{n'_s}$ --- the locations of the inner ground truth peaks.
  \item For each slice $s$, give each method the number $n_s$ and ask it to output
    up to $n_s$ guesses for the locations of the peaks in the slice:
    $\{x_{s,j}\}_{j=1}^{n_s}$. Add the boundary points of the slice
    (denoting them $x_{s,0}$ and $x_{s,n_s+1}$) to the list
    of guesses obtaining $\{x_{s,j}\}_{j=0}^{n_s + 1}$.
  \item For each slice and each ground truth peak $y_{s,j}$
    compute the distance to the closest
    predicted peak: $d_{s,j} = \min_{k} \abs{x_{s,k} - y_{s,j}}$.
\end{enumerate}
The PeakRMSE is then defined as
\begin{equation}
  \text{PeakRMSE} = \sqrt{\frac{1}{N} \sum_{s,j} d_{s,j}^2},
\end{equation}
where $N = \sum_s n'_s$ is the total number of inner ground truth peaks in all slices.
In other words, for each slice, we ask each method to produce the number of ``peaks'' equal to the total number of peaks in the ground truth (both inner and outer) but only score the distance for the inner peaks.
Lower PeakRMSE indicates better performance, and a method having access to the ground truth
would achieve PeakRMSE equal to zero.

This abstract definition of PeakRMSE needs the definitions of inner and outer peaks,
which we provide below in \cref{ass:peakrmse-boundary}, and the modifications
to the tested method so that it returns the guesses $\{x_{s,j}\}_{j=1}^{n_s}$
attempting to minimize the PeakRMSE metric, which are performed in
\cite[Sec 5 and App. F]{ClassiFIM-ML}: thus modified versions of W, SPCA, and ClassiFIM
are called mod-W, mod-SPCA, and mod-ClassiFIM.

\subsection{Determining the phase boundary of the exact FIM phase diagram (peak locations)}
\label{ass:peakrmse-boundary}
For the datasets we generate the ground truth is given in the form of the FIM phase diagram.
The phase boundary of the exact FIM phase diagram is defined by peaks in the FIM.
In the thermodynamic limit, these become sharp, but for any finite system,
we must determine these peaks by some numerical method. Consider, for example,
a horizontal slice where we sweep $\lambda_0$ from $0$ to $1$ while keeping $\lambda_1$ constant.
Phase transition-like features are then expected to appear
as local maxima of $g_{00}$ as a function of $\lambda_0$.
In order to avoid selecting very shallow maxima of $g_{00}$
we pick a prominence cutoff.
Here, by prominence, we mean the term widely used in topography,
defined as the height of the peak's summit above the lowest contour
line encircling it and no higher summit.
Formally, for a given peak (i.e., a local maximum) $P$ with elevation $h(P)$,
one first identifies the lowest contour line that encloses $P$.
This contour line determines the saddle point for the peak $P$.
Then, the prominence of peak $P$, denoted as $\text{prom}(P)$,
is the difference in elevation between $P$ and the elevation of the saddle point.
If the elevation of the saddle point is $h(C)$, then:
\begin{equation}
   \text{prom}(P) = h(P) - h(C).
\end{equation}
In our case, this definition is applied in the 1D slice, i.e. the contour ``line''
consists of just two points: to the left and to the right of $P$.

The innner peaks are then the local maxima of $g_{00}$ (as a function of $\lambda_0$)
that are not close to the boundary
of the slice and satisfy the prominence cutoff $C$.
The outer peaks are all other local maxima of $g_{00}$
satisfying the prominence cutoff $C/2$ or additional suspected peaks
outside the slice.
Specifically, we we consider the peaks within $6/64$ to be close to the boundary,
and if a neighboring slice had a peak within $3/64$ of the boundary
but the current one does not have one within $6/64$ of the boundary,
we add one outer peak.

The definition for vertical slices is analogous.


%

\end{document}